\documentclass[11pt, a4paper]{article}
\usepackage{jheparxiv}
\usepackage[utf8]{inputenc}
\usepackage{amsmath}
\usepackage{amsfonts}
\usepackage{amssymb}
\usepackage{ntheorem}
\usepackage{latexsym}
\usepackage{mathrsfs}
\usepackage{braket}		
\usepackage{graphicx}
\usepackage{color}
\usepackage{xcolor}
\usepackage{slashed}
\usepackage{twistor}
\usepackage{subcaption}
\usepackage{mathtools}
\usepackage{shuffle}

\newcommand{\sa}{\mathsf{a}}

\newcommand{\sT}{\mathsf{T}}

\newcommand{\ba}{\mathbf{a}}

\newcommand{\m}{\mathrm{m}}

\newcommand{\bolg}{\mathbf{g}}

\renewcommand{\d}{\mathrm{d}}

\newcommand{\qed}{\hfill \ensuremath{\Box}}

\newcommand{\msf}[1]{\mathsf{#1}}

\newcommand{\bigma}{\boldsymbol{\sigma}}

\renewcommand{\det}{\mathrm{det}}
\newcommand{\bolth}{\mathbf{\tilde{h}}}
\newcommand{\bolh}{\mathbf{h}}
\newcommand{\brpt}{\mathrm{pt}}

\newcommand{\fullpt}{\mathrm{PT}}

\newtheorem*{thmno}{Theorem}

\title{The KLT kernel in twistor space}

\author{Tim Adamo}
\author{\& Sonja Klisch}

\affiliation{School of Mathematics and Maxwell Institute for Mathematical Sciences \\
        University of Edinburgh, EH9 3FD, United Kingdom}

\emailAdd{t.adamo@ed.ac.uk}
\emailAdd{s.klisch@ed.ac.uk}

\abstract{The double copy relationship between Yang-Mills theory and general relativity can be stated in terms of a field theory Kawai-Lewellen-Tye (KLT) momentum kernel, which maps two colour-ordered gluon amplitudes to a graviton amplitude at tree-level. These amplitudes can also be written in compact, helicity-graded representations on twistor space which include the famous Parke-Taylor and Hodges formulae in the maximal helicity violating sector. However, a double copy formulation of these helicity-graded formulae has proved elusive. In this paper, we use graph-theoretic methods to obtain an explicit double copy representation of the tree-level, helicity graded S-matrix of general relativity in terms of a KLT-like integral kernel in twistor space. This integral kernel glues together two colour-ordered integrands for tree-level gluon scattering on twistor space to produce tree-level graviton amplitudes, and admits a chiral splitting into positive and negative helicity degrees of freedom. Furthermore, the kernel can be inverted to obtain a new formula for the tree-level S-matrix of biadjoint scalar theory, which we verify using recursion relations. We also derive extensions of this integral kernel to graviton scattering in anti-de Sitter space and self-dual radiative spacetimes, commenting on their potential double copy interpretations.}

\begin{document}

\maketitle

\section{Introduction}

The double copy is a relation between scattering amplitudes in gravitational and non-gravitational quantum field theories (QFTs), which has been the subject of more than 38 years of active study (see~\cite{Bern:2002kj,Bern:2019prr,Borsten:2020bgv,Bern:2022wqg,Adamo:2022dcm} for reviews). The original incarnation, due to Kawai-Lewellen-Tye (KLT)~\cite{Kawai:1985xq}, expresses closed string amplitudes as a product of two open string amplitudes, multiplied via a kinematic object now known as the \emph{KLT} or \emph{momentum kernel}. The field theory limit of this statement gives a powerful relation between tree amplitudes in general relativity and Yang-Mills: 
\begin{equation} \label{intro1}
\mathcal{M}^{\text{GR}}_{n} = \sum_{\alpha, \beta \in S_{n- 3}} \cA^{\text{YM}}_{n} [12\alpha n]\, S_n [\alpha |\beta]\, \cA^{\text{YM}}_{n}[2\beta1n]\,,
\end{equation}
where $\cM^{\text{GR}}_n$ is the $n$-point tree-level graviton amplitude and $\cA^{\text{YM}}_n[12\alpha n]$ is the $n$-point tree-level, colour-ordered partial amplitude for gluon scattering in colour ordering $12\alpha n$. The sum is over the set of orderings $\alpha,\beta\in S_{n-3}$, each of size $(n-3)!$, and $S_{n}[\alpha|\beta]$ is the KLT kernel, a polynomial of degree $n-3$ in 2-particle Mandelstam invariants whose form is known explicitly for arbitrary $n$~\cite{Bjerrum-Bohr:2009ulz,Bjerrum-Bohr:2010diw,Bjerrum-Bohr:2010kyi,Bjerrum-Bohr:2010pnr}. KLT kernels linking a remarkable array of other QFTs and string theories at tree-level have also been constructed~\cite{Chi:2021mio,Chen:2023dcx}.

The double copy has been extended to loop-level via the notion of \emph{colour-kinematics duality}~\cite{Bern:2008qj,Bern:2010ue,Bern:2010yg}, where the complexity of gravity amplitudes is traded for the simplicity of gluon amplitudes whose kinematic numerators are in a suitable representation. At tree-level, the existence of this colour-kinematic representation is in one-to-one correspondence with the KLT kernel~\cite{Bjerrum-Bohr:2009ulz,Stieberger:2009hq,HenryTye:2010tcy,Bjerrum-Bohr:2010mia,Feng:2010my}, and it has been shown to exist (or with some modifications) for 4-point scattering with maximal supersymmetry through 5-loops~\cite{Bern:2017yxu,Bern:2017ucb}. The existence of a corresponding KLT kernel beyond tree-level is less clear, although there has been substantial progress in the study of closely related string monodromy relations at higher genus~\cite{Tourkine:2016bak,Hohenegger:2017kqy,Mafra:2017ioj,Casali:2019ihm,Casali:2020knc,Stieberger:2021daa,Stieberger:2022lss,Stieberger:2023nol,Mazloumi:2024wys}. Moreover, the double copy toolkit has recently been applied to precision-frontier calculations in the gravitational 2-body problem; the most recent example at writing is the calculation of the scattering angle between two supersymmetric black holes to \emph{fifth-order} in the Post-Minkowski expansion and first-order in the self-force expansion of $\cN=8$ supergravity~\cite{Bern:2024adl}.

Yet despite these myriad advances, there are still some surprising gaps in our understanding of double copy. For instance, it is not known how to explicitly double copy some of the most famous scattering amplitude formulae, which are written in a helicity graded representation for the external states. In this paper, we establish a new double copy description of tree-level scattering using helicity graded representations of gluon and graviton amplitudes. 

\medskip

Indeed, the study of scattering amplitudes in four-dimensions has enjoyed a parallel set of advances inspired by the natural organisation of amplitudes into helicity sectors (also known as $R$-charge sectors in supersymmetric theories). The simplest non-vanishing tree-level amplitudes in gauge theory and gravity are those that are \emph{maximally helicity violating} (MHV), with two negative helicity and arbitrarily many positive helicity particles\footnote{The conjugate configuration with two positive and arbitrarily many negative helicity particles is known as $\overline{\text{MHV}}$ configuration, and is equal to the complex conjugate of the MHV amplitude for Lorentzian-real kinematics.}. It is a remarkable fact that these amplitudes in gauge theory and gravity admit compact, closed-form expressions at arbitrary multiplicity in the external positive helicity particles, known as the Parke-Taylor~\cite{Parke:1986gb} and Hodges formulae~\cite{Hodges:2012ym}\footnote{Other expressions for the graviton MHV amplitude were known before~\cite{Berends:1988zp,Bern:1998sv,Mason:2009afn,Nguyen:2009jk}, but these lack the manifest permutation invariance in the external gravitons which is manifest in Hodges' formula.}, for MHV gluon and graviton scattering, respectively. 

The incredible simplicity of scattering amplitudes in the MHV sector can be explained by the fact that both Yang-Mills and general relativity admit perturbative expansions around the self-dual sector in asymptotically flat spacetimes~\cite{Plebanski:1977zz,Chalmers:1996rq,Mason:2005zm,Sharma:2021pkl}. The MHV sectors represent the first non-trivial term in this perturbative expansion around a classically integrable subsector. One can then ask if a similar level of simplicity is lurking at all-orders in the expansion around self-duality at tree-level.

Answering this question involves using \emph{twistor theory}~\cite{Penrose:1967wn}, an algebro-geometric formalism which manifests the integrability of the self-dual sector in gauge theory and gravity~\cite{Ward:1977ta,Penrose:1976js,Mason:1991rf}. Indeed, it was noted long ago that the Parke-Taylor formula can be written in terms of a certain worldsheet correlator in twistor space~\cite{Nair:1988bq}, and Witten extended this idea to the full tree-level S-matrix of Yang-Mills theory~\cite{Witten:2003nn}. Here, the tree-level N$^k$MHV amplitude, with $k+2$ negative helicity and arbitrarily many positive helicity external gluons, is given by an integral over the moduli space of holomorphic, degree $k+1$ maps from the Riemann sphere to twistor space, which is localised by kinematic constraints. This can be written in a compact form for arbitrary multiplicity, and is known as the Roiban-Spradlin-Volovich-Witten (RSVW) formula~\cite{Roiban:2004yf} for the helicity-graded tree-level S-matrix of Yang-Mills theory. There is also a helicity-graded formula for the tree-level S-matrix of general relativity, again written as a localised moduli integral over holomorphic rational maps to twistor space, known as the Cachazo-Skinner formula~\cite{Cachazo:2012kg}.

\medskip

While studies of helicity graded all-multiplicity amplitudes and the double copy were initially closely entangled\footnote{Indeed, the first investigations of gravitational MHV scattering were closely tied to a double copy of the Parke-Taylor formula~\cite{Berends:1988zp}.}, the two paths have since diverged significantly. In particular, there is no obvious imprint of double copy visible in the Hodges formula for gravitational MHV scattering, or, more generally, in the Cachazo-Skinner formula for the helicity graded tree-level graviton S-matrix. Of course, on a basic level we know that a double copy relationship between the RSVW formula and the Cachazo-Skinner formula must exist: the KLT double copy \eqref{intro1} preserves the helicity grading. Yet an explicit double copy construction of the Cachazo-Skinner representation of the the graviton amplitudes has remained elusive\footnote{A formula for N$^k$MHV graviton scattering was developed by Cachazo-Geyer~\cite{Cachazo:2012da} building on a sort of double copy from the RSVW formula, but this formula remains a conjecture and has not been shown to be equal to the Cachazo-Skinner formula.}. Some recent progress was made in determining colour-kinematic numerators for the MHV sector directly from the Hodges formula~\cite{Frost:2021qju}.

In some sense, this opacity is compensated by the clarity of the double copy in the Cachazo-He-Yuan (CHY) formalism~\cite{Cachazo:2013gna,Cachazo:2013hca,Cachazo:2013iea}, where tree-level amplitudes are calculated by localising integrands on the $(n - 3)!$ solution of the scattering equations, a set of contraints linking kinematics with the moduli of a punctured Riemann sphere. Here, double copy is beautifully apparent at the level of the integrands; this representation of double copy led to the discovery that the matrix inverse of the KLT momentum kernel is equal to the amplitudes of another quantum field theory - bi-adjoint scalar (BAS) theory. Mathematically, this is an extremely non-trivial fact, as \emph{a priori} this requires the calculation of the inverse of an $(n-3)!\times (n-3)!$ matrix.  However, in the CHY formalism, the momentum kernel is a $1 \times 1$ matrix, with the remaining complexity absorbed by the scattering equations. Since then, the relation between the momentum kernel and BAS theory has been studied independently of CHY~\cite{Mizera:2016jhj,Carrasco:2016ldy,Mafra:2016mcc,Mizera:2017cqs,Arkani-Hamed:2017mur,Frost:2018djd,Chi:2021mio,Chen:2023dcx}, and proven in~\cite{Frost:2020eoa}.

The scattering equations and the degree $d$ maps that underlie N$^{d-1}$MHV amplitudes are closely connected. In fact, the moduli integrals in the RSVW or Cachazo-Skinner formulae are localised onto a subset of size $E(n - 3, d -1)$ (the Eulerian number) of the $(n-3)!$ solutions to the scattering equations~\cite{Spradlin:2009qr,Cachazo:2013iaa,Roehrig:2017wvh}. It is thus possible to `grade' the CHY formulae by helicity to obtain the corresponding N$^{d - 1}$MHV amplitudes in gauge theory and gravity. However, this involves a rather non-trivial basis change at the level of the moduli integrals~\cite{Litsey:2013jfa,Du:2016blz,He:2016vfi,Zhang:2016rzb,Geyer:2016nsh} which obscures the initially obvious double copy structure as well as any link with BAS amplitudes\footnote{In~\cite{Cachazo:2016sdc}, a formula for BAS scattering based on rational maps to twistor space was given, but it involves a sum over different degrees, meaning that it cannot be the inverse of a KLT kernel which preserves the helicity grading.}. 

\medskip

In this paper we provide the link that connects these two organising principles of amplitudes: the double copy and helicity grading. Our main results can be summarised as follows.
\begin{itemize}
    \item[] \emph{Theorem~\ref{thm:PT-KLT}:} The integrand of the Cachazo-Skinner formula $\mathcal{I}^{\text{GR}}_{n,d}$ is the product of two colour-ordered RSVW integrands $\mathcal{I}^{\text{YM}}_{n,d}[\alpha]$ multiplied via a \emph{helicity-graded integral kernel} $S_{n, d} [\alpha |\beta]$:
\begin{equation} \label{intro2}
\mathcal{I}^{\text{GR}}_{n, d} = \sum_{\substack{\rho, \omega \in \Omega_d \\ \bar{\rho}, \bar{\omega} \in \Omega_{n - d - 2}}} \mathcal{I}^{\text{YM}}_{n, d}[\rho \bar{\rho}]\, S_{n, d}[\rho \bar{\rho} |\omega^{\mathrm{T}} \bar{\omega}] \,\mathcal{I}^{\text{YM}}_{n, d}[\omega^{\mathrm{T}} \bar{\omega}]\,,
\end{equation}
where $\Omega_d, \Omega_{n - d - 2}$ are certain subsets of colour orderings with $d\,!$ and $(n - d- 2)!$ elements, respectively. This kernel $S_{n,d}$, for which we give an explicit formula, is an object at the level of \emph{integrands} on twistor space and admits a chiral splitting according to the helicity configurations of the external particles, as does the basis of colour orderings. 
 
    \item[] \emph{Theorem~\ref{BASthm}:}  The inverse of the helicity-graded integral kernel (viewed as a linear map on the bases of colour orderings) defines a new formula for all tree-level, colour-ordered partial amplitudes of BAS theory. Schematically,
    \begin{equation}
        m_{n}[\rho \bar{\rho} | \omega^{\mathrm{T}} \bar{\omega}^{\mathrm{T}}] = \int \d \mu_d \,S^{-1}_{n, d}[\rho \bar{\rho}|\omega^{\mathrm{T}} \bar{\omega}] \,\prod_{i = 1}^n \varphi_i\,,
    \end{equation}
    where the integral is over the moduli space of maps of degree $d$ from the Riemann sphere to twistor space, and external kinematics are encoded in the twistor wavefunctions $\{\varphi_i\}$. Note that this expression for the partial amplitude $m_{n}[\rho\bar{\rho}|\omega^{\mathrm{T}}\bar{\omega}]$ depends on the degree of the map only through the form of the colour-orderings.
\end{itemize}

The derivation of the integral kernel underpinning Theorem~\ref{BASthm} rests on a non-trivial re-writing of the integrand of the Cachazo-Skinner formula using various tools from graph theory. The integral kernel $S_{n,d}$ should be viewed as the KLT kernel in twistor space; indeed, in the MHV sector ($d=1$) where all moduli integrals can be performed explicitly, the resulting momentum space expression is the well-known KLT momentum kernel. The proof of the new formula for BAS tree-amplitudes in Theorem~\ref{BASthm} is by iteration, using BCFW recursion adapted to twistor space. The result is remarkable, as the final expression is a scalar amplitude, independent of helicity.

The graph-theoretic methods utilized here can naturally be adapted to any twistor space formula which has a similar mathematical structure to the original Cachazo-Skinner formula. This allows us to find generalizations of the integral kernel on certain \emph{curved} four-dimensional spacetimes; namely, anti-de Sitter space~\cite{Adamo:2015ina} and self-dual radiative spacetimes~\cite{Adamo:2022mev}. In both cases, this leads to tantalizing hints at the possibility of an explicit manifestation of double copy for scattering in curved spacetimes, although we observe several issues that prevent us from drawing any definitive conclusions in this regard.

\medskip

The paper is structured as follows. Section \ref{sec:rev} is a review of the tree-level S-matrices in gauge theory and gravity, as well as results in graph theory that will be used throughout the paper. Section \ref{sec:deriv} contains the derivation of the helicity-graded integral kernel. The derivation rests on the graph theoretic properties of the gravity integrand, and results in a chirally split momentum kernel. The special case of the MHV configuration is highlighted and reproduces the usual field theory KLT momentum kernel. In Section \ref{sec:proof} the twistor space BAS formula is presented and justified. It is then proven using BCFW recursion techniques. Section \ref{sec:bg} applies the methods of the previous two sections to amplitude formulae on non-trivial backgrounds: AdS and self-dual radiative gauge fields and spacetimes. Finally, Section \ref{sec:conc} concludes, while Appendices \ref{parityApp} and \ref{softApp} verify expected physical properties of the BAS formula.


\section{Review: tree-level S-matrices \& graph theory} \label{sec:rev}

Our derivation of a helicity-graded momentum kernel for the tree-level S-matrix of gravity, and its relationship with the double copy of the tree-level S-matrix of Yang-Mills theory, makes use of formulae for all tree-level gluon and graviton scattering amplitudes in a helicity-graded representation. These are the Roiban-Spradlin-Volovich-Witten (RSVW) formula~\cite{Witten:2003nn,Roiban:2004yf} for gluons, and the Cachazo-Skinner formula~\cite{Cachazo:2012kg}, for gravitons. Both are written in terms of moduli integrals over the space of holomorphic maps from the Riemann sphere to twistor space, with the degree of the map fixed by the helicity configuration of the scattering amplitude and all moduli integrals localized in terms of the kinematic data. In addition, the foundation of our results lies in several (fairly standard) results from algebraic combinatorics and graph theory applied to these formulae.

While these formulae and results have been in the literature and textbooks for many years, they are not necessarily well-known to a general mathematical physics audience. In this section we review the RSVW and Cachazo-Skinner formulae, providing a brief summary of their relevant features, as well as cataloging concepts and results in graph theory which will be relevant for the main results of the paper. The reader who is already familiar with these topics may wish to simply skim this section, familiarizing themselves with our notation and terminology.


\subsection{Tree-level scattering amplitudes of gauge theory and gravity}

In four spacetime dimensions, gluons and gravitons have only two on-shell polarizations which can be labeled by \emph{helicity}, which is simply a sign corresponding to whether the linearised field strength of the gluon or graviton is self-dual (positive helicity) or anti-self-dual (negative helicity). This allows the tree-level S-matrices of Yang-Mills and general relativity in 4-dimensions to be graded by the number of negative helicity external particles in the scattering process (where all external particles are assumed to be outgoing). Consistency and classical integrability of the self-dual sectors in these theories ensure that the tree-level amplitudes with less than two negative helicity particles vanish, so the first non-trivial amplitude in this helicity grading, known as the \emph{maximal helicity violating} (MHV) configuration, involves two negative helicity and arbitrarily many positive helicity external particles. 

In general, the $n$-point, tree-level scattering amplitude with $k+2$ negative helicity external particles is called the tree-level N$^k$MHV amplitude, denoted by $\cM_{n,k+1}$. For gluon scattering, it is often useful to further decompose the tree-level scattering amplitude into colour-ordered partial amplitudes. It is easy to show that all tree-amplitudes in Yang-Mills theory (with gauge group SU$(N)$, for instance) can be decomposed as
\be\label{cdecomp}
\cM_{n,k+1}=\sum_{\rho\in S_n\setminus\Z_n}\tr\left(\sT^{\sa_{\rho(1)}}\cdots\sT^{\sa_{\rho(n)}}\right)\,\cA_{n,k+1}[\rho]\,,
\ee
where the sum is over non-cyclic permutations on $n$ labels and $\sT^{\sa}$ are generators of the adjoint representation of the gauge group. All kinematic information is packaged in the \emph{colour-ordered partial amplitude} $\cA_{n,k+1}[\rho]$; knowing this partial amplitude, with the explicit colour-dependence stripped off, is equivalent to knowing the full tree-level amplitude.

It is a remarkable fact that MHV amplitudes in gauge theory and gravity have simple, compact expressions for arbitrary multiplicity when written in the appropriate variables. In 4-dimensions, a (complex) null momentum $k^{\mu}$ can be written as a simple $2\times2$ matrix $k^{\alpha\dot\alpha}=\kappa^\alpha\,\tilde{\kappa}^{\dot\alpha}$, where $\alpha=0,1$ and $\dot\alpha=\dot{0},\dot{1}$ are negative and positive chirality $\SL(2,\C)$ Weyl spinor indices, respectively. The little group acts as the rescaling $\kappa\to r\,\kappa$, $\tilde{\kappa}\to r^{-1}\,\tilde{\kappa}$ for $r$ any non-vanishing complex number. Spinor indices are raised and lowered using the $\SL(2,\C)$-invariant Levi-Civita symbols, and we adopt the conventions:
\be\label{shcon}
\kappa^{\alpha}=\epsilon^{\alpha\beta}\,\kappa_{\beta}\,, \qquad \kappa_{\alpha}=\kappa^{\beta}\,\epsilon_{\beta\alpha}\,,
\ee
and similarly for dotted indices. These are easily seen to be consistent with the normalization $\epsilon^{\alpha\beta}\,\epsilon_{\gamma\beta}=\delta^{\alpha}_{\gamma}$. A tree-level scattering amplitude involving $n$ gluons or gravitons will be a rational function of the Lorentz-invariant kinematical quantities
\be\label{genMan}
\la i\,j\ra:=\epsilon_{\beta\alpha}\,\kappa^{\alpha}_i\,\kappa^{\beta}_j\,, \qquad [i\,j]:=\epsilon_{\dot\beta\dot\alpha}\,\tilde{\kappa}^{\dot\alpha}_i\,\tilde{\kappa}^{\dot\beta}_j\,,
\ee
for $i,j=1,\ldots,n$. 

In these variables, tree-level MHV gluon and graviton scattering amplitudes take an incredibly simple form. The famous Parke-Taylor formula~\cite{Parke:1986gb} for MHV gluon scattering is given by
\be\label{ParkeTaylor}
\cA_{n,1}[12\cdots n]=\delta^{4}\!\left(\sum_{i=1}^n k_i\right)\,\frac{\la a\,b\ra^{4}}{\la 1\,2\ra\,\la2\,3\ra\cdots\la n\,1\ra}\,,
\ee
where gluons $a,b$ have negative helicity. This formula manifests the cyclicity associated with the colour-ordering (broken only by the choice of the two negative helicity gluons), and scales under the little group for each gluon with homogeneity appropriate to that gluon's helicity. 

The Hodges formulae~\cite{Hodges:2012ym} for tree-level MHV graviton scattering is the corresponding expression\footnote{There are several earlier explicit formulae for the gravitational MHV amplitude~\cite{Berends:1988zp,Bern:1998sv,Nguyen:2009jk,Mason:2009afn}, but none of them explicitly manifests the permutation invariance of the external gravitons or is as compact as the Hodges formula.} in general relativity:
\be\label{HodgesForm}
\cM_{n,1}=\delta^{4}\!\left(\sum_{i=1}^n k_i\right)\,\la1\,2\ra^8\,\mathrm{det}'(\mathrm{H})\,,
\ee
where gravitons 1,2 have negative helicity. Here, $\mathrm{H}$ is a $(n-2)\times(n-2)$ matrix whose entries are
\be\label{momHM}
\mathrm{H}_{ij}=\left\{\begin{array}{l}
                       \frac{[i\,j]}{\la i\,j\ra} \quad \mathrm{if}\: i\neq j \\
                       -\sum_{j\neq i}\frac{[i\,j]}{\la i\,j\ra}\,\frac{\la1\,j\ra\,\la2\,j\ra}{\la1\,i\ra\,\la2\,i\ra} \quad \mathrm{if}\: i=j
                       \end{array}\right.\,, \quad i,j\in\{3,\ldots,n\}\,,
\ee
while
\be\label{momHM2}
\mathrm{det}'(\mathrm{H}):=\frac{|\mathrm{H}^{i}_{i}|}{\la1\,2\ra^2\,\la1\,i\ra^2\,\la2\,i\ra^2}\,,
\ee
is a reduced determinant which is independent of the choice of positive helicity graviton $i$ on the support of momentum conservation.

\medskip

It is natural to ask if there are similarly compact expressions for gluon and graviton scattering beyond the MHV sector. This is indeed possible, by writing the S-matrix in terms of a localized integral over the moduli space of rational, holomorphic maps from the Riemann sphere to twistor space. The degree $d$ of the underlying map is related to the N$^k$MHV degree of the corresponding scattering amplitude by $d=k+1$. 

The twistor space of (complexified) Minkowski space is given by an open subset of 3-dimensional complex projective space (we follow the notation and conventions of~\cite{Adamo:2017qyl}). Let $Z^A=(\mu^{\dot\alpha},\lambda_{\alpha})$ be holomorphic, homogeneous coordinates on $\P^3$. Twistor space is then the open subset
\be\label{PT}
\PT=\left\{Z^A\in\P^3\,|\,\lambda_{\alpha}\neq0\right\}\,.
\ee
Let $\sigma^{\ba}=(\sigma^{\mathbf{0}},\sigma^{\mathbf{1}})$ be holomorphic homogeneous coordinates on $\P^1$, the Riemann sphere. A holomorphic map of degree $d$ from $\P^1$ to $\PT$ can be parameterized as
\be\label{degd}
Z^{A}(\sigma)=U^{A}_{\ba_1\cdots\ba_d}\,\sigma^{\ba_1}\cdots\sigma^{\ba_d}\equiv U^{A}_{\ba(d)}\,\sigma^{\ba(d)}\,,
\ee
where the $4(d+1)$ components of $U^{A}_{\alpha(d)}$ are the moduli of the map. This is a redundant parameterization, due to the $\SL(2,\C)$ automorphisms of $\P^1$ and $\C^*$ projective rescalings of the homogeneous coordinates, so the integration measure on the moduli space of such maps is given by
\be\label{modmeasure}
\d\mu_{d}:=\frac{\d^{4(d+1)}U}{\mathrm{vol}\,\GL(2,\C)}\,.
\ee
Here, the quotient by the (infinite) volume of $\GL(2,\C)\cong\SL(2,\C)\times\C^*$ is understood in the Faddeev-Popov sense. 

The extension of the Parke-Taylor formula to the full tree-level S-matrix of Yang-Mills theory is known as the Roiban-Spradlin-Volovich-Witten (RSVW) formula~\cite{Witten:2003nn,Roiban:2004yf}. Let $\tilde{\bolg}\subset\{1,\ldots,n\}$ denote the set of $d+1$ negative helicity gluons in the N$^{d-1}$MHV amplitude, and $\bolg=\{1,\ldots,n\}\setminus\tilde{\bolg}$ be its complement. Denote the M\"obius-invariant inner product between homogeneous coordinates on $\P^1$ as
\be\label{Mobip}
(i\,j):=\epsilon_{\mathbf{b}\ba}\,\sigma_i^{\ba}\,\sigma_j^{\mathbf{b}}\,.
\ee
Define the Vandermonde determinant of the set $\tilde{\bolg}$ on $\P^1$ by
\be\label{VDMdet}
|\tilde{\bolg}|:=\prod_{\substack{i,j\in\tilde{\bolg} \\ i<j}}(i\,j)\,.
\ee
This object is homogeneous of degree $d$ in each $\sigma_i$ for $i\in\tilde{\bolg}$. Finally, define the quantity
\be\label{wsPT}
\fullpt_n[\rho]:=\prod_{i=1}^{n}\frac{\D\sigma_{\rho(i)}}{\left(\rho(i)\,\rho(i+1)\right)}\,,
\ee
where $\rho\in S_n\setminus\Z_n$,
\be\label{Dsigma}
\D\sigma:=(\sigma\,\d\sigma)\,,
\ee
is the holomorphic 1-form trivializing the canonial bundle of $\P^1$ and the product in \eqref{wsPT} is understood cyclically.

With these ingredients, the RSVW formula for the tree-level N$^{d-1}$MHV gluon amplitude is~\cite{Witten:2003nn,Roiban:2004yf}
\be\label{RSVW}
\cA_{n,d}[\rho]=\int \d\mu_d\,|\tilde{\bolg}|^4\,\fullpt_n[\rho]\,\prod_{i\in\bolg}a_i(Z(\sigma_i))\,\prod_{j\in\tilde{\bolg}}b_j(Z(\sigma_j))\,.
\ee
Here the $\{a_i(Z)\}$ and $\{b_j(Z)\}$ are, respectively, the positive and negative helicity gluon wavefunctions on twistor space. By the Penrose transform~\cite{Penrose:1969ae,Eastwood:1981jy}, which equates solutions of the zero-rest-mass equations on spacetime to cohomology on twistor space, these wavefunctions are valued in
\be\label{PenTYM}
a_i(Z)\in H^{0,1}(\PT,\cO)\,, \qquad b_j(Z)\in H^{0,1}(\PT,\cO(-4))\,,
\ee
where $H^{0,1}$ is the Dolbeault cohomology group and $\cO(k)\to\PT$ denotes the sheaf of locally holomorphic functions of homogeneity $k\in\Z$ on $\PT$; by convention, we denote $\cO(0)\equiv\cO$. 

At first glance, it may seem inconceivable that an algebro-geometric integral formula like \eqref{RSVW} can truly capture all tree-level gluon scattering amplitudes in Yang-Mills theory. However, by choosing the twistor wavefunctions to represent momentum eigenstates, one can show that in fact \emph{all} of the integrals in \eqref{RSVW} are localized in terms of the kinematics, leaving a rational function of those kinematics with support dictated by overall 4-momentum conservation. Indeed, the twistor wavefunctions for momentum eigenstates take the form~\cite{Roiban:2004yf,Witten:2004cp,Adamo:2011pv}:
\be\label{glmomeig}
\begin{split}
a_i(Z)&=\int_{\C^*}\frac{\d t_i}{t_i}\,\bar{\delta}^{2}(\kappa_i-t_i\,\lambda)\,\e^{\im t_i\,[\mu\,i]}\,, \\
b_j(Z)&=\int_{\C^*}t_j^3\,\d t_j\,\bar{\delta}^{2}(\kappa_j-t_j\,\lambda)\,\e^{\im t_j\,[\mu\,j]}\,,
\end{split}
\ee
where the holomorphic delta function in $N$-dimensions is defined, for some holomorphic quantity $z\in\C^N$, by
\be\label{holdel}
\bar{\delta}^{N}(z):=\frac{1}{(2\pi\im)^N}\,\bigwedge_{a=1}^{N}\dbar\left(\frac{1}{z_a}\right)\,,
\ee
and the integral over the scaling parameter $t$ ensures that each wavefunction has the appropriate homogeneity on twistor space.

Inserting these wavefunctions into \eqref{RSVW}, all moduli integrals for the $\mu^{\dot\alpha}(\sigma)$ part of the holomorphic map $Z:\P^1\to\PT$ can be done against the exponentials to give delta functions. Combined with the holomorphic delta functions in each wavefunction, this gives $2(d+1)+2n$ total constraints, with $2(d+1)+2n-4$ remaining integrals (accounting for the $\GL(2,\C)$ redundancy in the moduli measure). Thus, each remaining integral is localized by a constraint, leaving 4 additional delta functions; it is straightforward to show that these correspond precisely to momentum conservation. Furthermore, the integrand of \eqref{RSVW} after performing the $\mu^{\dot\alpha}(\sigma)$-moduli integrals is a rational function of the remaining moduli as well as the kinematic data. The constraints are similarly rational functions, implying that the result, after all integrals have been localized against the constraints, is itself a rational function of the kinematic data, as required for tree-level scattering amplitudes.

More generally, it can be \emph{proved} that the RSVW formula is a correct representation of the tree-level S-matrix of Yang-Mills theory using tree-level unitarity arguments based on the factorization properties of the formula~\cite{Vergu:2006np,Skinner:2010cz,Dolan:2011za,Adamo:2013tca}. When $d=1$, the localization of the moduli integrals can be done explicitly and the Parke-Taylor amplitude \eqref{ParkeTaylor} is recovered for the MHV sector.

\medskip

The gravitational version of this story, the Cachazo-Skinner formula, follows similar lines, although requiring several more involved ingredients. External gravitons in a scattering process are represented on twistor space via the Penrose transform as twistor wavefunctions
\be\label{PenTGR}
h(Z)\in H^{0,1}(\PT,\cO(2))\,, \qquad \tilde{h}(Z)\in H^{0,1}(\PT,\cO(-6))\,,
\ee
for positive and negative helicity, respectively. The momentum eigenstate representatives for these wavefunctions are
\be\label{grmomeig}
\begin{split}
h_i(Z)&=\int_{\C^*}\frac{\d t_i}{t^3_i}\,\bar{\delta}^{2}(\kappa_i-t_i\,\lambda)\,\e^{\im t_i\,[\mu\,i]}\,, \\
\tilde{h}_j(Z)&=\int_{\C^*}t_j^5\,\d t_j\,\bar{\delta}^{2}(\kappa_j-t_j\,\lambda)\,\e^{\im t_j\,[\mu\,j]}\,.
\end{split}
\ee
Let $\bolth\subset\{1,\ldots,n\}$ be the set of $d+1$ negative helicity gravitons in a N$^{d-1}$MHV scattering process, and let $\bolh$ be the complementary set of positive helicity gravitons. 

Define a $(n-d-1)\times(n-d-1)$ matrix $\HH$ with entries\footnote{In general, the diagonal entries of $\HH$ are defined in terms of an arbitrary choice of section $\mathfrak{s}(\sigma)$ of $\cO(d+1)\to\P^1$~\cite{Cachazo:2012kg,Cachazo:2012pz}. Here, we make the simple choice $\mathfrak{s}(\sigma)=\prod_{l\in\bolth}(\sigma\,l)$; while the amplitude itself is invariant for different choices, our later calculations of the integral kernel will change. This is not surprising, since double copy representations (such as colour-kinematic numerators) are generically gauge-dependent~\cite{Bern:2008qj}.}
\be\label{HMat}
\begin{split}
\HH_{ij}&=-t_i\,t_j\,\frac{[i\,j]}{(i\,j)}\,\sqrt{\D\sigma_i\,\D\sigma_j}\,, \qquad i,j\in\bolh\,, i\neq j\,, \\
\HH_{ii}&=t_i\,\D\sigma_i\,\sum_{\substack{j\in\bolh \\ j\neq i}}t_j\,\frac{[i\,j]}{(i\,j)}\,\prod_{l\in\bolth}\frac{(j\,l)}{(i\,l)}\,, \quad i\in\bolh\,.
\end{split}
\ee
It can be shown that this matrix has corank 1, so it has a natural (generally non-vanishing) associated reduced determinant
\be\label{redHdet}
\mathrm{det}'(\HH):=\frac{|\HH^{i}_{i}|}{\left|\bolth\cup\{i\}\right|^2}\,\prod_{j\in\bolth\cup\{i\}}\D\sigma_i\,,
\ee
for any $i\in\bolh$. It can be shown that $\mathrm{det}'(\HH)$ is independent of the choice of this $i$. 

Next, define the `dual' $(d+1)\times(d+1)$ matrix $\HH^{\vee}$
\be\label{dHmat}
\begin{split}
\HH^{\vee}_{ij}&=\frac{\la\lambda(\sigma_i)\,\lambda(\sigma_j)\ra}{(i\,j)}\,, \qquad i,j\in\bolth\,, i\neq j \\
\HH^{\vee}_{ii}&=-\frac{\la\lambda(\sigma_i)\,\d\lambda(\sigma_i)\ra}{\D\sigma_i}\,, \qquad i\in\bolth\,.
\end{split}
\ee
The matrix $\HH^{\vee}$ also has corank 1, with the natural (generally non-vanishing) reduced determinant
\be\label{reddHdet}
\mathrm{det}'\!\left(\HH^\vee\right):=\frac{|\HH^{\vee\,i}_{i}|}{\left|\bolth\setminus\{i\}\right|^2}\,,
\ee
for any $i\in\bolth$. It can be shown that not only is $\mathrm{det}'(\HH^\vee)$ independent of this choice of $i\in\bolth$, but it is actually independent of \emph{all} the marked points $\{\sigma_1,\ldots,\sigma_n\}$~\cite{Skinner:2013xp}. In particular, this reduced determinant is actually equal to the \emph{resultant} (cf., \cite{Gelfand:2008}) of the $\lambda_{\alpha}(\sigma)$ components of the holomorphic map $Z:\P^1\to\PT$~\cite{Cachazo:2013zc}. This gives $\mathrm{det}'(\HH^\vee)$ an algebro-geometric meaning in its own right; among other properties, it ensures that $\mathrm{det}'(\HH^\vee)$ vanishes whenever $\lambda_{\alpha}(\sigma)=0$ -- that is, whenever the holomorphic map $Z$ lands outside of twistor space.

With these ingredients in place, the Cachazo-Skinner formula for the tree-level scattering amplitudes of general relativity is:
\be\label{CSform}
\cM_{n,d}=\int\d\mu_d\,|\bolth|^8\,\mathrm{det}'(\HH)\,\mathrm{det}'\!\left(\HH^{\vee}\right)\,\prod_{i\in\bolh}h_i(Z(\sigma_i))\,\prod_{j\in\bolth}\tilde{h}_j(Z(\sigma_j))\,.
\ee
As with the RSVW formula, when evaluated on the momentum eigenstates \eqref{grmomeig} all of the integrals in this expression are localized against delta function constraints, with four remaining delta functions imposing momentum conservation. The formula is a rational function of the kinematic data and can be shown to exhibit the correct factorization properties~\cite{Cachazo:2012pz,Adamo:2013tca}, thereby establishing that it is indeed a representation of the tree-level S-matrix of general relativity. Furthermore, when $d=1$ it is straightforward to show that the Cachazo-Skinner formula reduces to the Hodges formula \eqref{HodgesForm} for the graviton MHV amplitude.

\medskip

It should be noted that both the RSVW and Cachazo-Skinner formulae can be obtained as correlation functions in \emph{twistor string theories}: chiral 2d CFTs governing holomorphic maps from a closed Riemann surface to twistor space. This connection was first suggested long ago by Nair~\cite{Nair:1988bq}, who observed that the Parke-Taylor formula could arise in this way, and generalized to the full tree-level S-matrix of Yang-Mills theory by Witten~\cite{Witten:2003nn} and others~\cite{Berkovits:2004hg,Berkovits:2004jj,Mason:2007zv}. A twistor string theory for gravity was later discovered by Skinner~\cite{Skinner:2013xp}. For the considerations in the remainder of this paper, we will however not need to exploit any underlying twistor string description of the amplitudes.


\subsection{Relevant concepts in graph theory}
\label{graphTheory}

Many of the key arguments in this paper make use of results in graph theory. Most of these are standard results in algebraic combinatorics, but for the sake of completeness (and those readers unfamiliar with graph theory), we include a review of the necessary concepts here. Standard textbook references are~\cite{Stanley:1999,vanLint:2001,Stanley:2013}, with concise introductions to some of the relevant theorems in the amplitudes literature to be found in~\cite{Feng:2012sy,Adamo:2012xe,Frost:2021qju}.

\medskip

A \emph{graph} $G = (V, E)$ is a set $V$ of vertices, and a set $E$ of edges. Each edge $e \in E$ is a pair of two vertices $(v - w)$ if the graph is \emph{undirected}, or $(v \rightarrow w)$ if the graph is \emph{directed}. The graph is \emph{simple} if there is no more than one edge connecting any two vertices, and no edge connects the same vertex. A \emph{tree graph} is a connected simple graph without cycles. A \emph{spanning tree} of a connected graph $G$ is a connected tree on the same vertices and with the set of edges a subset of the set of edges of $G$.

Given an undirected tree graph $G$ on the vertex set $V = \{1, 2, \ldots, n\}$, we can fix some $b \in V$ to be the `root' (also sometimes called a `sink') of $G$. Then one can uniquely assign a direction to the edges of $G$ so that each vertex $i \neq b$ has exactly one outgoing edge, and $b$ only has incoming edges. Thus for a given vertex $i \neq k$ in $G$, we may define $o(i)\in V$ to be the vertex connected to $i$ via the outgoing edge. Given an undirected tree graph $G$ and a choice of root $b\in V$, we denote the directed tree graph arising in this way as $G^b$.

\begin{figure}
    \centering
    \includegraphics{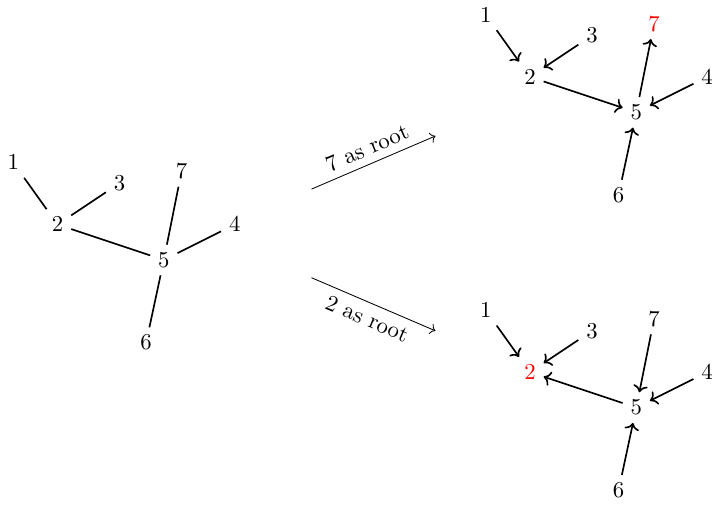}
    \caption{Labelling a point as a `root' gives a unique orientation to an undirected tree graph.}
    \label{fig:enter-label}
\end{figure}

From now on, we will consider graphs inscribed on the Riemann sphere $\P^1$, so that to each vertex $i\in V$ there is an associated point on $\P^1$, written in homogeneous coordinates as $\sigma^{\ba}_{i}$. 

Now, directed tree graphs can be related to orderings on the set of vertices. An \emph{ordering} of $V$ is a word that uses each letter $v \in V$ exactly once. For an ordering $\rho$ on $V$ we define the \emph{broken Parke-Taylor} factor $\mathrm{pt}_{n}[\rho]$ in the ring of rational functions on $(\P^1)^{n}$ as the product
\begin{equation}\label{brokePT}
\mathrm{pt}_{n}[\rho]=\prod_{i\in V\setminus\{\rho(n)\}}\frac{1}{(\rho(i)\,\rho(i+1))}\,,
\end{equation}
where $\rho(i+1)$ is the letter following $\rho(i)$ in the word, and $\rho(n)$ is the last element of $\rho$. For example,
\be\label{bPTex}
\mathrm{pt}_{n}[12\cdots n]=\frac{1}{(1\,2)\,(2\,3)\cdots (n-1\,n)}\,,
\ee
and in general these broken Parke-Taylor factors are related to their `full' counterparts \eqref{wsPT} by $\mathrm{pt}_n[\rho] = \mathrm{PT}_{n}[\rho] (\rho(n)\,\rho(1))$. 

It is a simple corollary of a result in~\cite{Frost:2021qju}\footnote{The proposition in \cite{Frost:2021qju} is the case where we fix $x = (1, 0)$, and all other points are of the form $\sigma_i^{\ba}=(z_i, 1)$. Similar relations also arise in the study of arrangements of hyperplanes in complex affine spaces~\cite{Schechtman:1991}.} that:
\begin{proposition}\label{PTprop}
Let $G^b$ be a directed tree graph on $\P^1$ with vertices $\{1 \, \ldots n\}$ and root $b$, and let $x\in\mathbb{P}^1$ be a non-intersecting point. Then:
\begin{equation}\label{PTpropident}
    \prod_{\substack{\mathrm{edges}\\ i \rightarrow j}} \frac{(j\,x)}{(i\,j)(i\,x)} = \sum_{\substack{\rho \\ o(i) <_{b\rho} i}} \mathrm{pt}_{n}[b\rho]\, \frac{(b\,x)}{(\rho(n)\,x)}.
\end{equation}
where the sum is over all orderings $\rho$ on $\{1,\ldots,n\}\setminus\{b\}$ such that $o(i)$ -- the vertex connected to $i$ by its unique outgoing edge -- precedes $i$ in the word $b\rho$ for all $i\neq b$.
\end{proposition}
In an intuitive sense, the sum appearing on the right-hand-side of the identity \eqref{PTpropident} is over all orderings on the vertices which are compatible with the directed tree graph $G^b$.

\begin{figure}
    \centering
    \includegraphics{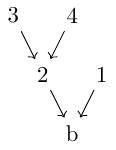} \qquad \qquad
    \includegraphics{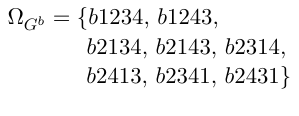}
    \caption{A tree graph $G^b$ and its set of compatible orderings $\Omega_{G^b}$.}
    \label{fig:compatords}
\end{figure}
\begin{figure}
\centering
\includegraphics{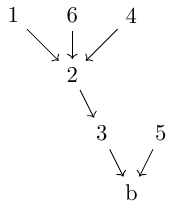} \qquad 
\includegraphics{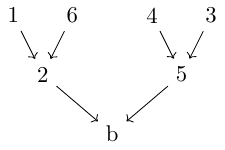}
\caption{The ordering $(b532461)$ is compatible with both of these directed tree graphs rooted at $b$. (Note that this is a non-exhaustive list.)}
\label{fig:compattrees1}
\end{figure}
\begin{figure}
\centering
\includegraphics{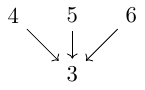} \qquad 
\includegraphics{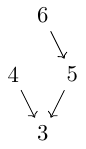} \qquad 
\includegraphics{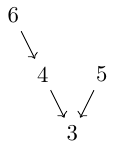}
\caption{All tree graphs rooted at $3$ which are simultaneously compatible with the orderings $(3456|3546)$.}
\label{fig:simcomp}
\end{figure}
More precisely, define an ordering $\rho$ on $\{1,\ldots,n\}\setminus\{b\}$ to be \emph{compatible} with the directed tree $G^b$ if $o(i) <_{b \rho} i$ for all $i$ in the tree. The set of all orderings compatible with $G^b$ will be denoted by $\Omega_{G^b}$. See Figure~\ref{fig:compatords} for an example.

Conversely, we say that a directed tree graph $G^b$ is \emph{compatible} with an ordering $\rho$ if $o(i) <_{b \rho} i$ for all vertices $i\neq b$ of $G^b$. Therefore, to each ordering $\rho$, we can associate a set of compatible directed trees $\cT^{b}_{\rho}$ that satisfy this relation. See Figure~\ref{fig:compattrees1} for an example. Further, it will be useful to consider the intersection $\cT^b_{\rho, \omega} = \cT^b_{\rho} \cap \cT^b_{\omega}$ of all directed trees that are simultaneously compatible with the orderings $\rho$ and $\omega$. See Figure~\ref{fig:simcomp} for an example.

\medskip

We now consider \emph{weighted graphs}, which are graphs where each edge $(i - j)$ is assigned a `weight' $w_{ij}\in\C$.  
\begin{proposition}\label{compTrees}
    The following summations are equivalent:
\begin{equation}\label{eqcompTrees}
    \sum_{T \in \cT^{b}_{\rho,\omega}} \prod_{(i \rightarrow j) \in T} w_{ij} = \prod_{j\neq b} \sum_{\substack{i<_{b\rho}j\\ i<_{b\omega} j}} w_{ij},
\end{equation}
for symmetric weights $w_{ij}=w_{ji}$.
\end{proposition}

\begin{proof}
    We proceed by induction. As the base case, consider two vertices $b$ and $c$, so the only possible ordering for $\rho$ and $\omega$ is the singlet $c$. In this case, \eqref{eqcompTrees} is simply $w_{bc} = w_{bc}$, which is true for symmetric weights.

    Now, suppose the equality \eqref{eqcompTrees} holds when $b \rho,\, b\omega$ are some orderings on $n-1$ letters. We can now create a new orderings $b\rho^+,\, b\omega^+$ on $n$ letters by inserting the new letter $n$ anywhere after $b$ into either ordering. Without loss of generality, we may assume that $n$ is the last letter in $b \rho^+$ (upon relabelling our alphabet appropriately). All compatible directed trees in $\cT^b_{\rho^+, \omega^+}$ can now be constructed by joining the $n$ vertex appropriately to trees from $\cT^{b}_{\rho, \omega}$. Consider some $T \in \cT^{b}_{\rho, \omega}$. It's possible to generate a tree with $n$ vertices from $T$ by either putting $n$ in the middle of an existing edge, or connecting it to a vertex with a single line. Since $n$ is the \emph{last} letter in $\rho^+$, for a compatible tree there can be no $(k \rightarrow n)$ edge. Therefore the only possibility is that $n$ has valence one, with one $(n \rightarrow k)$ edge joining it to $T$. The new vertex $n$ can only join onto vertices $k$ such that $k<_{b\omega}n$, and so the compatible trees formed out of $T$ can be labelled by this joining vertex $i$. 

    Each tree in $\cT^b_{\rho^+, \omega^+}$ is described uniquely in this way. In particular, given any $T^+\in\cT^b_{\rho^+, \omega^+}$, the corresponding `base tree' $T \in \cT^{b}_{\rho, \omega}$ is easily identified by removing the $n$ vertex (which, by the preceding argument, must have valence one), which is uniquely connected to some other vertex, $k$. Therefore.
    \begin{align}
        \sum_{T \in \cT^b_{\rho^+, \omega^+}} \prod_{(i \rightarrow j) \in T} w_{ij} &= \sum_{T \in \cT^b_{\rho, \omega}} \left(\sum_{k <_{b\omega} n} w_{kn} \prod_{(i \rightarrow j) \in T} w_{ij} \right)\\
        &= \sum_{k <_{b\omega} n} w_{kn} \sum_{T \in \cT^b_{\rho, \omega}}  \prod_{(i \rightarrow j) \in T} w_{ij} \\
        &=  \sum_{k <_{b\omega} n} w_{kn} \prod_{j \neq b, n} \sum_{\substack{i<_{b\rho}j\\ i<_{b\omega} j}} w_{ij} \\
        &= \prod_{j \neq b} \sum_{\substack{i<_{b\rho^+}j\\ i<_{b\omega^+} j}} w_{ij}\,,
    \end{align}
    as desired. \qed
\end{proof}

\medskip

A further key result from graph theory tells us how to sum all of the weights of spanning tree graphs of $G$. Suppose $G$ has $n$ vertices, and let $w_{ij}\in\C$ again denote the weight assigned to each edge $(i - j)$, with $w_{ij}=0$ if $(i - j)\notin E$. One then constructs the $n\times n$ \emph{weighted Laplacian matrix of} $G$, $W(G)$, with the entries
\begin{equation}
    W(G)_{ij} = \begin{cases}
        \sum_{(k - i)\in E} w_{ik} & \mathrm{if} \, i = j, \\
        - w_{ij} & \mathrm{if} \, i \neq j,
    \end{cases}\label{Lmatrix}
\end{equation}
This matrix is degenerate, as the sum of elements in each column and row vanishes, so its determinant is zero. However, the minor corresponding to removing the row and column corresponding to vertex $b$, $|W(G)^{b}_{b}|$, is generically non-vanishing.

The following is a classical result in graph theory:
\begin{thmno}[Weighted Matrix-Tree Theorem]\label{WMTthm}
The minor $|W(G)^{b}_{b}|$, obtained by deleting the $b^{\mathrm{th}}$ row and column of the weighted Laplacian matrix of $G$ is independent of $b$, and equal to
\begin{equation}
    \begin{split}
    \left|W(G)^i_i\right| &= \sum_{\substack{T\\ \mathrm{spanning}\,\,  G}} \left( \prod_{(i - j)\in E(T)} w_{ij} \right) \\
     &=\sum_{\substack{T^b\\ \mathrm{spanning}\,\,  G}} \left( \prod_{(i \to j)\in E(T^b)} w_{ij} \right)\,,
    \end{split}
\end{equation}
where the first sum is over all spanning trees of $G$, and the second sum is over all directed spanning trees of $G$ rooted at $b\in V(G)$.
\end{thmno}

\medskip

A further refinement of directed tree graphs is the concept of a \emph{rooted binary tree}. These are tree graphs, with a choice of root inducing a direction, where each vertex (other than the root) in the graph has exactly two incoming edges (hence the `binary') and one outgoing edge; the root itself has only a single incoming edge. As such, rooted binary trees are naturally labelled by the (semi-ordered) set of leaves $L$, which encode the (unique) external edge structure of the graph. By labelling the external edges of the binary tree, the label $L$ can be encoded using the nested bracket notation, as illustrated in Figure \ref{fig:totalBracketing}.
\begin{figure}
    \centering
    \includegraphics{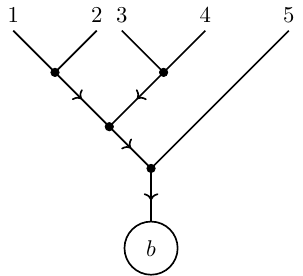}
    \caption{The binary tree corresponding to the complete bracketing $[[[1, 2], [3, 4]], 5]$ of the leaves $L = \{1, 2, 3, 4, 5\}$, with root $b$.}
    \label{fig:totalBracketing}
\end{figure}

There is a natural notion of weighting which can be associated to a rooted binary tree. Suppose a rooted binary tree has $n-1$ external edges, which we label $\{1,\ldots,n-1\}$, and let $E$ denote the set of \emph{internal} edges in the rooted binary tree. To any internal edge $e\in E$, we can assign a weight $w_e\in\C$. Now, any $e\in E$ is uniquely identified with the set $I(e)\subset L$ of leaves in the tree which are `upstream' (with respect to the direction of the graph) from $e$. Equivalently, $I(e)$ labels the (smaller) rooted binary tree rooted at the endpoint of the edge $e$. We can then assign
\be
w_e = \sum_{\{i, j\} \subset I(e)} w_{ij}
\ee
in terms of some symmetric $w_{ij} \in \mathbb{C}$. We denote the product of weights for a rooted binary tree $T$ by
\be\label{totweight}
w_T=\prod_{e\in E}w_e\,,
\ee
and the sum of all weights by
\be\label{sumweight}
w_{\mathrm{total}}=\sum_{e\in E}w_e\,.
\ee
The sum $w_{\mathrm{total}}$ can be viewed as the total weight of the rooted binary tree `flowing' into the root.

Just like tree graphs, binary trees can also be \emph{compatible} with pairs of orderings on $L$. A binary tree $T$ is compatible with the two orderings $\rho$, $\omega$ on $L$ if it is planar with respect to both. In practice, it is possible to algorithmically construct all binary trees compatible with two given orderings (cf., \cite{Cachazo:2013iea}). We denote the set of binary trees, rooted at the vertex $b$, which are compatible with both $\rho$ and $\omega$ by $\mathcal{BT}^{b}_{\rho,\omega}$. There is also a \emph{sign} associated with each pair of orderings $\rho,\,\omega$, given by the number of times the ordering $\omega$ winds around $\rho$ (see~\cite{Mizera:2016jhj} for further discussion).

The following result is an easy generalization of one obtained in~\cite{Frost:2020eoa} for a specific class of weightings associated with cubic Feynman diagrams:
\begin{proposition}\label{propInv}
    Let $\rho$, $\omega$ be orderings on the letters $\{2,\ldots,n-1\}$ and $w_{ij}=w_{ji}\in \C$. Then the two maps
    \be\label{KLTinvS}
S[1\rho|1\omega]=\sum_{T\in\mathcal{T}^{n}_{1\rho,1\omega}}\prod_{(i\rightarrow j)\in E(T)}w_{ij}\,,
    \ee
    \be\label{KLTinvT}
    T[1\rho|1\omega]=\pm\frac{1}{w_{\mathrm{total}}}\,\sum_{T\in\mathcal{BT}^{n}_{1\rho,1\omega}}\frac{1}{w_{T}}\,,
    \ee
    on the space of orderings are inverses of each other.
\end{proposition}

\medskip

The genesis of this result lies in the close relationship between weighted, rooted binary trees and the tree-level Feynman diagrams of \emph{biadjoint scalar} (BAS) \emph{theory}, a theory of massless scalars valued in the tensor product of two Lie algebras, interacting cubically via the structure constants~\cite{Cachazo:2013iea, Mafra:2016ltu, Mafra:2020qst}. The tree-level scattering amplitudes of BAS theory can be decomposed into doubly-colour-ordered partial amplitudes $m_{n}[\rho|\omega]$, and the Feynman diagrams contributing to such a partial amplitude correspond to the rooted binary trees in $\mathcal{BT}^{n}_{\rho,\omega}$, where the root is chosen (without loss of generality) to correspond to the $n^{\mathrm{th}}$ external scalar. 

The weights naturally assigned to these trees are simply the Mandelstam invariants corresponding to the momenta flowing through each internal edge of the Feynman diagram. In particular, if $\{k^{\mu}_i\}_{i=1,\ldots,n}$ are the massless external momenta for the $n$-point BAS scattering process, let
\be\label{mandle1}
k_{I}^{\mu}:=\sum_{i\in I}k_{i}^{\mu}\,.
\ee
Then (as discussed above), the edge $e$ in the binary tree is equivalently labeled by some $I(e)\subset L$ corresponding to the `leaves' contributing to the edge, and one can make the assignment 
\be\label{BASweights}
 w_I\equiv s_I=\sum_{\{i,j\}\subset I(e)}s_{ij}\,, \qquad s_{ij}:=2\,k_i\cdot k_j\,,
\ee
so that
\be\label{BASweight2}
2 \,k_{I}\cdot k_{J}=\sum_{i\in I,\,j\in J}s_{ij}\,.
\ee
In this language, Proposition~\ref{propInv} is the statement of the relation between the tree-level KLT kernel, encoded by the map $S$ in \eqref{KLTinvS}, and the tree-level partial amplitudes, encoded by
\be\label{BAStree0}
m_{n}[1\rho n|1\omega n]=w_{\mathrm{total}}\,T[1\rho|1\omega]\,,
\ee
with $T$ defined by \eqref{KLTinvT}.

 \begin{figure}
 \centering
 \includegraphics{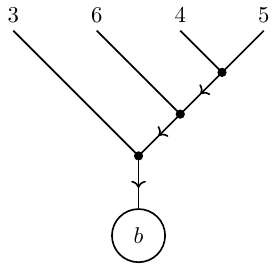} \qquad \qquad
 \includegraphics{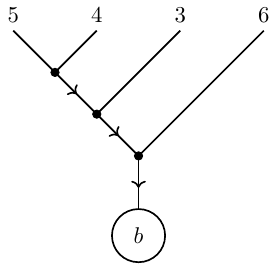}
\caption{The two binary trees compatible with the orderings $(3456|3546)$, rooted at $b$. Legs with arrows are those for which weights are naturally assigned taking contributions from the total momentum of the subtree it connects. Both of these come with a coefficient $-1$ from the winding number.}
 \end{figure}


\section{Helicity-graded integral kernel} \label{sec:deriv}

We now use the machinery of graph theory, reviewed in Section~\ref{graphTheory}, to re-write the Cachazo-Skinner formula \eqref{CSform} in a fashion that manifests a KLT-like double copy structure. In particular, the result is expressed in terms of an integral kernel, which is graded by the helicity configuration and defined as an integrand factor to be integrated over the moduli space of rational maps from the Riemann sphere to twistor space of the appropriate degree (fixed by the helicity configuration). The result builds on the colour-kinematics dual numerators derived for the MHV sector in~\cite{Frost:2021qju}, but extended to identify a double copy kernel and generalised to all helicity configurations.


\subsection{Reduced determinants from orderings and directed trees}

The key idea is to rewrite the reduced determinants appearing in the Cachazo-Skinner formula, namely $\mathrm{det}'(\HH)$ and $\mathrm{det}'(\HH^\vee)$, into a form resembling a minor of a weighted Laplacian matrix\footnote{The connection with weighted Laplacian matrices has been noted many times before, particularly for the reduced determinant of $\HH$ (cf., \cite{Feng:2012sy,Adamo:2012xe,Skinner:2013xp,Adamo:2021bej}).}. Then the weighted matrix-tree theorem will enable us to translate the formula into the language of sums over weighted graphs, where we can apply further results from Section~\ref{graphTheory}.

To begin, consider the matrix $\HH$, whose entries are given by \eqref{HMat}. Using the definition of the reduced determinant \eqref{redHdet} and elementary properties of determinants, it is straightforward to show that $\mathrm{det}'(\HH)$ can be written as
\be\label{reddetH}
\det'(\mathbb{H}) =\frac{|B_b^b|}{|\bolth|^2}\, \prod_{\substack{j \in \bolh \\ l \in\bolth} }\frac{1}{(j\,l)^2}\,\prod_{k=1}^{n}\D\sigma_k\,,
\ee
where $b\in\bolh$ is arbitrary and $B$ is a $(n-d-1)\times (n-d-1)$ matrix with entries 
\begin{equation}\label{BMat}
\begin{split}
B_{ij} &= \mathbb{H}_{ij}\, \prod_{l \in \bolth} (i\,l)\, (j\,l)\,, \qquad i,j\in\bolh\,,\, i\neq j\,,\\
B_{ii} &= - \sum_{\substack{j\in\bolh \\ j\neq i}} B_{ij}\,, \qquad i\in\bolh\,.
\end{split}
\end{equation}
In particular, $B$ has the form of a weighted Laplacian matrix for the simple graph on the set of vertices $\bolh$, where the weight corresponding to the edge $(i - j)$, for $i,j\in\bolh$, is
\be\label{Bweight}
B_{ij}:=-t_i\,t_j\,\frac{[i\,j]}{(i\,j)}\,\prod_{l \in \bolth} (i\,l)\, (j\,l)\,.
\ee
Note that, as required, this is symmetric (i.e., $B_{ij}=B_{ji}$). 

By the weighted matrix tree theorem, this means that the reduced determinant of $\HH$ is equal to
\be\label{HMT1}
\det'(\mathbb{H})=\frac{1}{|\bolth|^2}\,\prod_{\substack{k\in\bolh \\ l\in\bolth}}\frac{1}{(k\,l)^2}\sum_{\substack{T^{b} \\ \mathrm{spanning}\,\,\bolh}}\prod_{(i\to j)} B_{ij}\,,
\ee
where we have stripped off an overall factor of $\prod_{i=1}^n\D\sigma_i$ and the sum is over directed trees on the set $\bolh$, with the direction fixed by the (arbitrary) choice of $b\in\bolh$. Now, for each $T^b$ -- that is, for each directed tree rooted at $b$ -- each vertex $i\in\bolh\setminus\{b\}$ appears only once as the source of an edge $(i\to j)$, as every vertex other than $b$ has a unique outgoing edge. Consequently, it follows that
\be\label{HMT2}
\begin{split}
\det'(\mathbb{H})&=\frac{1}{|\bolth|^2}\,\prod_{\substack{k\in\bolh \\ l\in\bolth}}\frac{1}{(k\,l)^2}\,\prod_{\substack{m\in\bolth \\ a\in\bolh\setminus\{b\}}}t_a\,(a\,m)\,\sum_{\substack{T^{b} \\ \mathrm{spanning}\,\,\bolh}}\prod_{(i\to j)}\left(-t_j\,\frac{[i\,j]}{(i\,j)}\,\prod_{l\in\bolth}(j\,l)\right) \\
&=\frac{1}{|\bolth|^2}\,\prod_{l\in\bolth}\frac{1}{(b\,l)^2}\,\prod_{\substack{m\in\bolth \\ a\in\bolh\setminus\{b\}}}\frac{t_a}{(a\,m)}\,\sum_{\substack{T^{b} \\ \mathrm{spanning}\,\,\bolh}}\prod_{(i\to j)}\left(-t_j\,\frac{[i\,j]}{(i\,j)}\,\prod_{l\in\bolth}(j\,l)\right) \\
&=\frac{1}{|\bolth|^2}\,\prod_{l\in\bolth}\frac{1}{(b\,l)^2}\,\sum_{\substack{T^{b} \\ \mathrm{spanning}\,\,\bolh}}\prod_{(i\to j)}\left(-t_i\,t_j\,\frac{[i\,j]}{(i\,j)}\,\prod_{l\in\bolth}\frac{(j\,l)}{(i\,l)}\right)\,.
\end{split}
\ee
Now, pick any two distinct $x,y\in\bolth$; these can be used to trivially re-express the reduced determinant as
\begin{multline}\label{HMT3}
\det'(\mathbb{H})=\frac{1}{|\bolth|^2}\,\prod_{l\in\bolth}\frac{1}{(b\,l)^2}\,\sum_{\substack{T^{b} \\ \mathrm{spanning}\,\,\bolh}}\prod_{(i\to j)}\left(-t_i\,t_j\,[i\,j]\,(i\,j)\,\prod_{l\in\bolth\setminus\{x,y\}}\frac{(j\,l)}{(i\,l)}\right) \\
\times\,\prod_{(i\to j)}\frac{(j\,x)}{(i\,j)\,(i\,x)}\,\prod_{(i\to j)}\frac{(j\,y)}{(i\,j)\,(i\,y)}\,,
\end{multline}
which can now be further processed using Proposition~\ref{PTprop}.

In particular, applying \eqref{PTpropident} to each of the factors in the second line of \eqref{HMT3}, we arrive at an expression for the reduced determinant in terms of the broken Parke-Taylor factors \eqref{brokePT}:
\begin{multline}\label{HMT4}
\det'(\mathbb{H})=\frac{1}{|\bolth|^2}\,\prod_{l\in\bolth}\frac{1}{(b\,l)^2}\,\sum_{\substack{T^{b} \\ \mathrm{spanning}\,\,\bolh}}\prod_{(i\to j)}\left(-t_i\,t_j\,[i\,j]\,(i\,j)\,\prod_{l\in\bolth\setminus\{x,y\}}\frac{(j\,l)}{(i\,l)}\right) \\
\times\,\sum_{\rho\in\Omega_{T^b(\bolh)}}\mathrm{pt}_{n-d-1}[b\rho]\,\frac{(b\,x)}{(\rho^{*}\,x)}\,\sum_{\omega\in\Omega_{T^b(\bolh)}}\mathrm{pt}_{n-d-1}[b\omega]\,\frac{(b\,y)}{(\omega^*\,y)}\,,
\end{multline}
where the sums in the second line are over orderings on $\bolh\setminus\{b\}$ compatible with the directed, rooted tree $T^b$, and $\rho^*,\omega^*\in\bolh\setminus\{b\}$ denote the last elements of the orderings $\rho$ and $\omega$, respectively. The summations over graphs and compatible orderings can now be swapped to give
\begin{multline}\label{HMT5}
\det'(\mathbb{H})=\frac{(b\,x)\,(b\,y)}{\left|\bolth\cup\{b\}\right|^2}\,\sum_{\rho,\omega\in\mathcal{S}(\bolh\setminus\{b\})}\mathrm{pt}_{n-d}[b\rho x]\,\mathrm{pt}_{n-d}[b\omega y] \\
\times\sum_{T\in\mathcal{T}^{b}_{\rho,\omega}}\prod_{(i\to j)\in E(T)}\left(-t_i\,t_j\,[i\,j]\,(i\,j)\,\prod_{l\in\bolth\setminus\{x,y\}}\frac{(j\,l)}{(i\,l)}\right)\,,
\end{multline}
where $\mathcal{S}(\bolh\setminus\{b\})\cong S_{n-d-2}$ is the group of permutations on the $n-d-2$ letters of $\bolh\setminus\{b\}$. Note that the summation in the second line is now over directed trees on $\bolh$ rooted at $b$ which are compatible with the orderings $\rho$ and $\omega$.

\medskip

A sequence of similar manipulations can also be performed to express the reduced determinant of $\HH^{\vee}$, with entries \eqref{dHmat}, in terms of a sum of orderings and compatible directed trees. The steps closely follow those we have gone through for $\mathrm{det}'(\HH)$; only the initial part of the argument -- connecting the reduced determinant of $\HH^{\vee}$ to the determinant of a weighted Laplacian matrix -- differs significantly.

This is primarily because the form of the diagonal entries of $\HH^{\vee}$ does not, as stated in \eqref{dHmat}, take the form of a sum, as needed for the general form of a weighted Laplacian matrix. However, there is a non-trivial identity\footnote{The right-hand side of this identity is actually the form of the diagonal entries for $\HH^{\vee}$ originally used in the formulation of the Cachazo-Skinner formula~\cite{Cachazo:2012kg,Cachazo:2012pz}.} (first proven in a slightly more general form in~\cite{Skinner:2013xp}) for each $i\in\bolth$: 
\be\label{dHdiagi}
\frac{\la\lambda(\sigma_i)\,\d\lambda(\sigma_i)\ra}{\D\sigma_i}=\sum_{\substack{j\in\bolth \\ j\neq i}}\frac{\la\lambda(\sigma_i)\,\lambda(\sigma_j)\ra}{(i\,j)}\,\prod_{k\in\bolth\setminus\{i,j\}}\frac{(k\,i)}{(k\,j)}\,,
\ee
which enables us to perform the necessary manipulations. Indeed, using the basic properties of determinants and the definition \eqref{reddetH} of $\mathrm{det}'(\HH^{\vee})$, one finds that
\be\label{dHMT1}
\det'(\mathbb{H}^{\vee}) = |\bolth|^2\,\left|B^{\vee\,a}_a\right|\,,
\ee
where $B^{\vee}$ is a $(d+1)\times(d+1)$ matrix with entries
\be\label{BveeMat}
\begin{split}
B^{\vee}_{ij}&=-\HH^{\vee}_{ij}\,\prod_{k\in\bolth\setminus\{i\}}\frac{1}{(i\,k)}\,\prod_{l\in\bolth\setminus\{j\}}\frac{1}{(l\,j)}\,, \qquad i,j\in\bolth\,,\, i\neq j\,, \\
B^{\vee}_{ii}&=-\sum_{\substack{j\in\bolth \\ j\neq i}}B^{\vee}_{ij}\,, \qquad i\in\bolth\,.
\end{split}
\ee
In particular, $B^{\vee}$ is a symmetric matrix of corank 1, with the structure of a weighted Laplacian matrix for the simple graph on the set of vertices $\bolth$. Note that this matrix carries projective weight on $\P^1$: each entry $B_{ij}^{\vee}$ has homogeneity $-2$ each $\sigma_i$, for $i\in\bolth$.

At this point, one can apply the weighted matrix tree theorem, Proposition~\ref{PTprop} and a bit of algebra to arrive at the identity
\begin{multline}\label{dHMT2}
\mathrm{det}'(\HH^\vee)=\frac{(a\,s)\,(a\,t)}{\left|\bolth\setminus\{a\}\right|^2}\sum_{\bar{\rho},\bar{\omega}\in\mathcal{S}(\bolth\setminus\{a\})}\mathrm{pt}_{d+2}[a\bar{\rho}s]\,\mathrm{pt}_{d+2}[a\bar{\omega}t] \\
\times\sum_{\bar{T}\in\mathcal{T}^{a}_{\bar{\rho},\bar{\omega}}}\prod_{(i\to j)\in E(\bar{T})}\left(\la\lambda(\sigma_i)\,\lambda(\sigma_j)\ra\,(i\,j)\,\prod_{k\in(\bolth\cup\{s,t\})\setminus\{i,j\}}\frac{(k\,i)}{(k\,j)}\right)\,,
\end{multline}
where $s,t\in\P^1$ are arbitrary points which do not coincide with the vertices of $\bolth$, and $\mathcal{S}(\bolth\setminus\{a\})\cong S_{d}$ is the group of permutations on the $d$ letters of $\bolth\setminus\{a\}$.


\subsection{The momentum kernel}

At this point, the expressions \eqref{HMT5} and \eqref{dHMT2} lead to a new representation of the Cachazo-Skinner formula for the helicity-graded, tree-level graviton S-matrix. \emph{A priori}, it may seem that all that we have accomplished at this point is to turn the compact, manifestly permutation-invariant integrand of \eqref{CSform} into an unwieldy mess of sums over permutations and compatible tree graphs. However, the utility of this description lies in the fact that the broken Parke-Taylor factors appearing in \eqref{HMT5} and \eqref{dHMT2} can be auspiciously combined.

In particular, note that the points $x,y\in\bolth$ appearing in \eqref{HMT5} and $s,t\in\P^1$ appearing in \eqref{dHMT2} are arbitrarily chosen. So without loss of generality, we can set $s=b\in\bolh$ and $x=a\in\bolth$. Then making use of the definition \eqref{brokePT}, it follows that 
\begin{multline}\label{btofPT}
\mathrm{pt}_{n-d}[b\rho a]\,\mathrm{pt}_{n-d}[b\omega y]\,\mathrm{pt}_{d+2}[a\bar{\rho}b]\,\mathrm{pt}_{d+2}[a\bar{\omega}t] = \\
\frac{(a\,b)\,(\omega^*\,\bar{\omega}^*)}{(y\,\omega^*)\,(\bar{\omega}^*\,t)}\,\frac{\mathrm{PT}_{n}[a\bar{\rho}b\rho]\,\mathrm{PT}_{n}[\bar{\omega}^{\mathrm{T}}ab\omega]}{\prod_{i=1}^{n}\D\sigma_i^2}\,,
\end{multline}
where $\omega^*$ and $\bar{\omega}^*$ denote the final entries in the orderings $\omega$ and $\bar{\omega}$, respectively.

This enables us to re-write the tree-level N$^{d-1}$MHV graviton amplitude as
\begin{multline}\label{altCS1}
\cM_{n,d}=\sum_{\substack{b\rho,b\omega\in\mathcal{S}(\bolh) \\ a\bar{\rho},a\bar{\omega}\in\mathcal{S}(\bolth)}}\int\d\mu_{n,d}\,|\bolth|^{4}\,\frac{(a\,b)\,(\omega^*\,\bar{\omega}^*)\,(b\,y)\,(a\,t)}{(y\,\omega^*)\,(\bar{\omega}^*\,t)} \prod_{k\in\bolth\setminus\{a\}}\frac{(k\,a)^2}{(k\,b)^2}\,\fullpt_{n}[a\bar{\rho}b\rho] \\
\fullpt_{n}[\bar{\omega}^{\mathrm{T}}ab\omega]\left[\sum_{T\in\mathcal{T}^{b}_{\rho,\omega}}\prod_{(i\to j)\in E(T)}\left(-t_i\,t_j\,[i\,j]\,(i\,j)\,\prod_{l\in\bolth\setminus\{a,y\}}\frac{(j\,l)}{(i\,l)}\right)\right] \\
\left[\sum_{\bar{T}\in\mathcal{T}^{a}_{\bar{\rho},\bar{\omega}}}\prod_{(i\to j)\in E(\bar{T})}\left(\la\lambda(\sigma_i)\,\lambda(\sigma_j)\ra\,(i\,j)\prod_{k\in(\bolth\cup\{b,t\})\setminus\{i,j\}}\frac{(k\,i)}{(k\,j)}\right)\right]\prod_{i\in\bolh}h_i \prod_{j\in\bolth}\tilde{h}_j\,,
\end{multline}
where
\be\label{ndmeasure}
\d\mu_{n,d}:=\frac{\d\mu_d}{\prod_{i=1}^{n}\D\sigma_i}\,,
\ee
and we have abbreviated $h_i\equiv h_i(Z(\sigma_i))$, $\tilde{h}_j\equiv\tilde{h}_j(Z(\sigma_j))$. By further manipulating the summations over tree graphs, this can be equivalently written as
\begin{multline}\label{altCS2}
\cM_{n,d}=\sum_{\substack{b\rho,b\omega\in\mathcal{S}(\bolh) \\ a\bar{\rho},a\bar{\omega}\in\mathcal{S}(\bolth)}}\int\d\mu_{n,d}\,|\bolth|^{8}\,\frac{(a\,b)\,(\omega^*\,\bar{\omega}^*)\,(b\,y)\,(a\,t)\,(y\,t)^2}{(y\,\omega^*)\,(\bar{\omega}^*\,t)} \prod_{\substack{k\in\bolh\setminus\{b,t\} \\ l\in\bolth\setminus\{a,y\}}}\frac{1}{(k\,l)^2}\,\fullpt_{n}[a\bar{\rho}b\rho] \\
\times\fullpt_{n}[\bar{\omega}^{\mathrm{T}}ab\omega]\left[\sum_{T\in\mathcal{T}^{b}_{\rho,\omega}}\prod_{(i\to j)\in E(T)}\left(-t_i\,t_j\,[i\,j]\,(i\,j)\,\prod_{l\in\bolth\setminus\{a,y\}}(i\,l)\,(j\,l)\right)\right] \\
\left[\sum_{\bar{T}\in\mathcal{T}^{a}_{\bar{\rho},\bar{\omega}}}\prod_{(i\to j)\in E(\bar{T})}\left(\frac{\la\lambda(\sigma_i)\,\lambda(\sigma_j)\ra}{(i\,j)}\prod_{k\in(\bolth\cup\{b,t\})\setminus\{i,j\}}\frac{1}{(k\,i)\,(k\,j)}\right)\right]\prod_{i\in\bolh}h_i \prod_{j\in\bolth}\tilde{h}_j\,.
\end{multline}
At this stage, the advantage of processing the Cachazo-Skinner formula in this way becomes clear: there is now a double copy structure apparent in the formula. In particular, the integrand now contains \emph{two} copies of the $\P^1$ Parke-Taylor factor $\fullpt_n$, each of which defines the integrand of the RSVW formula for the tree-level gluon S-matrix, glued together with an integral `kernel' and summed over a basis of colour-orderings.

To see this more explicitly, define
\be\label{genweightsH}
\phi_{ij}:=-t_i\,t_j\,[i\,j]\,(i\,j)\,\prod_{l\in\bolth\setminus\{a,y\}}(i\,l)\,(j\,l)\,, \qquad i,j\in\bolh\,,
\ee
which is weightless in each $i,j\in\bolh$ and weight $2$ in each $l\in\bolth\setminus\{a,y\}$, 
\be\label{genweightsHt}
\tilde{\phi}_{ij}:=\frac{\la\lambda(\sigma_i)\,\lambda(\sigma_j)\ra}{(i\,j)}\prod_{k\in(\bolth\cup\{b,t\})\setminus\{i,j\}}\frac{1}{(k\,i)\,(k\,j)}\,, \qquad i,j\in\bolth\,,
\ee
which carries weight $-2$ in each $k\in\bolth\cup\{b,t\}$, and
\be\label{genwfactor}
\cD[\omega,\bar{\omega}]:=\frac{(a\,b)\,(\omega^*\,\bar{\omega}^*)\,(b\,y)\,(a\,t)\,(y\,t)^2}{(y\,\omega^*)\,(\bar{\omega}^*\,t)} \prod_{\substack{k\in\bolh\setminus\{b,t\} \\ l\in\bolth\setminus\{a,y\}}}\frac{1}{(k\,l)^2}\,,
\ee
which has weight $+2$ in $\{a,b,t,y\}$, weight $-2(d-1)$ in all elements of $\bolh\setminus\{b,t\}$ and weight $-2(n-d-3)$ in all elements of $\bolth\setminus\{a,y\}$. Let
\be\label{intkernel}
S_{n,d}[\rho,\bar{\rho}|\omega,\bar{\omega}]:=\cD[\omega,\bar{\omega}]\,\left[\sum_{T\in\mathcal{T}^{b}_{\rho,\omega}}\prod_{(i\to j)\in E(T)}\phi_{ij}\right]\,\left[\sum_{\bar{T}\in\mathcal{T}^{a}_{\bar{\rho},\bar{\omega}}}\prod_{(i\to j)\in E(\bar{T})}\tilde{\phi}_{ij}\right]\,,
\ee
be the \emph{integral kernel}; this can be viewed as a linear transformation $S:\mathcal{S}(\bolh\setminus\{b\})\times\mathcal{S}(\bolth\setminus\{a\}) \to \mathcal{S}(\bolh\setminus\{b\})\times\mathcal{S}(\bolth\setminus\{a\})$ on the space of orderings.

In terms of this integral kernel, the graviton tree amplitude becomes
\be\label{altCS3}
\cM_{n,d}=\sum_{\substack{b\rho,b\omega\in\mathcal{S}(\bolh) \\ a\bar{\rho},a\bar{\omega}\in\mathcal{S}(\bolth)}}\int\d\mu_{n,d}\,|\bolth|^{8}\,\fullpt_{n}[a\bar{\rho}b\rho] \, S_{n,d}[\rho,\bar{\rho}|\omega,\bar{\omega}]\,\fullpt_{n}[\bar{\omega}^{\mathrm{T}}ab\omega] \,\prod_{i\in\bolh}h_i\,\prod_{j\in\bolth}\tilde{h}_j\,.
\ee
The double copy structure here can be made even more explicit by writing the RSVW formula for the tree-level N$^{d-1}$MHV colour-ordered partial amplitude as
\be\label{altRSVW}
\cA_{n,d}[\rho]=\int \d\mu_{d}\,\cI^{\tilde{\bolg}}_{n}[\rho]\,\prod_{i\in\bolg}a_i(Z(\sigma_i))\,\prod_{j\in\tilde{\bolg}}b_{j}(Z(\sigma_j))\,,
\ee
for the integrand
\be\label{RSVWint}
\cI^{\tilde{\bolg}}_{n}[\rho]:=|\tilde{\bolg}|^4\,\fullpt_n[\rho]\,.
\ee
In particular, the new representation \eqref{altCS3} of the Cachazo-Skinner formula is
\be\label{altCS4}
\cM_{n,d}=\sum_{\substack{b\rho,b\omega\in\mathcal{S}(\bolh) \\ a\bar{\rho},a\bar{\omega}\in\mathcal{S}(\bolth)}}\int\d\mu_{n,d}\,\cI_{n}^{\bolth}[a\bar{\rho}b\rho] \, S_{n,d}[\rho,\bar{\rho}|\omega,\bar{\omega}]\,\cI_{n}^{\bolth}[\bar{\omega}^{\mathrm{T}}ab\omega] \,\prod_{i\in\bolh}h_i\,\prod_{j\in\bolth}\tilde{h}_j\,,
\ee
which is clearly a double copy of the Yang-Mills integrand, glued together via the integral kernel.

Indeed, each of $\cI^{\bolth}_{n}[a\bar{\rho}b\rho]$ and $\cI_{n}^{\bolth}[\bar{\omega}^{\mathrm{T}}ab\omega]$ can be viewed as a vector in the $(n-d-2)!\times d!$-dimensional space of orderings $\mathcal{S}(\bolh\setminus\{b\})\times\mathcal{S}(\bolth\setminus\{a\})$. The sum over orderings in \eqref{altCS4} then simply corresponds to the multiplication of these two vectors via the square matrix $S$ corresponding to the integral kernel. 

It should be noted that this formula is independent of the choices of $\{a,b,y,t\}$. While this is far from obvious from the final expression \eqref{altCS4}, we saw that at each stage in the derivation where these choices were introduced, it was clear that they were entirely arbitrary.

\medskip

At this point, it is worth comparing the structure of \eqref{altCS4} with the more standard double copy structure arising from the usual KLT momentum kernel~\cite{Kawai:1985xq}. In particular, the KLT double copy representation of tree-level graviton scattering amplitudes takes the form
\be\label{KLTdc1}
\cM_{n,d}=\sum_{\alpha,\beta\in S_{n-3}}\cA_{n,d}[(n-1)n\alpha 1]\,S^{\mathrm{KLT}}[\alpha|\beta]\,\cA_{n,d}[1\beta(n-1)n]\,,
\ee
where $S^{\mathrm{KLT}}[\alpha|\beta]$ is the KLT momentum kernel, a rational function of the kinematic invariants given by~\cite{Bjerrum-Bohr:2010diw,Bjerrum-Bohr:2010kyi}
\be\label{KLTker}
S^{\mathrm{KLT}}[\alpha|\beta]=\prod_{i=2}^{n-2}\left(s_{1\alpha(i)}+\sum_{j>i}^{n-2}\theta_{\beta}(\alpha(i),\alpha(j))\,s_{\alpha(i)\alpha(j)}\right)\,,
\ee
where $s_{ij}:=(k_i+k_j)^2$ are the Mandelstam invariants and
\be\label{thetaorder}
\theta_{\beta}(\alpha(i),\alpha(j))=\left\{\begin{array}{l}
0 \qquad \mbox{if } \alpha(i)<\alpha(j) \mbox{ in } \beta \\
1 \qquad \mbox{otherwise }
\end{array}\right.\,.
\ee
Here, the choice of three external legs, $1$, $n-1$ and $n$ is arbitrary -- any three legs can be chosen; the point is that the graviton takes the form of a sum over two copies of a basis of $(n-3)!$ colour-ordered gluon amplitudes, multiplied together via the KLT kernel. The size of this basis can be understood as the reduction of the na\"ive number of colour-orderings, namely $n!$, to $(n-3)!$ due to the photon decoupling~\cite{Mangano:1990by}, Kleiss-Kuijff~\cite{Kleiss:1988ne,DelDuca:1999rs} and Bern-Carassco-Johansson~\cite{Bern:2008qj} relations.

By contrast, our expression \eqref{altCS4} is a sum over integrals of a basis of $(n-d-2)!\times d!$ colour-ordered gluon integrands multiplied together by the integral kernel. A closer comparison with momentum space formulations of KLT double copy becomes clear when we perform a chiral splitting of our kernel \eqref{intkernel}.


\subsection{Chirally split momentum kernel}

It is easy to see that the integral kernel \eqref{intkernel} respects a chiral splitting in terms of the underlying helicity decomposition of the external states in the scattering amplitude. In particular, let $\hat{\rho}:=\rho\cup\bar{\rho}$ be the ordering on $\{1,\ldots,n\}\setminus\{a,b\}$ (i.e., and element of $S_{n-2}$) arising from the union of the orderings $\rho$ on $\bolh\setminus\{b\}$ and $\bar{\rho}$ on $\bolth\setminus\{a\}$, with $\hat{\omega}:=\omega\cup\bar{\omega}$ defined similarly. Then the prefactor $\mathcal{D}(\omega,\bar{\omega})=\mathcal{D}(\hat{\omega})$ can be viewed as a diagonal matrix $D$ acting as a linear transformation on $S_{n-2}$:
\be\label{diagD}
D=\mathrm{diag}\left(\mathcal{D}(\hat{\omega})\right)\,,
\ee
whose inverse $D^{-1}$ is simply given by inverting $\mathcal{D}(\hat{\omega})$ for each ordering.

Then we define the \emph{reduced kernel} $\bbS_{n,d}:=S_{n,d}\,D^{-1}$, which admits a natural chiral splitting in terms of orderings on $\bolth$ and $\bolh$ separately. To see this, define the negative helicity kernel as
\be\label{nhmker}
\bbS_{\bolth}[\bar{\rho}|\bar{\omega}]:=\sum_{\bar{T}\in\mathcal{T}^{a}_{\bar{\rho},\bar{\omega}}}\prod_{(i\to j)\in E(\bar{T})}\tilde{\phi}_{ij}\,,
\ee
and the positive helicity kernel as
\be\label{phmker}
\bbS_{\bolh}[\rho|\omega]:=\sum_{T\in\mathcal{T}^{b}_{\rho,\omega}}\prod_{(i\to j)\in E(T)}\phi_{ij}\,,
\ee
upon which the reduced kernel splits as
\be\label{splitker}
\bbS_{n,d}=\bbS_{\bolth}\otimes \bbS_{\bolh}\,.
\ee
In particular, the reduced kernel admits a chiral splitting that respects the underlying decomposition of sums over (partial) colour-orderings of the single-copy gluon amplitudes.

Our results so far can be summarized as the following:
\begin{thm}\label{thm:PT-KLT}
The tree-level S-matrix of general relativity in Minkowski spacetime admits the helicity-graded double copy representation:
\be\label{PT-KLT}
\cM_{n,d}=\sum_{\substack{b\rho,b\omega\in\mathcal{S}(\bolh) \\ a\bar{\rho},a\bar{\omega}\in\mathcal{S}(\bolth)}}\int\d\mu_{n,d}\,\cI_{n}^{\bolth}[a\bar{\rho}b\rho] \, S_{n,d}[\hat{\rho}|\hat{\omega}]\,\cI_{n}^{\bolth}[\bar{\omega}^{\mathrm{T}}ab\omega] \,\prod_{i\in\bolh}h_i\,\prod_{j\in\bolth}\tilde{h}_j\,,
\ee
in terms of the integral kernel $S_{n,d}$, which itself obeys a chiral splitting:
\be\label{ikchiral}
S_{n,d}[\hat{\rho}|\hat{\omega}]=\mathcal{D}(\hat{\omega})\,\bbS_{\bolth}[\bar{\rho}|\bar{\omega}]\,\bbS_{\bolh}[\rho|\omega]\,,
\ee
in terms of $\bbS_{\bolth}$ and $\bbS_{\bolh}$ the negative and positive helicity kernels \eqref{nhmker} and \eqref{phmker}, respectively.
\end{thm}

\medskip

The chiral splitting \eqref{ikchiral} of the integral kernel helps to make the connection between the representation \eqref{PT-KLT} and other field theory KLT double copy representations in the literature more apparent. Although the standard version of the KLT double copy is as given in \eqref{KLTdc1}, a generalization of this formula in terms of `shortened' KLT kernels was derived in~\cite{Bjerrum-Bohr:2010diw,Bjerrum-Bohr:2010mtb,Bjerrum-Bohr:2010kyi,Bjerrum-Bohr:2010pnr} using (the field theory limit of) string monodromy relations\footnote{We will use one particular version of this generalization, a special case of which appeared much earlier in~\cite{Bern:1998ug}.}. This generalized KLT double copy can also be stated in helicity-graded form as.
\begin{multline}\label{gKLTdc1}
\cM_{n,d}=(-1)^{n-3}\sum_{\hat{\alpha}\in S_{n-3}}\:\sum_{\substack{\rho\in\mathcal{S}(\bolh\setminus\{b,n\}) \\ \bar{\omega}\in\mathcal{S}(\bolth\setminus\{a\})}} \cA_{n,d}[a\hat{\alpha} bn]\, \bbS^{\mathrm{KLT}}[\alpha|\rho(\alpha)] \\
\times\, \widetilde{\bbS}^{\mathrm{KLT}}[\bar{\omega}(\bar{\alpha})|\bar{\alpha}]\:\cA_{n,d}[\bar{\omega}(\alpha)ab\rho(\alpha)n]\,,
\end{multline}
where the first sum is over orderings $\hat{\alpha}=\alpha\cup\bar{\alpha}$ on $n-3$ letters, decomposed into $\alpha\in S_{n-d-3}$ and $\bar{\alpha}\in S_{d}$. The chirally-split KLT momentum kernels are
\be\label{KLTkerph}
\bbS^{\mathrm{KLT}}[\alpha|\rho(\alpha)]:=\prod_{i\in\bolh\setminus\{b,n\}}\left(s_{\alpha(i)\,b}+\sum_{j>i}\theta_{\rho(\alpha)}(\alpha(i),\alpha(j))\,s_{\alpha(i)\,\alpha(j)}\right)\,,
\ee
and
\be\label{KLTkernh}
\widetilde{\bbS}^{\mathrm{KLT}}[\bar{\omega}(\bar{\alpha})|\bar{\alpha}]:=\prod_{i\in\bolth\setminus\{a\}}\left(s_{\bar{\alpha}(i)\,a}+\sum_{j<i}\theta_{\bar{\omega}(\bar{\alpha})}(\bar{\alpha}(j),\bar{\alpha}(i))\,s_{\bar{\alpha}(i)\,\bar{\alpha}(j)}\right)\,,
\ee
with the ordering function $\theta$ defined as before in \eqref{thetaorder} and $s_{ij}$ the Mandelstam invariants.

Clearly, this expression more closely resembles our representation \eqref{PT-KLT} than the generic KLT double copy \eqref{KLTdc1}. The sums over chirally-split orderings $\rho,\bar{\omega}$ also appear in our formula, although the overall sum over orderings $\hat{\alpha}\in S_{n-3}$ does not appear in our expression. Instead, we have an \emph{integral} over the moduli space of rational maps from $\P^1$ to $\PT$, as well as additional complementary sums over orderings $\bar{\rho}$ and $\omega$.

These differences are not so stark as they may seem, and it is interesting to speculate that the two formulae \eqref{PT-KLT} and \eqref{gKLTdc1} are in fact equivalent. The moduli integral in \eqref{PT-KLT} is \emph{completely localised} against the delta function constraints appearing in the twistor wavefunctions, just as it is in the RSVW and Cachazo-Skinner formulae. These constraints are a helicity-graded version of the \emph{scattering equations} -- rational constraints localizing $n$ marked points on $\P^1$ in terms of on-shell kinematic data -- so the moduli integral is really a sum over solutions to these helicity-graded scattering equations. At N$^{d-1}$MHV, the number of such solution is given by the Eulerian number $E(n - 3, d-1)$~\cite{Spradlin:2009qr, Cachazo:2013iaa, Roehrig:2017wvh}, so the total number of terms summed in \eqref{PT-KLT} is $E(n -3, d-1) \times (d\,!)^2 \times ((n - d - 2)!)^2$. 
This is generally higher than the number of term in the sums \eqref{gKLTdc1} and \eqref{KLTdc1}, however there may be some form of redundancy in the solutions owing to a version of KLT orthogonality~\cite{Cachazo:2012da, Cachazo:2013gna}.

While these speculations suggest a very close relationship between the chirally split, integral kernel manifestation of double copy given by \eqref{PT-KLT} and the more standard momentum space formulation of \eqref{gKLTdc1}, their equivalence is, for generic $d$, purely conjectural. However, for the MHV helicity configuration (corresponding to $d=1$) where \eqref{gKLTdc1} coincides with the un-graded KLT double copy \eqref{KLTdc1}, we can make this equivalence completely precise.


\subsection{Example: the MHV configuration}

The MHV helicity configuration is particularly straightforward to analyse, as the moduli integrals can be performed explicitly. In this case, the various ingredients of \eqref{PT-KLT} simplify considerably.

Fixing $d=1$, without loss of generality let $\bolth = \{1, 2\}$ and then we can choose $a = 2$, $y=1$. Furthermore, the orderings $\bar{\rho},\,\bar{\omega}$ on $\bolth\setminus\{a\}$ are now an ordering on a single letter; namely, $\bar{\rho}=\bar{\omega}=1$, and thus $\bar{\omega}^*=1$ trivially. This means hat for the $d=1$ configuration,
\be\label{MHVcD}
\mathcal{D}(\hat{\omega})=(b\,1)\,(b\,2)\,(t\,1)\,(t\,2)\,,
\ee
and
\be\label{MHVSth}
\bbS_{\bolth}[1|1]=\frac{\left\la\lambda(\sigma_1)\,\lambda(\sigma_2)\right\ra}{(1\,2)}\,\frac{1}{(b\,1)\,(b\,2)\,(t\,1)\,(t\,2)}\,.
\ee
Furthermore, the residual GL$(2,\C)$ freedom in the parametrization of the map moduli $U^{A}_{\ba}$ can be fixed by setting
\be\label{MHVmoduli}
U^{A}_{\ba}=\left(x^{\beta\dot\alpha}\,\delta_{\beta\ba},\,\delta_{\alpha\ba}\right)\,,
\ee
where this choice of gauge for the moduli identifies the anti-self-dual SL$(2,\C)$ spinor indices $\alpha,\beta,\ldots$ with the holomorphic homogeneous coordinate indices $\ba,\mathbf{b},\ldots$ of $\P^1$; the four $x^{\alpha\dot\alpha}$ are the remaining, gauge-fixed moduli of the degree one map.

With these choices, the $d=1$ double copy representation \eqref{PT-KLT} simplifies to
\begin{multline}\label{MHVKLT1}
\cM_{n,1}=\sum_{b\rho,b\omega\in\mathcal{S}(\bolh)}\int\d^{4}x\,(1\,2)^6\,\brpt_{n}[1b\rho2]\,\brpt_{n}[2b\omega1]\,\left[\sum_{T\in\mathcal{T}^{b}_{\rho,\omega}}\prod_{(i\to j)\in E(T)}-t_i\,t_j\,[i\,j]\,(i\,j)\right] \\
\times \prod_{k=1,2}t_k^{5}\,\d t_k\,\bar{\delta}^{2}(\kappa_k-t_k\,\sigma_k)\,\e^{\im\,t_k\,x^{\alpha\dot\alpha}\,\sigma_{\alpha\,k}\,\tilde{\kappa}_{\dot\alpha\,k}}\,\prod_{l=3}^{n}\frac{\d t_l}{t_l^3}\,\bar{\delta}^{2}(\kappa_l-t_l\,\sigma_l)\,\e^{\im\,t_l\,x^{\alpha\dot\alpha}\,\sigma_{\alpha\,l}\,\tilde{\kappa}_{\dot\alpha\,l}}\,.
\end{multline}
At this stage, all of the integrals can be performed explicitly against holomorphic delta functions or as simple exponential integrals, leaving
\begin{multline}\label{MHVKLT2}
\cM_{n,1}=(2\pi)^4\,\delta^{4}\!\left(\sum_{i=1}^{n}k_i\right)\,\sum_{\rho,\omega\in S_{n-3}}\hat{\cA}_{n,1}[1b\rho 2] \left[\sum_{T\in\mathcal{T}^{b}_{\rho,\omega}}\prod_{(i\to j)\in E(T)}-[i\,j]\,\la i\,j\ra\right]\hat{\cA}_{n,1}[12b\omega] \\
=(2\pi)^4\,\delta^{4}\!\left(\sum_{i=1}^{n}k_i\right)\,\sum_{\rho,\omega\in S_{n-3}}\hat{\cA}_{n,1}[1b\rho 2] \left[\sum_{T\in\mathcal{T}^{b}_{\rho,\omega}}\prod_{(i\to j)\in E(T)}s_{ij}\right]\hat{\cA}_{n,1}[12b\omega] \,,
\end{multline}
where $\hat{\cA}_{n,1}$ is the colour-ordered Parke-Taylor formula \eqref{ParkeTaylor}, stripped of its overall momentum-conserving delta functions.

Now, by Proposition~\ref{compTrees} and the results of \cite{Bjerrum-Bohr:2010pnr} it follows that
\be\label{MHVKLTKer}
\sum_{T\in\mathcal{T}^{b}_{\rho,\omega}}\prod_{(i\to j)\in E(T)}s_{ij}=\prod_{\substack{j=3 \\ i\neq b}}^{n}
\sum_{\substack{i<_{b\rho}j \\ i<_{b\omega}j}} s_{ij}=S^{\mathrm{KLT}}[\rho|\omega]\,.
\ee
In particular, it means that \eqref{MHVKLT2} is equivalent to 
\be\label{MHVKLT3}
\cM_{n,1}=(2\pi)^4\,\delta^{4}\!\left(\sum_{i=1}^{n}k_i\right)\sum_{\rho,\omega\in S_{n-3}}\hat{\cA}_{n,1}[1b\rho 2]\,S^{\mathrm{KLT}}[\rho|\omega]\,\hat{\cA}_{n,1}[12b\omega]\,,
\ee
which reproduces the known field theory KLT double copy for MHV amplitudes~\cite{Berends:1988zp}. A representation similar to this in terms of colour-kinematics dual numerators was also derived directly from the Hodges formula by Frost in~\cite{Frost:2021qju}. Here we have demonstrated that our general N$^{d-1}$MHV double copy representation \eqref{PT-KLT} is equivalent to the Berends-Giele-Kuijf formula for graviton scattering in the MHV sector.


\section{Biadjoint scalar amplitudes in twistor space} \label{sec:proof}

It is a well-known fact that the \emph{inverse} of the field theory KLT kernel encodes the tree-level S-matrix of biadjoint scalar (BAS) theory~\cite{Cachazo:2013gna, Mizera:2016jhj, Frost:2020eoa}; these tree-level scattering amplitudes can be decomposed into doubly colour-ordered partial amplitudes $m_n[\alpha|\beta]$ for $\alpha,\beta\in S_n$. In particular, written in terms of permutations $\rho,\omega$ on $n-3$ letters, 
\be\label{BAS-KLT}
m_{n}[12b\rho|21b\omega]=(2\pi)^4\,\delta^{4}\!\left(\sum_{i=1}^{n}k_i\right)\,S^{-1}_{\mathrm{KLT}}[\rho|\omega]\,,
\ee
where the field theory KLT kernel for these orderings is given by \eqref{MHVKLTKer}, and $S^{-1}_{\mathrm{KLT}}$ is its inverse, computed via Proposition~\ref{propInv}:
\be\label{KLTinverse}
S^{-1}_{\mathrm{KLT}}[\rho|\omega]=\frac{1}{s_{34\cdots n}}\,\sum_{T\in\mathcal{BT}_{b\rho,b\omega}}\frac{1}{s_T}=\frac{1}{s_{12}}\,\sum_{T\in\mathcal{BT}_{b\rho,b\omega}}\frac{1}{s_T}\,,
\ee
where the second equality follows on the support of overall momentum conservation. 

In light of this connection between inverted field theory KLT kernels and BAS tree-level amplitudes, it seems natural to ask if there is some way to encode BAS scattering amplitudes in twistor space by inverting the integral kernel \eqref{ikchiral}. Remarkably, this is indeed the case, resulting in a series of representations, graded by degree, all of which encode BAS tree amplitudes that see the degree only through their colour-orderings. We first present the resulting formula, and then prove that it is correct by looking at its factorization properties. The detailed proof is contained in Section~\ref{sec:BASproof}.


\subsection{BAS tree-amplitudes}

To invert the integral kernel $S_{n,d}[\hat{\rho}|\hat{\omega}]$, we need to invert each of its constituent parts, as linear maps on the space of ordering: $\mathcal{D}(\hat{\omega})$, $\bbS_{\bolth}[\bar{\rho}|\bar{\omega}]$ and $\bbS_{\bolh}[\rho|\omega]$. The inversion of $\mathcal{D}(\hat{\omega})$ is trivial,
\be\label{cDinverse}
\mathcal{D}^{-1}(\hat{\omega})=\frac{(y\,\omega^*)\,(\bar{\omega}^*\,t)}{(a\,b)\,(\omega^*\,\bar{\omega}^*)\,(b\,y)\,(a\,t)\,(y\,t)^2}\,\prod_{\substack{k\in\bolh\setminus\{b,t\} \\ l\in\bolth\setminus\{a,y\}}}(k\,l)^2\,,
\ee
while the inversion of the chiral factors $\bbS_{\bolth}$, $\bbS_{\bolh}$ is accomplished using Proposition~\ref{propInv}. 

This gives
\be\label{Shinv}
\bbS^{-1}_{\bolh}[\rho|\omega]=\frac{1}{\phi_{\mathrm{total}}}\,\sum_{T\in\mathcal{BT}_{b\rho,b\omega}}\frac{1}{\phi_T}\,,
\ee
and
\be\label{Sthinv}
\bbS^{-1}_{\bolth}[\bar{\rho}|\bar{\omega}]=\frac{1}{\tilde{\phi}_{\mathrm{total}}}\,\sum_{\bar{T}\in\mathcal{BT}_{a\bar{\rho},a\bar{\omega}}}\frac{1}{\tilde{\phi}_{\bar{T}}}\,.
\ee
Here, recall that $\phi_{ij}$ and $\tilde{\phi}_{ij}$ are the weights \eqref{genweightsH} and \eqref{genweightsHt}. The `total' quantities appearing as prefactors in these expressions are 
\be\label{totalphis}
\phi_{\mathrm{total}}=\sum_{\substack{i,j\in\bolh \\ i\neq j}}\phi_{ij}\,, \qquad \tilde{\phi}_{\mathrm{total}}=\sum_{\substack{i,j\in\bolth \\ i\neq j}}\tilde{\phi}_{ij}\,,
\ee
while $\phi_T$, $\tilde{\phi}_{\bar{T}}$ denote the product over all edge weights appearing in the binary trees $T$, $\bar{T}$, respectively. Here each edge weight is $ \phi_E = \sum_{i, j \in I(E)}\phi_{ij}$ where $I(E)$ is the set of leaves \emph{upstream} from $E$, and similarly for $\tilde{\phi}_{\tilde{E}}$.

Collecting these results, one has
\be\label{inkerinverse}
S^{-1}_{n,d}[\hat{\rho}|\hat{\omega}]=\mathcal{D}^{-1}(\hat{\omega})\,\bbS^{-1}_{\bolh}[\rho|\omega]\,\bbS^{-1}_{\bolth}[\bar{\rho}|\bar{\omega}]\,,
\ee
which is an inverse of the integral kernel \eqref{intkernel} in the sense that
\be\label{inversenorm}
\sum_{\substack{\omega\in S_{n-d-2} \\ \bar{\omega}\in S_{d}}}S_{n,d}[\hat{\rho}|\hat{\omega}]\,S^{-1}_{n,d}[\hat{\omega}|\hat{\alpha}]=\sum_{\substack{\omega\in S_{n-d-2} \\ \bar{\omega}\in S_{d}}}S^{-1}_{n,d}[\hat{\alpha}|\hat{\omega}]\,S_{n,d}[\hat{\omega}|\hat{\rho}]=\delta_{\rho\alpha}\,\delta_{\bar{\rho}\bar{\alpha}}\,,
\ee
where the Kronecker deltas represent the identity on the space of linear maps between permutations.

\medskip

Armed with the inverse integral kernel \eqref{inkerinverse}, we are now in a position to state the following result:
\begin{thm}\label{BASthm}
Let
\be\label{PT-BAS}
m_{n,d}[a\bar{\rho}b\rho|\bar{\omega}^{\mathrm{T}}ab\omega]:=\int\d\mu_{d}\,S^{-1}_{n,d}[\hat{\rho}|\hat{\omega}]\,\prod_{i=1}^{n}\varphi_i(Z(\sigma_i))\,\D\sigma_i\,,
\ee
for $\varphi_i\in H^{0,1}(\PT,\cO(-2))$. This expression is equal to the doubly colour-ordered tree-level partial amplitudes of BAS theory, in the sense that
\be\label{PT-BAS2}
m_{n,d}[a\bar{\rho}b\rho|\bar{\omega}^{\mathrm{T}}ab\omega]=m_{n}[a\bar{\rho}b\rho|\bar{\omega}^{\mathrm{T}}ab\omega]\,,
\ee
for each $d=1,\ldots,n-3.$
\end{thm}

\medskip

We will prove this theorem in the following subsection, but for now, let us make a few elementary observations about the formula \eqref{PT-BAS}. Firstly, it is easy to check that the formula is mathematically well-defined, in the sense that all of the integrals make sense projectively. 

To see this, observe that external states in this formula are still partitioned into `positive helicity' and `negative helicity' sets, $\bolh$ and $\bolth$, respectively, despite the fact that there is no helicity associated with the massless scalars being described by the twistor wavefunctions $\varphi_i$. Now, $\mathcal{D}^{-1}(\hat{\omega})$ carries weight $-2$ in $\{a,b,t,y\}$, weight $2(d-1)$ in all elements of $\bolh\setminus\{b,t\}$ and weight $2(n-d-3)$ in all elements of $\bolth\setminus\{a,y\}$. Meanwhile, $\bbS^{-1}_{\bolh}$ carries weight $-2(n-d-2)$ in each element of $\bolth\setminus\{a,y\}$ and $\bbS^{-1}_{\bolth}$ carries weight $2d$ in each element of $\bolth\cup\{b,t\}$. Combined with the scaling weight $-2d$ of each twistor wavefunction and the weight $2$ holomorphic forms $\D\sigma_i$, it follows that the integrand of \eqref{PT-BAS} is a $(1,1)$-form, homogeneous of weight zero in each insertion $\sigma_i$ on the Riemann sphere.

Furthermore, when the twistor wavefunctions correspond to massless scalar momentum eigenstates
\be\label{momeigsc}
\varphi_i(Z)=\int_{\C^*}t_i\,\d t_i\,\bar{\delta}^2(\kappa_i-t_i\,\lambda)\,\e^{\im t_i\,[\mu\,i]}\,,
\ee
the same counting of integrals versus delta functions that holds for the RSVW and Cachazo-Skinner formulae shows that all of the integrals appearing in \eqref{PT-BAS} are localized against delta functions, with four remaining delta functions encoding momentum conservation.

In this sense, the formula for $m_{n,d}$ passes the most basic sanity tests, and it can also be evaluated explicitly when $d=1$. In this case, following the same steps as in our evaluation of $\cM_{n,1}$ above, the formula evaluates to
\be\label{MHV-BAS1}
\begin{split}
m_{n,1}[12 b \rho|21b\omega]&=\frac{1}{s_{3\cdots n}}\,\sum_{T\in\mathcal{BT}_{b\rho,b\omega}}\frac{1}{s_{T}}\,\int\d^{4}x\,\e^{\im (k_1+\cdots+k_n)\cdot x} \\
&=(2\pi)^4\,\delta^{4}\!\left(\sum_{i=1}^{n}k_i\right)\,\frac{1}{s_{12}}\,\sum_{T\in\mathcal{BT}_{b\rho,b\omega}}\frac{1}{s_{T}}\,,
\end{split}
\ee
which is precisely the BAS tree-amplitude $m_{n}[12 b \rho|21b\omega]$. A sort of `parity symmetry' inherited from the graviton amplitude \eqref{PT-KLT} (shown in Appendix \ref{parityApp}) means that we also obtain the concrete, correct expression for $d=n-3$.

The remarkable claim is that this is true \emph{for all degrees} $d$: namely, \eqref{PT-BAS} is equal to the BAS tree amplitude, which is sensitive to the degree of the underlying map to twistor space only through the arrangement of its double colour-orderings. Intuitively, the counting in degree can be interpreted as follows. Due to the specific structure $[a \bar{\rho}b \rho|\bar{\omega}^{\mathrm{T}} a b \omega]$ of the colour-orderings, the corresponding BAS amplitude is divided into two halves: the binary trees on the left-hand set of $(d+1)$ particles, and the binary trees on the right-hand set of $(n - d - 1)$ particles. These trees are treated independently, and so the degree of the map can be interpreted as enumerating the `complexity' of the trees contained in the corresponding BAS formula.


\subsection{Proof of BAS formula}\label{sec:BASproof}

The proof proceeds by using an argument based on BCFW recursion~\cite{Britto:2005fq} in twistor space~\cite{Mason:2009afn}, applied to each tree sub-amplitude of $\mathbb{S}_{\bolh}[\rho|\omega]$, to establish that each of these sub-amplitudes is equivalent to a trivalent Feynman diagram contributing to the tree-level BAS scattering amplitude in the relevant colour-ordering. In this way, we establish that the formula \eqref{PT-BAS} captures the full tree-level Feynman diagram expansion of BAS theory. Our notation and method will closely follow~\cite{Cachazo:2012pz}. 

The relevant tree structure is apparent in \eqref{PT-BAS}, in which $S_{n, d}^{-1}[\hat{\rho}|\hat{\omega}]$ can be expanded as
\begin{equation}\label{eqtreedecomp}
    S_{n, d}^{-1}[\hat{\rho}|\hat{\omega}] = \mathcal{D}^{-1}(\hat{\omega}) \sum_{\substack{T \in \mathcal{BT}_{b \rho, b\omega}\\ \tilde{T} \in \mathcal{BT}_{a \bar{\rho},a \bar{\omega}}}} \frac{1}{\phi_{\mathrm{total}}}\frac{1}{ \phi_T}\times  \frac{1}{\tilde{\phi}_{\mathrm{total}}}\frac{1}{\tilde{\phi}_{\tilde{T}}}.
\end{equation} This decomposition is also illustrated in Figure \ref{fig:trees}. From this we define a tree sub-amplitude integrand to be, for the tree $\hat{T} = T \cup \tilde{T}$,
\begin{equation}
    S^{-1 \, \hat{T}}_{n, d}[\hat{\rho}|\hat{\omega}] = \mathcal{D}^{-1}(\hat{\omega}) \frac{1}{\phi_{\mathrm{total}}}\frac{1}{ \phi_T}\times  \frac{1}{\tilde{\phi}_{\mathrm{total}}}\frac{1}{\tilde{\phi}_{\tilde{T}}}.
\end{equation}
The proof will then proceed on each of these subamplitudes, and show that each one is equal to the equivalent trivalent Feynman diagram in BAS theory. The formula \eqref{PT-BAS} thus captures the full tree-level Feynman diagram expansion of the theory.

\begin{figure}
    \centering
    \includegraphics{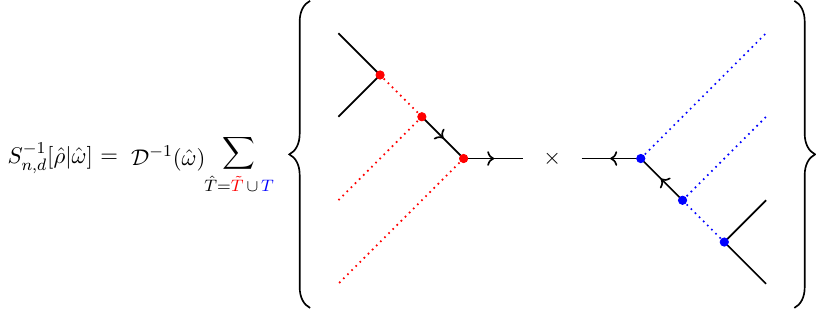}
    \caption{The tree decomposition of $S^{-1}_{n, d}[\hat{\rho}|\hat{\omega}]$, where the sum is over all trees compatible with the orderings on both $\bolth$ and $\bolh$. }
    \label{fig:trees}
\end{figure}

The factorization analysis is made easier by explicitly performing the moduli integrals for the $\mu^{\dot\alpha}(\sigma)$-components of the degree $d$ holomorphic map to twistor space in \eqref{PT-BAS}, evaluated on the momentum eigenstates \eqref{momeigsc}. This leads to the equivalent formula
\begin{multline}\label{PT-BASint*}
m_{n,d}[a\bar{\rho}b\rho|\bar{\omega}^{\mathrm{T}}ab\omega]=\int\frac{\d^{2(d+1)}\lambda}{\mathrm{vol}\,\GL(2,\C)}\,\mathcal{D}^{-1}(\hat{\omega})\,\bbS^{-1}_{\bolh}[\rho|\omega]\,\bbS^{-1}_{\bolth}[\bar{\rho}|\bar{\omega}] \\
\times\,\delta^{2(d+1)}\!\left(\sum_{i=1}^{n}t_i\,\tilde{\kappa}_i^{\dot\alpha}\,\sigma_i^{\mathbf{a}(d)}\right) \prod_{i=1}^{n}\D\sigma_i\,t_i\,\d t_i\,\bar{\delta}^{2}(\kappa_i-t_i\,\lambda(\sigma_i))\,,
\end{multline}
where an irrelevant overall factor of $(2\pi)^{2(d+1)}$ has been ignored. As the factorization analysis is also essentially local in the moduli space of holomorphic maps, it is also convenient to work in an affine patch $\sigma^{\ba}=(1,u)$ of $\P^1$, under which the remaining components of the holomorphic map are written
\be\label{affinelambda}
\lambda_{\alpha}(\sigma)=\lambda_{\alpha}(u)=\sum_{r=0}^{d}\lambda_{\alpha\,r}\,u^r\,,
\ee
with the set of $2(d+1)$ complex parameters $\{\lambda_{\alpha\,r}\}$ the remaining moduli of the map. The formula \eqref{PT-BASint*} in this representation becomes
\begin{multline}\label{PT-BASint}
m_{n,d}[a\bar{\rho}b\rho|\bar{\omega}^{\mathrm{T}}ab\omega]=\int\frac{\d^{2(d+1)}\lambda}{\mathrm{vol}\,\GL(2,\C)}\,\mathcal{D}^{-1}(\hat{\omega})\,\bbS^{-1}_{\bolh}[\rho|\omega]\,\bbS^{-1}_{\bolth}[\bar{\rho}|\bar{\omega}] \\
\times\,\prod_{r=0}^{d}\delta^{2}\!\left(\sum_{i=1}^{n}t_i\,\tilde{\kappa}_i^{\dot\alpha}\,u_i^r\right) \prod_{i=1}^{n}\d u_i\,t_i\,\d t_i\,\bar{\delta}^{2}(\kappa_i-t_i\,\lambda(u_i))\,,
\end{multline}
where all quantities are now evaluated on the affine patch, where $(i\,j)=u_i-u_j$ for all $i,j=1,\ldots,n$.


\subsubsection{The BCFW shift}

Consider a BCFW-deformation of the $n$-particle kinematics: this is a one-parameter complex deformation which preserves the on-shell condition for each external 4-momentum as well as overall 4-momentum conservation. Without loss of generality, we can take this shift to act on the momenta of particles $1,n\in\bolh$ as\footnote{The proof is qualitatively un-changed if we take the two shifted particles to be in $\bolth$. We expect that taking the shifted particles in different sets would yield the same result, with added complexity in the argument.}:
\begin{equation}\label{BCFWshift}
\kappa_1 \rightarrow \kappa_1 + z \,\kappa_n, \qquad \tilde{\kappa}_n \rightarrow \tilde{\kappa}_n - z\, \tilde{\kappa}_1\,,
\end{equation}
where $z\in\C$ is the shift parameter. Under this shift, the quantity $m_{n,d}$ becomes a rational function of $z$, $m_{n,d}(z)$. We say that such a shift is \emph{admissible} if $m_{n,d}(z)$, described by the formula \eqref{PT-BAS} under the shift \eqref{BCFWshift}, is analytic as $|z|\to\infty$ -- in other words, there is no pole at $z\to\infty$.

\paragraph{General shift behaviour:} Under the shift \eqref{BCFWshift}, the various constituents of the formula \eqref{PT-BASint} pick up dependence on the shift parameter $z$; for instance, the delta functions
\begin{equation} \label{eq:zshift1}
\bar{\delta}^2\!\left(\kappa_1 + z\, \kappa_n - t_1 \sum_{r = 0}^{d} \lambda_r\, u_1^r \right)\prod_{r = 0}^d \delta^{2}\!\left(\sum_{i = 1}^n t_i \tilde{\kappa}_i\, u_i^r - z\, t_n\, \tilde{\kappa}_1\, u_n^r \right)
\end{equation}
now have explicit $z$-dependence. The arguments of these delta functions are apparently divergent as $z\to\infty$, but this can be absorbed by introducing a new scaling parameter and affine coordinate on $\P^1$ for the particle 1, defined by
\begin{equation}\label{hatvardef}
\hat{t}_1 := t_1 - z\, t_n\,, \qquad \hat{t}_1\,\hat{u}_1^d := t_1\,u_1^d - z\, t_n\,u_n^d\,,
\end{equation}
which imply that
\be\label{uhatdef}
u_1=\left(\frac{\hat{t}_1\,\hat{u}_1^d+z\,t_n\,u_n^d}{\hat{t}_1+z\,t_n}\right)^{1/d}\,.
\ee
On the support of the (un-shifted) delta function
\be\label{ndelta}
\bar{\delta}^{2}\!\left(\kappa_n-t_n\,\sum_{r=0}^d \lambda_{r}\,u_n^r\right)\,,
\ee
in \eqref{PT-BASint}, it then follows that 
\begin{equation}
\kappa_1 + z\, \kappa_n - t_1 \sum_{r = 0}^{d} \lambda_r\, u_1^r=
\kappa_1 - \hat{t}_1 \sum_{r = 0}^d \lambda_r \left(\frac{d-r}{d}\, u_n^r + \frac{r}{d}\, \hat{u}_1^d\, u_n^{r-d} \right) + O(z^{-1})\,,
\end{equation}
and
\be
\sum_{i = 1}^n t_i \tilde{\kappa}_i\, u_i^r - z\, t_n\, \tilde{\kappa}_1\, u_n^r = \sum_{i=2}^{n}t_i\,\tilde{\kappa}_i\,u_i^r+\hat{t}_1\,\tilde{\kappa}_1\left(\frac{d-r}{d}\,u_n^r+\frac{r}{d}\,\hat{u}_1^d\,u_n^{r-d}\right)+O(z^{-1})\,.
\ee
In particular, the arguments of the delta functions \eqref{eq:zshift1} are finite in the $z\to\infty$ limit when expressed in terms of the new variables \eqref{hatvardef}.

One must now account for the additional $z$-dependence in \eqref{PT-BASint}. With the assumption that $1,n\in\bolh$, we can also take $1,n\neq b,t$ (where $b,t\in\bolh$ are the arbitrarily-chosen points singled out in the definition of $\mathcal{D}^{-1}(\hat{\omega})$ and $\phi_{ij}$), in which case only $\mathcal{D}^{-1}(\hat{\omega})$ and $\bbS^{-1}_{\bolh}$ in the first line of \eqref{PT-BASint} have any $z$-dependence. It is easy to see that
\begin{multline}\label{bigzD}
\mathcal{D}^{-1}(\hat{\omega})\xrightarrow{z\to\infty} \frac{(u_y-u_{\omega^*})\,(u_{\bar{\omega}^*}-u_t)}{(u_a-u_b)\,(u_{\omega^*}-u_{\bar{\omega}^*})\,(u_b-u_y)\,(u_a-u_t)\,(u_y-u_t)^2} \\
\times\,\prod_{j\in\bolth\setminus\{a,y\}}(u_n-u_j)^4\,\prod_{\substack{k\in\bolh\setminus\{1,n,b,t\} \\ l\in\bolth\setminus\{a,y\}}}(u_k-u_l)^2 + O(z^{-1})\,,
\end{multline}
while the $z$-dependence of $\bbS^{-1}_{\bolh}$ is inherited from that of the weights $\phi_{ij}$:
\be\label{phiijz}
\phi_{ij}\xrightarrow{z\to\infty}-t_i\,t_j\,[i\,j]\,(u_i-u_j)\prod_{l\in\bolth\setminus\{a,y\}}(u_i-u_l)\,(u_j-u_l)\,, \qquad i,j\neq 1,n\,,
\ee
\be\label{phi1jz}
\phi_{1j}\xrightarrow{z\to\infty}-z\,t_n\,t_j\,[1\,j]\,(u_n-u_j)\,\prod_{l\in\bolth\setminus\{a,y\}}(u_n-u_l)\,(u_j-u_l)+O(z^0)\,, \qquad j\neq n\,,
\ee
\be\label{phinjz}
\phi_{1j}\xrightarrow{z\to\infty}z\,t_n\,t_j\,[1\,j]\,(u_n-u_j)\,\prod_{l\in\bolth\setminus\{a,y\}}(u_n-u_l)\,(u_j-u_l)+O(z^0)\,, \qquad j\neq 1\,,
\ee
and
\be\label{phi1nz}
\phi_{1n}\xrightarrow{z\to\infty}-t_n\,\hat{t}_1\,[1\,n]\,u_n\left(\frac{\hat{u}_1^d}{u_n^d}+1\right)\,\prod_{l\in\bolth\setminus\{a,y\}}(u_n-u_l)^2+O(z^{-1})\,, \qquad j\neq n\,.
\ee
In addition, we must account for potential large-$z$ dependence elsewhere in the integration measure:
\be\label{du1z}
\mathrm{d}u_1 \xrightarrow{z\to\infty} \frac{\hat{t}_1}{z\, t_n}\, u_n^{1-d}\, \hat{u}_1^{d-1}\, \mathrm{d} \hat{u}_1 + O(z^{-2})\,,
\ee
and
\be\label{dt1z}
t_1\,\d t_1\xrightarrow{z\to\infty}z\,t_n\,\d\hat{t}_1+O(z^0)\,,
\ee
where terms which wedge to zero in \eqref{PT-BASint} have been dropped in the last expression. Finally, observe that as the shifted particles $1,n\in\bolh$, there is no $z$-dependence coming from the factor $\mathbb{S}^{-1}_{\bolth}$.

\paragraph{Admissible shifts:}
\begin{figure} 
\centering
\includegraphics{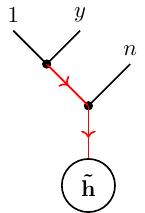} \qquad
\includegraphics{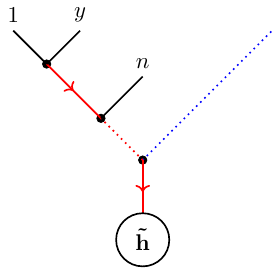} \qquad
\includegraphics{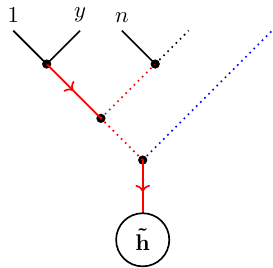}
\caption{Examples of the type of shift configuration for the points $1$ and $n$ that are admissible. Note that 1 has a unique neighbour $y$, and the next closest point is $n$. The red lines indicate a decay in $z^{-1}$ for the propagator, whilst the blue lines indicate that the lines go as $z^0$. Overall, the large-$z$ falloff due to these propagators will be at least $z^{-2}$, rendering the large-$z$ falloff of at least $m_{n,d}(z)\sim z^{-1}$. Note that we only consider cases with at least two points in $\bolh$, otherwise  we are already in the $\overline{\text{MHV}}$ sector which we can evaluate directly from \eqref{MHV-BAS1}.
}\label{fig:ads}
\end{figure}
Consider each binary tree over $\bolh$ in the sum appearing in \eqref{PT-BASint} separately. From \eqref{bigzD} -- \eqref{dt1z}, it follows that the tree-independent constituents of $m_{n,d}(z)$ contribute large-$z$ behaviour of the form $z^0$. The shift \eqref{BCFWshift} is \emph{admissible} for the given tree if the tree is such that it contains at least two factors of $\phi_{1i}^{-1}$ or $\phi_{ni}^{-1}$ for some $i\neq1,n$. Examples of trees, where $1$ has a unique nearest neighbour, are shown in Figure~\ref{fig:ads}. If a given binary tree in the sum is such that the shift on $1,n$ is not admissible, then we simply shift two other particles in $\bolh$ such that the shift is admissible for this tree.

Overall, this means that choosing an admissible shift for each binary tree in the sum gives large-$z$ behaviour
\be\label{treelzb}
m_{n,d}(z)\Big|^{\hat{T}}_{\mathrm{admissible}}\xrightarrow{z\to\infty}O(z^{-2})\,.
\ee
This means that, by choosing an admissible shift for each binary tree contributing to \eqref{PT-BASint}, there is no pole at $z\to\infty$ in the deformed quantity $m_{n,d}(z)$.


\subsubsection{Factorisation}

At this point, we restrict our focus to a single term in the sum over the set of binary trees $\mathcal{BT}_{b\rho,b\omega}$ on $\bolh$ and $\mathcal{BT}_{a\bar{\rho},a\bar{\omega}}$ on $\bolth$ compatible with the colour-orderings. Let this single term correspond to the binary tree $\hat{T}=T\cup\bar{T}$, and (without loss of generality) assume that the shift \eqref{BCFWshift} is admissible for this binary tree. We have established that the contribution of this term to the shifted expression \eqref{PT-BASint}, which we will denote by $m_{n,d}^{\hat{T}}(z)$, has no pole as $z\to\infty$. Our next task is to establish that the only other poles contained in $m_{n,d}^{\hat{T}}(z)$ correspond to certain multi-particle factorization channels.

To do this, we observe that $m_{n,d}^{\hat{T}}$ is, by definition, associated to a degree $d$ holomorphic rational curve in twistor space, and consider a degeneration of this underlying curve. If $m_{n,d}^{\hat{T}}$ develops poles in the modulus controlling this degeneration, then by standard arguments it will have a pole of the corresponding order in the shift parameter $z$ corresponding to a multi-particle factorization compatible with the curve degeneration.

The degeneration of the curve has a standard description by modelling the underlying $\P^1$, with homogeneous coordinates $\sigma^{\ba}=(1,u)$ in the chosen affine patch, as a conic in $\P^2$:
\begin{equation}
\Sigma_s = \{xy = s^2 z^2 \} \subset \mathbb{P}^2,
\end{equation}
where $s$ is a parameter and $[x,y,z]$ are homogeneous coordinates on $\P^2$, related to $\sigma^{\ba}$ by 
\begin{equation}
(x, y, z) = \left( (\sigma^{\mathbf{0}})^2, (\sigma^{\mathbf{1}})^2, \frac{\sigma^{\mathbf{0}}\, \sigma^{\mathbf{1}}}{s}\right)\,.
\end{equation}
As $s\rightarrow 0$, $\Sigma_s\cong\P^1$ degenerates into two components:
\begin{equation}
\lim_{s \rightarrow 0} \Sigma_s = \Sigma_L \cup \Sigma_R.
\end{equation}
These components are defined by $\Sigma_L = \{y = 0\}$ and $\Sigma_R = \{x = 0 \}$, with intrinsic homogeneous coordinates
\begin{equation}
\sigma^{\ba}_L = (z, x) = \sigma^{\mathbf{0}} \left(\frac{\sigma^{\mathbf{1}}}{s},\, \sigma^{\mathbf{0}}\right), \qquad \sigma^{\ba}_R = (z, y) = \sigma^{\mathbf{1}} \left(\frac{\sigma^{\mathbf{0}}}{s},\, \sigma^{\mathbf{1}}\right)\,.
\end{equation}
The affine coordinate $u=\sigma^{\mathbf{1}}/\sigma^{\mathbf{0}}$ on $\Sigma_s$ is related to the corresponding affine coordinates on $\Sigma_{L, R}$ by 
\begin{equation}
u_L = \frac{s}{u}, \qquad u_R = s\,u.
\end{equation}
The point at which the two components intersect in the degenerate $s\to0$ limit (which we will call the node $\bullet$) is fixed at $u_{L, R} = 0$ in both coordinate patches. 

Now, in the $s\to0$ limit, the $n$ marked points on $\Sigma_s$ are distributed among the two resulting components $\Sigma_L$ and $\Sigma_R$. Let $n_{L}$, $n_R$ denote the number of initial marked points distributed on $\Sigma_L$ and $\Sigma_R$, respectively, with $n_L+n_R=n$.
To understand the dependence of the integrand of $m_{n,d}^{\hat{T}}$ on $s$, it is useful to observe that holomorphic separations $(u_i-u_j)$ between marked points behave differently depending on where those points wind up in the degenerate limit. It follows that
\begin{equation}\label{udegens}
u_i - u_j = \begin{cases}
s \,\frac{u_{jL}- u_{iL}}{u_{iL}\,u_{jL}}\,, & i, j \in L \\
\frac{u_{iR} - u_{jR}}{s}\,, & i, j \in R \\
\frac{s^2 - u_{iL}u_{jR}}{s\,u_{iL}}\,, & i \in L, j\in R 
\end{cases}\,,
\end{equation}
in terms of the affine coordinates on $\Sigma_L$ and $\Sigma_R$. In addition, various other ingredients in $m_{n,d}^{\hat{T}}$ have standard behaviour in the degenerate limit which were determined in~\cite{Cachazo:2012pz}. The invariant measure 
\be\label{SL2meas}
\mathrm{d}\nu = \frac{1}{\mathrm{vol}\,\SL(2, \mathbb{C})} \prod_{i = 1}^n \d u_i\,,
\ee
on the marked point locations behaves as
\be\label{eq:measuretransfo}
\mathrm{d}\nu = s^{n_L - n_R - 4}\, \mathrm{d}s^2\, \mathrm{d}\nu_{L}\, \mathrm{d} \nu_R \, \prod_{i \in L} \frac{1}{u_{iL}^2}+O(s^{n_L-n_R-2})\,, 
\ee
where
\begin{equation}
\d \nu_{L} := \frac{1}{\mathrm{vol}\,\SL(2, \mathbb{C})} \prod_{i \in L\cup\bullet} \d u_{i L}\,, \qquad \d\nu_{R}:= \frac{1}{\mathrm{vol}\,\SL(2, \mathbb{C})} \prod_{i \in R\cup\bullet} \d u_{i R}\,.
\end{equation}
In particular, in the degenerate limit, the node $\bullet$ becomes a marked point, fixed by the choice of affine coordinate to be located at the origin in the affine patches on both $\Sigma_L$ and $\Sigma_R$.

The behaviour of the degree $d$ holomorphic map $Z:\Sigma_s\to\PT$ in the $s\to0$ limit is also easily deduced~\cite{Cachazo:2012pz}. To do this, a conveniently re-scaled set of map moduli $\{\mathcal{Z}_{r}\}$ can be defined by taking
\be\label{newmapmods}
\cZ_r\, = s^r\, U_{d_L - r}\,, \qquad \cZ_{\bullet} = U_{d_L}\,, \qquad \cY_r = s^r\, U_{d_L +r}\,,
\ee
for $d_L+d_R=d$, under which the holomorphic map admits a degeneration compatible with that of $\Sigma_s$ itself:
\begin{equation}\label{degenmap1}
\begin{split}
Z(u; s) &= \sum_{r = 0}^d U_r\, u^r \\
& = u^{d_L} \left( \sum_{r = 1}^{d_L}\cZ_r\, u^r_L + \cZ_{\bullet} + \sum_{b = 1}^{d_R} \cY_b \frac{s^{2b}}{u^b_L} \right) \quad u\in\Sigma_L \\
& = u^{d_L} \left( \sum_{b = 1}^{d_L} \cZ_b\, \frac{s^{2b}}{u^b_R} + \cZ_{\bullet} + \sum_{r = 1}^{d_R} \cY_r\, u^r_R \right) \quad u\in\Sigma_R\,.
\end{split}
\end{equation} 
In particular, note that as $s\to 0$ the holomorphic map $Z:\Sigma_s\to\PT$ degenerates into a degree $d_L$ holomorphic map on $\Sigma_L$ and a degree $d_R$ holomorphic map on $\Sigma_R$, up to the shared overall factor of $u^{d_L}$.

Implementing the re-scaling \eqref{newmapmods} at the level of the measure on the map moduli introduces explicit $s$-dependence:
\be\label{mapmodscale}
\frac{\d^{4(d+1)}U}{\vol\,\C^*}=s^{2d_L(d_L+1)}\,s^{2d_R(d_R+1)}\,\frac{\d^{4}\cZ_{\bullet}}{\vol\,\C^*}\,\prod_{r=1}^{d_L}\d^{4}\cZ_{r}\,\prod_{b=1}^{d_R}\d^{4}\cY_b\,,
\ee
which nearly represents a splitting of the measure into measures for two distinct maps of degree $d_L$ and $d_R$, respectively. The obstruction to this splitting is due to the fact that the node $\bullet\in\Sigma_{L,R}$ is mapped to the \emph{same} point $\cZ_{\bullet}\in\PT$ for both degenerate maps. The image of the node on the $\Sigma_R$ side can trivially be disentangled by inserting a factor of unity into the integrand:
\be\label{eq:nodeid}
1 = \int \mathrm{D}^{3}\cY_{\bullet}\,\mathrm{D}^3 Z\, \frac{\d v}{v}\,r^3\, \d r\, \bar{\delta}^4 (Z - v\, \cZ_{\bullet})\, \bar{\delta}^4 (Z - r\, \cY_{\bullet})\,.
\ee
The holomorphic delta functions here set 
\be\label{ydotident}
\cY_{\bullet}=\frac{v}{r}\,\cZ_{\bullet}\,,
\ee
so the point $\cY_{\bullet}\in\PT$ is (projectively) just another copy of the image of the node as $s\to0$. Observing that
\be\label{ydotmeas}
\D^{3}\cY_{\bullet}=\frac{\d^{4}\cY_{\bullet}}{\vol\,\C^*}\,,
\ee
including the factor of unity \eqref{eq:nodeid} allows us to completely factorize the measure on the map moduli into degree $d_L$ and $d_R$ components:
\be\label{mapmodscale2}
\D^{3}\cY_{\bullet}\,\frac{\d^{4(d+1)}U}{\vol\,\C^*}=s^{2d_L(d_L+1)}\,s^{2d_R(d_R+1)}\,\frac{\d^{4}\cZ_{\bullet}}{\vol\,\C^*}\,\prod_{r=1}^{d_L}\d^{4}\cZ_{r}\,\frac{\d^{4}\cY_{\bullet}}{\vol\,\C^*}\,\prod_{b=1}^{d_R}\d^{4}\cY_b\,,
\ee
as desired.

However, there are still some inconvenient overall factors in the degenerate maps as $s\to0$. In particular, at this stage
\be\label{ZmL1}
Z(u)|_{\Sigma_L}=u^{d_L}\left(\cZ_{\bullet}+\sum_{r=1}^{d_L}\cZ_r\,u_L^{r}+O(s^2)\right)\,,
\ee
\be\label{ZmR1}
Z(u)|_{\Sigma_R}=u^{d_L}\,\frac{v}{r}\left(\cY_{\bullet}+\sum_{b=1}^{d_R}\cY_b\,u_R^b+O(s^2)\right)\,,
\ee
where we have re-scaled $\cY_b\to(v/r)\,\cY_b$ to obtain an overall factor of $v/r$ on the $\Sigma_R$ component. These overall factors can be removed from the degenerate maps by exploiting the fact that all ingredients of $m_{n,d}^{\hat{T}}$ are homogeneous functions of $Z(u)$. By the ingredients
\be\label{urescales1}
\varphi_i(Z(u_i))\to\left\{\begin{array}{l}
                           u_i^{-2d_L}\,\varphi_i(Z(u_i)) \quad \mbox{ if } u_i\in\Sigma_L \\
                           u_i^{-2d_L}\,\frac{r^2}{v^2}\,\varphi_i(Z(u_i)) \quad \mbox{ if } u_i\in\Sigma_R
                           \end{array}\right.\,,
\ee
as well as
\be\label{urescales2}
\phi_{ij}\to\left\{\begin{array}{l}
                  (u_i\,u_j)^{-d_L}\,\phi_{ij} \quad \mbox{ if } u_i,u_j\in\Sigma_L \\
                  (u_i\,u_j)^{-d_L}\,\frac{r}{v}\,\phi_{ij} \quad \mbox{ if } u_{i/j}\in\Sigma_L,\:u_{j/i}\in\Sigma_R \\
                  (u_i\,u_j)^{-d_L}\,\frac{r^2}{v^2}\,\phi_{ij} \quad \mbox{ if } u_i,u_j\in\Sigma_R
                  \end{array}\right.\,,
\ee
\be
\tilde{\phi}_{ij}\to\left\{\begin{array}{l}
                           (u_i\,u_j)^{d_L}\,\tilde{\phi}_{ij} \quad \mbox{ if } u_i,u_j\in\Sigma_L \\
                           (u_i\,u_j)^{d_L}\,\frac{v}{r}\,\tilde{\phi}_{ij} \quad \mbox{ if } u_{i/j}\in\Sigma_L,\: u_{j/i}\in\Sigma_R \\
                           (u_i\,u_j)^{d_L}\,\frac{v^2}{r^2}\,\tilde{\phi}_{ij} \quad \mbox{ if } u_i,u_j\in\Sigma_R
                           \end{array}\right.\,,
\ee
as well as
\be\label{measrescale3}
\prod_{b=1}^{d_R}\d^{4}\cY_b\to\left(\frac{v}{r}\right)^{4d_R}\,\prod_{b=1}^{d_R}\d^{4}\cY_b\,,
\ee
the overall factors can be removed from the degenerate limit of the holomorphic map. That is, after these re-scalings it follows that
\be\label{degenmapfin}
Z(u_L)=\cZ_{\bullet}+\sum_{r=1}^{d_L}\cZ_r\,u_L^{r}+O(s^2)\,, \qquad Z(u_R)=\cY_{\bullet}+\sum_{b=1}^{d_R}\cY_b\,u_R^b+O(s^2)\,,
\ee
in the degenerate limit.

\medskip

\paragraph{Isolating singular contributions:} The goal is to determine the leading behaviour of the formula as $s\to0$. In particular, we are interested in contributions which are singular; that is, terms which have poles in $s$. All contributions which are regular as $s\to 0$ are irrelevant as they do not correspond to kinematic poles and therefore will not result in poles in the BCFW shift parameter. At this point, we can tightly constrain the parameters $d_{L,R}$ and $n_{L,R}$ underlying the degeneration by demanding that we look at the \emph{most} singular contributions.

In~\cite{Cachazo:2012pz}, it was established that
\be\label{CSformdegen1}
\frac{\mathrm{det}'(\HH)\,\mathrm{det}'(\HH^{\vee})}{\prod_{i=1}^{n}\d u_i}\sim s^{2d_L(d_L+1)}\,s^{2d_R(d_R+1)}\,s^{2d_L(n_R-n_L)}\,s^{n_R-n_L+2}\,,
\ee
as $s\to 0$. Now, we have shown that 
\be\label{CSformalt1}
\mathrm{det}'(\HH)\,\mathrm{det}'(\HH^{\vee})=\sum_{\substack{b\rho,b\omega\in\mathcal{S}(\bolh) \\ a\bar{\rho},a\bar{\omega}\in\mathcal{S}(\bolth)}}\fullpt_{n}[a\bar{\rho}b\rho] \, S_{n,d}[\rho,\bar{\rho}|\omega,\bar{\omega}]\,\fullpt_{n}[\bar{\omega}^{\mathrm{T}}ab\omega] \,,
\ee
and as it is straightforward to determine the behaviour of the Parke-Taylor factors as $s\to 0$, this allows us to infer the behaviour of $S_{n,d}$ and its inverse by comparison with \eqref{CSformdegen1}.

For a generic $\fullpt_n[\alpha]$ on $\Sigma_s$, the behaviour as $s\to 0$ depends on how many times the colour-ordering $\alpha$ crosses the node in the degenerate limit. Clearly, it must cross an even number of times, $2m_b$, where $m_b$ denotes the number of crossings or, equivalently, the number of colour traces which will appear on each of $\Sigma_L$ and $\Sigma_R$ as $s\to 0$ -- see Figure~\ref{fig:brtrace}. A straightforward calculation shows that
\be\label{PTdegen1}
\frac{\fullpt_n[\alpha]}{\prod_{i=1}^n\d u_i}\sim s^{n_R-n_L+2m_b}\,,
\ee
as $s\to 0$. Thus, letting $m_{B}$ denote the total number of colour-trace crossings for \emph{both} Parke-Taylor factors in \eqref{CSformalt1}, it follows that
\be\label{invkerscaling1}
S_{n,d}^{-1}[\rho,\bar{\rho}|\omega,\bar{\omega}]\sim s^{2m_B}\,s^{-2d_L(d_L+1)}\,s^{-2d_R(d_R+1)}\,s^{2d_L(n_L-n_R)}\,s^{n_R-n_L-2}\,,
\ee
upon comparison with \eqref{CSformdegen1}. 

Since $m_{n,d}^{T}$ goes like $S_{n,d}^{-1}$, the most singular contributions will be those with the \emph{fewest} colour-order crossings of the node; that is, with the smallest possible value of $m_B$. Now, configurations with $m_B=0$ are not allowed; we have already established that (without loss of generality) $1,n\in\bolh$ must be on opposite sides of the degeneration, and this forces each Parke-Taylor factor in \eqref{CSformalt1} to cross the degeneration at least once, meaning that $m_B\geq2$ with the most singular contributions corresponding to $m_B=2$. This can only happen if all points in $\bolth$ are on a single side of the degeneration; otherwise $m_B$ will be greater than two. See Figure~\ref{fig:split} for an example. In other words, if $\bolth$ is distributed across both $\Sigma_L$ and $\Sigma_R$ then this configuration cannot be the most singular as $s\to 0$.

So in order to study the most singular contributions to the formula, we can (again, without loss of generality) assume that all elements of $\bolth$ are located on $\Sigma_R$ as $s\to 0$. By definition, this fixes $d_L=0$ and $d_R=d$, and it is easy to see that this in turn fixes $n_L=2$, $n_R=n-2$: all other configurations will have vanishing support in terms of the various delta functions appearing in \eqref{PT-BASint}.

\begin{figure}
    \centering
    \begin{subfigure}[t]{0.45\textwidth}
    \centering
    \resizebox{0.9\textwidth}{!}{\includegraphics{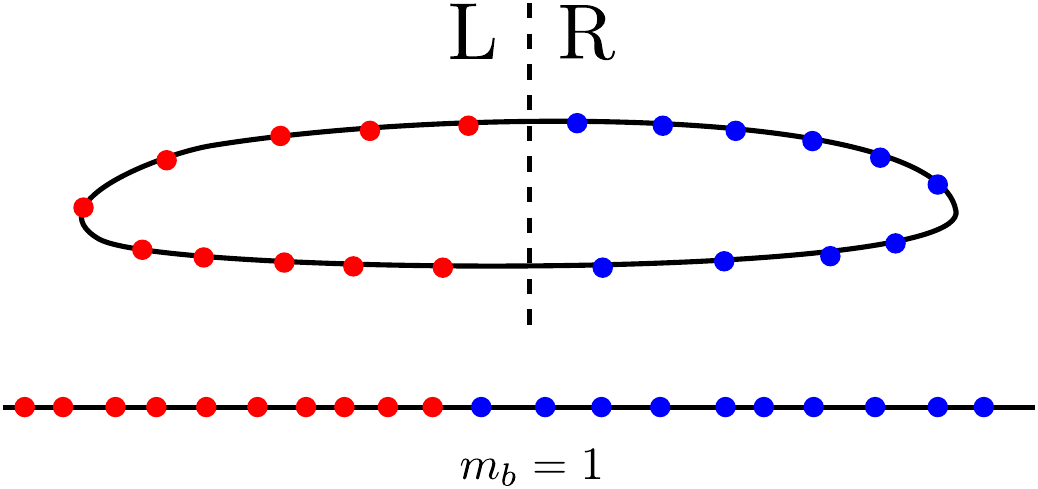}}
    \caption{The case in which the trace is broken the minimal number of times. In this case $m_b =1$.}
    \label{fig:ubtr}
    \end{subfigure}
    \qquad
    \begin{subfigure}[t]{0.45\textwidth}
    \centering
    \resizebox{0.9\textwidth}{!}{\includegraphics{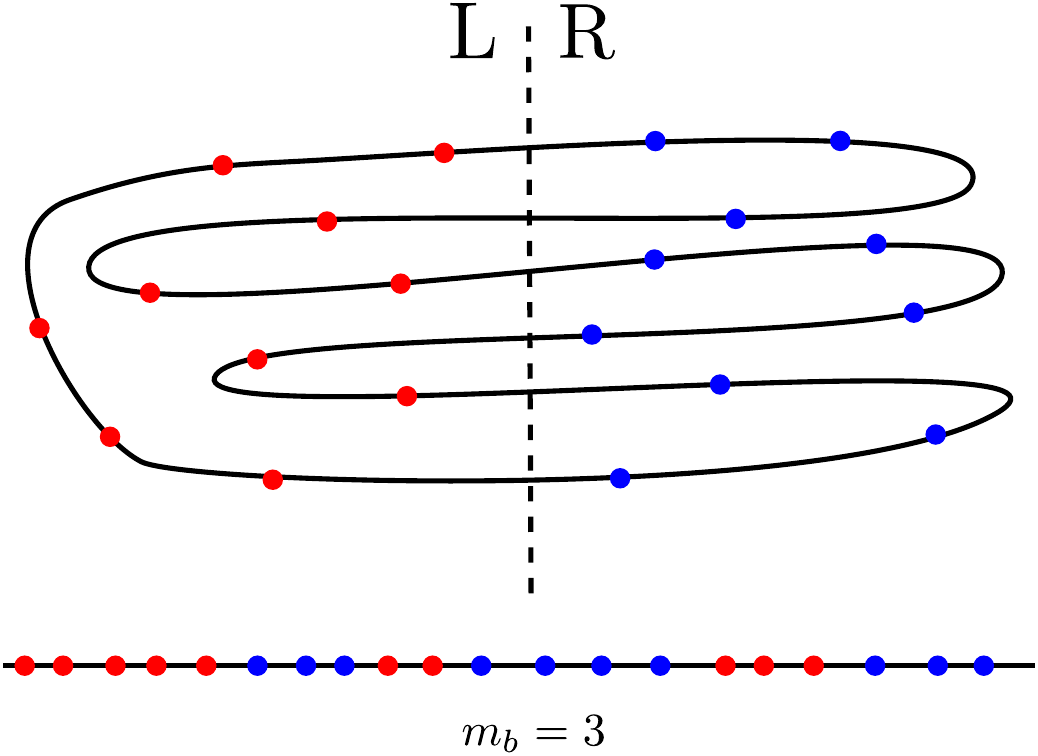}}
    \caption{For this colour ordering the trace is broken multiple times, and we end up with 3 trace components on each side, so $m_b = 3$.}
    \label{fig:brtr}
    \end{subfigure}
    \caption{Two examples of the possible trace configurations for a given distribution of points among the left and right sides, and their corresponding values of $m_b$. An easy way of counting $m_b$ is counting how many connected traces end up (e.g.) on the right side.}
    \label{fig:brtrace}
\end{figure}

\begin{figure}
    \centering
   \includegraphics{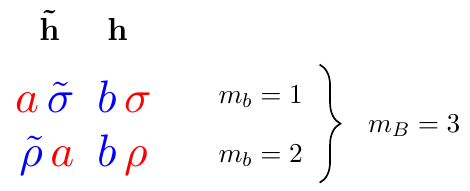}
    \caption{When $\bolth$ is distributed among the two sides the values of $m_B$ is always at least 3. Here this is illustrated in a simple configuration.}
    \label{fig:split}
\end{figure}

\medskip

Having established that the most singular contribution to the integrand occurs for the configuration where $n_L=2$ and $d_L=0$, we can now determine precisely how the entire integrand of $m_{n,d}^{\hat{T}}$ behaves in the degenerate limit. Collecting \eqref{eq:measuretransfo}, \eqref{mapmodscale2} and using \eqref{urescales1}, \eqref{degenmapfin}, it follows that
\begin{multline}\label{measurewfscale}
\frac{\d^{(4d+1)}U}{\vol\,\C^*}\,\d\nu\,\prod_{i=1}^{n}\varphi_i(Z(u_i))=s^{2d(d+1)-n}\,\d s^2\,\left[\d\nu_L\,\frac{\d^{4}\cZ_{\bullet}}{\vol\,\C^*}\,\prod_{i\in\bolh_L}\frac{\varphi_i(Z(u_{iL}))}{(u_{\bullet}-u_i)_L^2}\right]\left(\frac{v}{r}\right)^{4d-2n+4} \\
\times \left[\d\nu_{R}\,\frac{\d^{4}\cY_{\bullet}}{\vol\,\C^*}\,\prod_{b=1}^{d}\d^{4}\cY_{b}\,\prod_{j\in\bolh_R\cup\bolth}\varphi_j(Z(u_{jR}))\right]+O(s^{2d(d+1)-n+2})\,,
\end{multline}
where $\bolh_L\subset\bolh$ is the set comprising the two marked points which are left on $\Sigma_L$ as $s\to0$. From \eqref{invkerscaling1}, we also have that
\be\label{invkerscaling2}
S^{-1}_{n,d}[\hat{\rho}|\hat{\omega}]\sim s^{-2d(d+1)+n-2}+O(s^{-2d(d+1)+n})\,,
\ee
so that overall
\be\label{deglimscaling}
m_{n,d}^{\hat{T}}=\int \frac{\d s^2}{s^2}\left(1+O(s^2)\right)\,,
\ee
in the degenerate limit.

In particular, this establishes that the maximally singular $m_B=2$ configuration, with $d_L=0$ and $n_L=2$, is the \emph{only} singular configuration, having a simple pole in $s^2$. All other configurations will not contribute to the BCFW recursion; see Figure~\ref{fig:treefact} for an illustration.

\begin{figure}[h]
\centering
\includegraphics{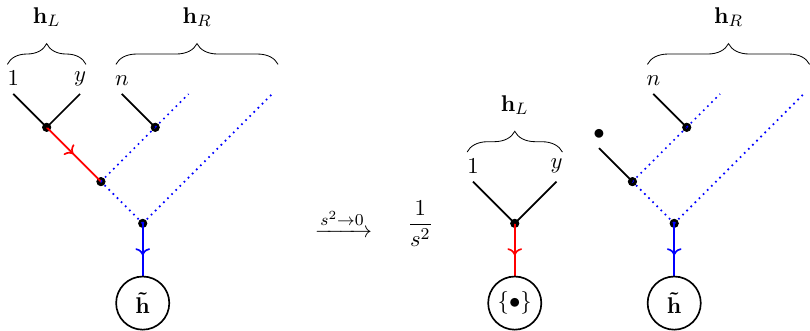}
\caption{The general decomposition of the product of propagators into propagators on the two halves. On the left, the red propagator takes contributions only from $\bolh_L$, whereas the leading (in $s^2$) contributions to the blue propagators come only from $\bolh_R$ on the right. In effect this `plucks' a branch from the tree.}
\label{fig:treefact}
\end{figure}

\medskip

\paragraph{Factorisation:} We can now focus entirely on this singular configuration, and determine the precise behaviour of each ingredient in $S^{-1\,\hat{T}}_{n,d}$ as $s\to 0$. Of the three factors making up $S^{-1\,\hat{T}}_{n,d}$, the easiest to consider is $\bbS^{-1}_{\bolth}[\bar{\rho}|\bar{\omega}]$, since this only depends on points in $\bolth\cup\{b,t\}$, all of which can be taken to lie on $\Sigma_R$ in the degenerate limit. Therefore, it follows from 
\be\label{tphidegscale}
\tilde{\phi}_{ij}=\frac{v^2}{r^2}\,s^{2d+3}\,\tilde{\phi}_{ij\,R}+O(s^{2d+5})\,,
\ee
that
\be\label{Stildedegscale}
\bbS^{-1}_{\bolth}[\bar{\rho}|\bar{\omega}]=\left(\frac{v}{r}\right)^{-2d}\,s^{-d(2d+3)}\,\bbS^{-1}_{\bolth}[\bar{\rho}|\bar{\omega}]_{R}\,,
\ee
where the subscript `$R$' on quantities such as $\tilde{\phi}_{ij\,R}$ or $\bbS^{-1}_{\bolth}[\bar{\rho}|\bar{\omega}]_{R}$ indicates that these are now defined in the affine coordinate patch on $\Sigma_R$. Observe that the scaling \eqref{Stildedegscale} is independent of $\bolh_L$, the distribution of points on $\Sigma_L$.

Next, we consider the pre-factor $\mathcal{D}^{-1}(\hat{\omega})$. Without loss of generality\footnote{It is a straightforward, if somewhat tedious, calculation to show that the same result is obtained with a generic distribution of these points.}, we can take all of the points $b,t,\omega^*\in\bolh$ to lie on $\Sigma_R$ in the limit $s\to 0$. In this case, we have
\begin{multline}\label{calDdegscale}
\mathcal{D}^{-1}(\hat{\omega})=\frac{s^{4-2(d-1)(n-d-3)}\,(u_{\omega^*}-u_y)_R\,(u_t-u_{\bar{\omega}^*})_R}{(u_a-u_b)_R\,(u_{\omega^*}-u_{\bar{\omega}^*})_R\,(u_b-u_y)_R\,(u_a-u_t)_R\,(u_y-u_t)^2_R} \\
\times\,\prod_{\substack{i\in\bolh_R\setminus\{b,t\} \\ j\in\bolth\setminus\{a,y\}}}(u_i-u_j)_R^2\,\prod_{j\in\bolth\setminus\{a,y\}}(u_\bullet-u_j)^2_R +\cdots\\
=s^{4-2(d-1)(n-d-3)}\,\mathcal{D}^{-1}(\hat{\omega})_R\,\prod_{j\in\bolth\setminus\{a,y\}}(u_\bullet-u_j)^2_R+\cdots\,,
\end{multline}
where the `$+\cdots$' indicates terms which are higher-order in $s$. 

Finally, we must account for the behaviour of $\bbS^{-1}_{\bolh}[\rho|\omega]$, which is the most complicated due to its dependence on $\bolh_L$. Using \eqref{udegens} and \eqref{urescales2}, one finds that
\begin{equation}\label{phidegenscal}
    \phi_{ij} = \begin{dcases}
        -s^{3-2d}\,t_it_j\,[i\,j]\,\frac{(u_i-u_j)_L}{(u_\bullet-u_i)_L\,(u_\bullet-u_j)_L}\prod_{l\in\bolth\setminus\{a,y\}}(u_{\bullet}-u_l)^2_R\,, & i,j\in L \\
        -\frac{s^{1-2d}\,r^2}{v^2}\,t_it_j\,[i\,j]\,(u_i-u_j)_R\,\prod_{l\in\bolth\setminus\{a,y\}}(u_i-u_l)_R\,(u_i-u_l)_R\,, & i,j\in L \\
        -\frac{s^{1-2d}\,r}{v}\,t_it_j\,[i\,j]\,(u_\bullet-u_j)_R\prod_{l\in\bolth\setminus\{a,y\}}(u_\bullet-u_l)_R\,(u_j-u_l)_R\,, & i\in L,\,j\in R
    \end{dcases}
\end{equation}
Recall that for a specified binary tree $T$
\begin{equation}\label{ShinvT}
\mathbb{S}^{-1\,T}_{\bolh}[\rho|\omega] = \frac{1}{\phi_{\mathrm{total}}} \prod_{E \in T} \frac{1}{\phi_E}\,,
\end{equation}
where 
\begin{equation}
\phi_E = \sum_{(i \rightarrow j) \in E} \phi_{ij}\,.
\end{equation}
Any $\bolh_L$ which is not connected with respect to $T$ (i.e., appearing as a leaf in the tree $T$) will involve more factors of $\phi_{ij}^{-1}$ with one index valued in $L$ and the other valued in $R$ and fewer factors of $\phi_{ij}^{-1}$ with both indices valued in $L$ in contrast to connected contributions (i.e., where $\bolh_L$ forms a leaf of $T$). 

Consequently, from \eqref{phidegenscal}, we see that such disconnected/non-leaf contributions will scale with a less singular power of $s$ as $s\to0$ than connected/leaf contributions; see Figure~\ref{fig:possconfigs}. However, we already established that the maximally-singular contribution to $m_{n,d}^{T}$ is in fact the only singular one, so to obtain the factorization behaviour of $\bbS^{-1}_{\bolh}$ corresponding to \eqref{invkerscaling2}, we must take the terms in \eqref{ShinvT} with the most singular scaling possible, corresponding to $\bolh_L$ being a single leaf of $T$, as illustrated in Figure~\ref{fig:treefact}. This means that on $\Sigma_L$, the node $u_{\bullet L}$ is treated as being the only member of $\bolth_L$, while on $\Sigma_R$ the node $u_{\bullet R}$ is treated as an additional element of $\bolh_R$. 

\begin{figure}
    \centering
    \includegraphics{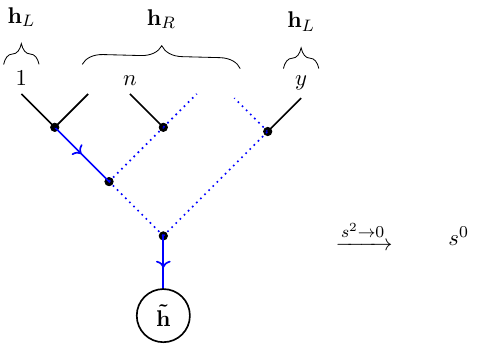}
    \caption{Example of a configuration where we have to split up the left and right leaves in the sub-amplitude, so that 1 and $y$ are no longer adjacent. Note that this now has more blue propagators (scaling as $s$ overall in the contribution from this factorisation channel, compared to $s^{-1}$ coming from red propagators) than the unbroken configuration, and therefore will not create a pole (as all other contributions remain unchanged).}
    \label{fig:possconfigs}
\end{figure}

Explicitly, this leads to
\be
    \bbS^{-1\,T}_{\bolh}[\rho|\omega] = s^{d(2n-2d-3)-n}\left(\frac{v}{r}\right)^{-2(d-n+2)} \prod_{j \in \bolth \setminus \{a, y\}} (u_{\bullet}-u_j)_R^{-2}\,\bbS^{-1\,T_L}_{\bolh_L}\,\bbS^{-1\,T_R}_{\bolh_R\cup\{\bullet_R\}}+\cdots\,, \label{nontildfactorise}
\ee
where $\bbS^{-1\,T_R}_{\bolh_R\cup\{\bullet_R\}}$ is defined in the obvious way on $\Sigma_R$ and $\bbS^{-1\,T_L}_{\bolh_L}$ is defined in terms of
\begin{equation}
    \phi_{ij}^{L} := -t_i\,t_j\,[i\,j]\, \frac{(u_i-u_j)_L}{(u_i-u_\bullet)_L\,(u_j-u_\bullet)_L}\,, \label{eq:zerod}
\end{equation}
for $\bolh_L=\{i,j\}$. The binary tree $T_L$ is composed of a single cubic vertex with external edges $\{i,j\}=\bolh_L$ (where one of $i,j$ is $1$ or $n$) and rooted at $\bullet_L$. It then immediately follows that
\be\label{SLexpression}
\bbS^{-1\,T_L}_{\bolh_L}=-\frac{(u_i-u_{\bullet})_L\,(u_j-u_\bullet)_L}{t_i t_j\,[i\,j]\,(u_i-u_j)_L}\,,
\ee
from the extremely simple structure of $T_L$.

\medskip

\paragraph{Collecting factors:} At this stage, we have isolated the leading behaviour of all of the ingredients of $m_{n,d}^{\hat{T}}$ in the $s\to 0$ degeneration limit. Collecting \eqref{eq:nodeid}, \eqref{measurewfscale}, \eqref{Stildedegscale}, \eqref{calDdegscale} and \eqref{nontildfactorise}, one observes a remarkable cancellation among the various powers of $s$, $v$ and $r$ to leave:
\begin{multline}\label{totfact1}
m_{n,d}^{\hat{T}}=\int\frac{\d s^2}{s^2}\,\D^{3}Z\left[\frac{\d^{4}\cZ_{\bullet}}{\vol\,\C^*}\,\d\nu_{L}\,\bbS^{-1\,T_L}_{\bolh_L}\,v\,\d v\,\bar{\delta}^{4}(Z-v\,\cZ_{\bullet})\,\prod_{i\in\bolh_L}\frac{\varphi_i(Z(u_{iL}))}{(u_{\bullet}-u_i)_L^2}\right] \\
\times \left[\frac{\d^{4(d+1)}\cY}{\vol\,\C^*}\,\d\nu_R\,\mathcal{D}^{-1}_R\,\bbS^{-1\,T_R}_{\bolh\cup\{\bullet_R\}}\,\bbS^{-1\,\bar{T}}_{\bolth}\,r\,\d r\,\bar{\delta}^{4}(Z-r\,\cY_{\bullet})\,\prod_{j\in\bolh_R\cup\bolth}\varphi_j(Z(u_{jR}))\right]\,,
\end{multline}
where we have dropped all regular terms as $s\to 0$, suppressed explicit dependence on the two colour-orderings and
\be\label{Rmodmeas}
\d^{4(d+1)}\cY:=\d^{4}\cY_{\bullet}\prod_{b=1}^{d}\d^{4}\cY_{b}\,,
\ee
denotes the measure on the map moduli corresponding to the right-hand branch of the degeneration.

Now, observe that as a distribution
\be\label{weighteddelta1}
\int_{\C^*}v\,\d v\,\bar{\delta}^{4}(Z-v\,\cZ_{\bullet})=\bar{\delta}^{3}_{-2,-2}\!\left(Z,Z(u_{\bullet L})\right)\,,
\ee
where the right-hand-side is a holomorphic delta function, that is, a $(0,3)$-distribution on $\P^3$, with support where the node on $\Sigma_L$ is mapped to the point $Z\in\PT$ (i.e., $Z(u_{\bullet L})\propto Z$) which is homogeneous of weight $-2$ in both of its arguments. In \eqref{totfact1}, this distribution is integrated against both $\D^{3} Z$ and $\d u_{\bullet L}$. Consequently, viewed as a differential form
\be\label{weighteddelta2}
\bar{\delta}^{3}_{-2,-2}(Z,Z(u_{\bullet L}))\in\Omega^{0,2}(\PT_Z,\cO(-2))\otimes\Omega^{0,1}(\P^1_{\bullet_L},\cO(-2))\,,
\ee
where the line bundles in each factor are understood to be defined over the respective projective varieties.

Integration against $\d u_{\bullet L}$ reduces the distribution to a $(0,2)$-distribution with compact support on $\PT$. Together with the integral over $\D^3 Z$ appearing in \eqref{totfact1}, this defines an exact pairing with twistor representatives of massless scalars, which, by the Penrose transform, are classes in $H^{0,1}(\PT,\cO(-2))$. So given any $\varphi(Z)\in H^{0,1}(\PT,\cO(-2))$, the integral
\be\label{cohpairing}
\int_{\PT_Z\times\P^1_{\bullet_L}}\D^{3}Z\,\d u_{\bullet L}\,\bar{\delta}^{3}_{-2,-2}(Z,Z(u_{\bullet L}))\,\varphi(Z)\,,
\ee
is simply a complex number. Observing that, similarly,
\be\label{weighteddelta3}
\int_{\C^*}r\,\d r\,\bar{\delta}^{4}(Z-r\,\cY_{\bullet})=\bar{\delta}^{3}_{-2,-2}(Z,Z(u_{\bullet R}))\,,
\ee
it follows that the integral over $\D^3Z$ is playing the role of a sum over a complete basis of massless, on-shell scalar wavefunctions inserted at the nodes on $\Sigma_L$ and $\Sigma_R$, respectively.

Consequently, we can define the wavefunctions
\be\label{nodewfs}
\varphi_{\bullet_L}(Z(u_{\bullet L})):=\bar{\delta}^{3}_{-2,-2}(Z,Z(u_{\bullet L}))\,, \qquad \varphi_{\bullet_R}(Z(u_{\bullet R})):=\bar{\delta}^{3}_{-2,-2}(Z,Z(u_{\bullet R}))\,,
\ee
at each node. The factorisation behaviour \eqref{totfact1} can then be written as
\be\label{totfact2}
m_{n,d}^{\hat{T}}=\int \D^{3}Z\,\d s^2\left[m_{3,0}^{T_L}(\bolh_L\cup\{\bullet_L\})\,\frac{1}{s^2}\,m_{n-1,d}^{T_R\cup\bar{T}}(\bolh_R\cup\{\bullet_R\}\cup\bolth)+O(s^2)\right]\,,
\ee
where
\be\label{mL1}
m_{3,0}^{T_L}(\bolh_L\cup\{\bullet_L\})=\int\frac{\d^{4}\cZ_{\bullet}}{\vol\,\C^*}\,\d\nu_L\,\prod_{i\in\bolh_L}\frac{1}{(u_i-u_{\bullet})_L^2}\,\bbS^{-1\,T_L}_{\bolh_L}\,\prod_{j\in\bolh_L\cup\{\bullet_L\}}\varphi_{j}(Z(u_{j L}))\,,
\ee
and 
\be\label{mR1}
m_{n-1,d}^{T_R\cup\bar{T}}(\bolh_R\cup\{\bullet_R\}\cup\bolth)=\int\frac{\d^{4(d+1)}\cY}{\vol\,\C^*}\,\d\nu_R\,\mathcal{D}^{-1}_R\,\bbS^{-1\,T_R}_{\bolh\cup\{\bullet_R\}}\,\bbS^{-1\,\bar{T}}_{\bolth}\,\prod_{i\in\bolh_R\cup\bolth\cup\{\bullet_R\}}\varphi_j(Z(u_{jR}))\,.
\ee
As the notation suggests, these two objects are simply the formula \eqref{PT-BAS} for the contributions to the BAS tree amplitudes corresponding to the binary rooted trees $T_L$ and $T_R\cup\bar{T}$, respectively. In the case of $m_{n-1,d}^{T_R\cup\bar{T}}$ given by \eqref{mR1}, this is obvious by comparison with the ingredients of \eqref{PT-BAS}, but for $m_{3,0}^{T_L}$ given by \eqref{mL1} the correspondence is less clear, as this expression is defined for a degree-zero map, whereas \eqref{PT-BAS} is defined for $d\geq1$. In fact, $m_{3,0}^{T_L}$ is equal to the 3-point amplitude of BAS theory, as can be verified by evaluating on momentum eigenstate wavefunctions.

\begin{lemma}\label{lemma:BAS3pt}
The quantity $m_{3,0}^{T}$, for $T$ the trivial cubic graph with a single vertex, is equal to the 3-point tree amplitude of BAS theory, in the sense that
\be\label{BAS3p0}
m_{3,0}^{T}=\delta^{4}\!\left(\sum_{i=1}^{3}k_i\right)\,,
\ee
when evaluated on momentum eigenstates.
\end{lemma}

\proof Without loss of generality, let $T$ be the trivial cubic graph whose external legs are labeled by 1,2 and 3, and evaluate \eqref{mL1} on twistor momentum eigenstate representatives for the external wavefunctions. In this case, \eqref{mL1} is given by
\begin{multline}\label{BAS3p1}
m_{3,0}^{T}=\frac{1}{[1\,2]}\int\frac{\d^{4}\cZ_{\bullet}}{\vol\,\C^*}\,\frac{\d u_1\,\d u_2\,\d u_3}{\vol\,\SL(2,\C)}\,\frac{t_3\,\d t_1\,\d t_2\,\d t_3}{(u_1-u_2)\,(u_2-u_3)\,(u_3-u_1)} \\
\times\,\prod_{i=1}^{3}\bar{\delta}^{2}(\kappa_i-t_i\,\lambda_{\bullet})\,\e^{\im\,t_i\,[\mu_{\bullet}\,i]}\,.
\end{multline}
The $\SL(2,\C)$ freedom in the integral can be used to eliminate all dependence on the marked points $(u_1,u_2,u_3)$ with unit Jacobian, and the $\C^*$ freedom can be used to fix $t_3=1$ with unit Jacobian. The integrals over $\d^{4}\cZ=\d^{2}\mu_{\bullet}\,\d^{2}\lambda_{\bullet}$ can then be performed against the exponential and delta functions to leave
\be\label{BAS3p2}
\frac{1}{[1\,2]}\int\d t_1\,\d t_2\,\delta^2(\kappa_1-t_1\,\kappa_3)\,\delta^2(\kappa_2-t_2\,\kappa_3)\,\delta^{2}(t_1\,\tilde{\kappa}_1+t_2\,\tilde{\kappa}_2+\tilde{\kappa}_3)\,,
\ee
ignoring an irrelevant overall factor of $(2\pi)^2$. The two scaling integrals can now be performed against the final set of delta functions to give
\be\label{BAS3p3}
-\frac{1}{[1\,2]^2}\,\delta^2\!\left(\kappa_1+\frac{[2\,3]\,\kappa_3}{[2\,1]}\right)\,\delta^2\!\left(\kappa_2+\frac{[1\,3]\,\kappa_3}{[1\,2]}\right)=\delta^{4}\!\left(\sum_{i=1}^{3}k_i\right)\,,
\ee
as claimed. \qed

\medskip

In other words, as the notation suggest, $m_{3,0}^{T_L}$ is indeed equal to the tree-level BAS 3-point amplitude, simply written in terms of a degree zero (as opposed to a degree one) holomorphic rational map to twistor space. As \eqref{totfact2} captures all of the singular behaviour as $s\to0$, it follows that
\be\label{totfact3}
\underset{s\to0}{\mathrm{Res}}\:m_{n,d}^{\hat{T}}(\bolh\cup\bolth)=\sum_{\mathrm{sing.}\:\mathrm{trees}}\int\D^{3}Z\,m_{3,0}^{T_L}(\bolh_L\cup\{Z\})\,\,m_{n-1,d}^{T_R\cup\bar{T}}(\bolh_R\cup\{Z\}\cup\bolth)\,,
\ee
where the sum is over all decompositions of $T$ into $T_L$ and $T_R$ containing the simple pole in $s^2$ as $s\to0$. As showed above, the only such decompositions correspond to $T_L$ being a cubic tree graph with a single vertex -- see Figure~\ref{decompFig}.

This process can then be repeated on $m_{n-1,d}^{T_R\cup\bar{T}}$, splitting off further cubic tree graphs until the elements of $\bolh$ are exhausted. At this point, we have a $\overline{\text{MHV}}$ amplitude on $\bar{T}$ with two points remaining in $\bolh$, which we know is equal to the corresponding BAS amplitude via \eqref{MHV-BAS1} and `partity symmetry' (shown in Appendix \ref{parityApp}). We have therefore decomposed $m_{n,d}^{T}$ totally into single-vertex cubic graphs on $\bolh$ with each decomposition taking the form of \eqref{totfact3}, and the expected Feynman diagram contribution from the tree on $\bolth$, viewed as a $\overline{\text{MHV}}$ amplitude. 

Now, with the starting expression $m_{n,d}^{\hat{T}}$ evaluated on momentum eigenstates, it is known that factorisation of the form \eqref{totfact3} in twistor space is equivalent to a multi-particle factorization channel in momentum space~\cite{Arkani-Hamed:2009hub,Bullimore:2009cb,Mason:2009sa,Skinner:2010cz}. Thus, we have established that $m_{n,d}^{\hat{T}}$ contains all of the factorization channels of the Feynman diagram $\hat{T}$ contributing to the BAS tree amplitude $m_n$, with the correct residues when evaluated on the singularities of those channels. The same argument works for every other term in the sum over binary trees in \eqref{eqtreedecomp}. It is easy to see that there are no other singularities in the formula -- we have already established that there are no poles for large values of the BCFW shift parameter, and other spurious poles would spoil soft limits, which we check in Appendix~\ref{softApp} -- meaning that the formula \eqref{PT-BAS} has the same singularity structure as the BAS tree-amplitude. Therefore, it is in fact equal to the BAS tree-amplitude, as claimed in \eqref{PT-BAS2}, and obeys the same BCFW recursion relation.

\begin{figure}
\resizebox{0.9\textwidth}{!}{%
    \includegraphics{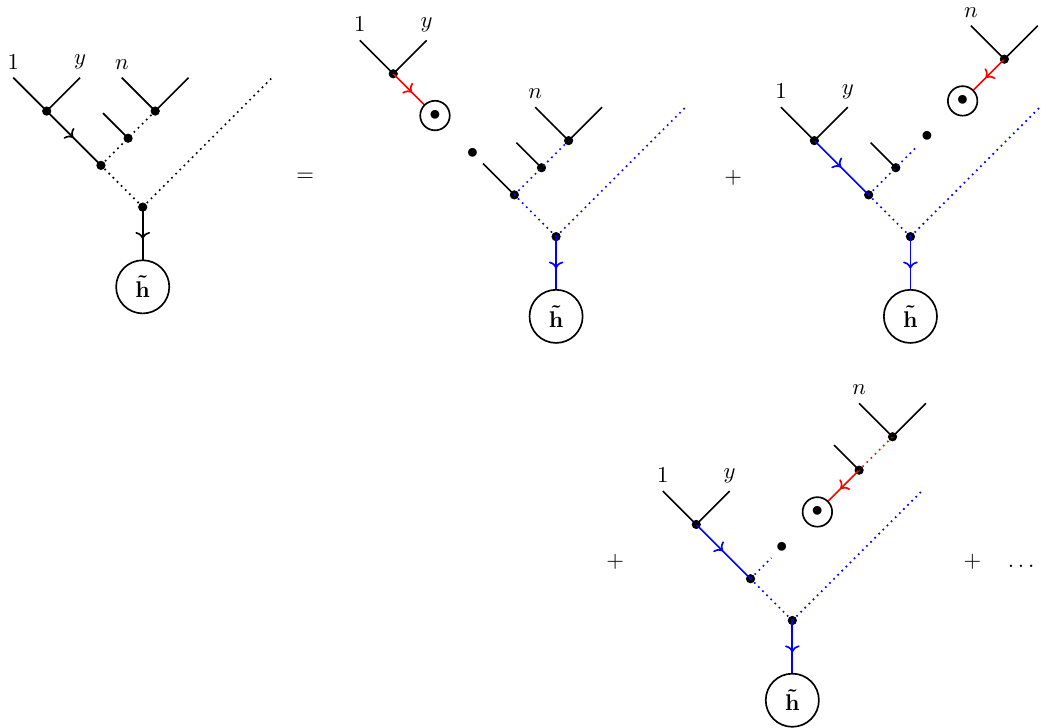}
}
\caption{Possible poles picked up by a tree subamplitude . Note that only the first line contributes because the rest contain $d=0$ amplitudes with more than 3 external particles, which evaluate to zero.}
    \label{decompFig}
\end{figure}


\section{Beyond flat spacetime}\label{sec:bg}

The astute reader will have observed that the graph theoretic arguments used in Section~\ref{sec:deriv} to derive the integral kernel on twistor space can be applied to \emph{any} formula based on determinants which encode a weighted counting of tree graphs via the matrix-tree theorem. Furthermore, twistor theory can be used to describe any solution to the vacuum Einstein equations with self-dual Weyl curvature and vanishing trace-free Ricci curvature~\cite{Penrose:1976js,Ward:1980am}. Consequently, one can ask if there are integral formulae for graviton scattering amplitudes in curved spacetimes which also admit an integral kernel representation indicative of double copy.

In fact, there are two known generalisations of the Cachazo-Skinner formula which also have this character. The first~\cite{Adamo:2015ina} is an adaptation of the formula to govern holomorphic maps from the Riemann sphere to the twistor space of AdS$_4$, aiming to describe tree-level graviton boundary correlation functions\footnote{We emphasize that although this formula is mathematically well-defined, its precise relationship to graviton boundary correlators in AdS$_4$, computed in positions space and for the Poincar\'e patch, remains unclear.}. The second~\cite{Adamo:2020syc,Adamo:2022mev} is a formula for the tree-level graviton scattering amplitudes in a class of chiral, asymptotically flat curved spacetimes known as self-dual radiative spacetimes.

In this section, we present the integral kernels that can be derived for these formulae; this is achieved using identical methods to those explained in detail in Section~\ref{sec:deriv}, so we gloss over most of the details. The resulting integral kernel defines a representation of the formula which is suggestive of a double copy, and we comment on potential difficulties with this interpretation in each case.


\subsection{AdS}

The metric of four-dimensional anti-de Sitter space, AdS$_4$, can be written as
\be\label{affinMAdS}
\d s^2=\frac{\d x_{\alpha\dot\alpha}\,\d x^{\alpha\dot\alpha}}{(1+\Lambda\,x^2)^2}\,,
\ee
where $x^{\alpha\dot\alpha}$ are coordinates on an affine Minkowski spacetime and $\Lambda<0$ is the cosmological constant. These coordinates have the advantage of admitting a smooth flat space limit ($\Lambda\to0$), although the boundary of AdS$_4$ is the finite hypersurface $x^2=-\Lambda^{-1}$. In~\cite{Adamo:2015ina}, it was shown how to extend the Cachazo-Skinner formula to a well-defined expression governing holomorphic maps from the Riemann sphere to the twistor space of AdS$_4$ in coordinates \eqref{affinMAdS}. 

The resulting formula looks superficially identical to the Cachazo-Skinner formula \eqref{CSform}:
\be\label{AdS1}
\cM^{\Lambda}_{n,d}=\int\d\mu_d\,|\bolth|^8\,\mathrm{det}'(\HH_{\Lambda})\,\mathrm{det}'\!\left(\HH^{\vee}_{\Lambda}\right)\,\prod_{i\in\bolh}h_i(Z(\sigma_i))\,\prod_{j\in\bolth}\tilde{h}_j(Z(\sigma_j))\,,
\ee
with the main distinction being in the entries of the matrices $\HH_{\Lambda}$ and $\HH^{\vee}_{\Lambda}$, which are where the formula is explicitly sensitive to the cosmological constant. In particular, the AdS conformal structure is encoded on twistor space through an \emph{infinity twistor} and its inverse; for the coordinates \eqref{affinMAdS} these are given by
\be\label{AdSIT}
I^{AB}=\left(\begin{array}{cc}
             \epsilon^{\dot\alpha\dot\beta} & 0 \\
             0 & \Lambda\,\epsilon_{\alpha\beta}
             \end{array}\right)\,, \qquad I_{AB}=\left(\begin{array}{cc}
                                                       \Lambda\,\epsilon_{\dot\alpha\dot\beta} & 0 \\
                                                       0 & \epsilon^{\alpha\beta}
                                                       \end{array}\right)\,,
\ee
which can be used to define the skew-symmetric brackets on quantities valued in cotangent and tangent bundles of twistor space, respectively:
\be\label{AdSbrackets}
\left[A,\,B\right]:=I^{AB}\,A_{A}\,B_{B}\,, \qquad \left\la C,\,D\right\ra:=I_{AB}\,C^{A}\,D^{B}\,.
\ee
It is these brackets which define the entries of the matrices:
\be\label{AdSHMat}
\begin{split}
\HH_{\Lambda\,ij}&=\frac{\sqrt{\D\sigma_i\,\D\sigma_j}}{(i\,j)}\,\left[\frac{\partial}{\partial Z(\sigma_i)},\,\frac{\partial}{\partial Z(\sigma_j)}\right]\,, \qquad i,j\in\bolh\,, i\neq j\,, \\
\HH_{\Lambda\,ii}&=-\,\D\sigma_i\,\sum_{\substack{j\in\bolh \\ j\neq i}}\frac{1}{(i\,j)}\,\left[\frac{\partial}{\partial Z(\sigma_i)},\,\frac{\partial}{\partial Z(\sigma_j)}\right]\prod_{l\in\bolth}\frac{(j\,l)}{(i\,l)} \,, \quad i\in\bolh\,.
\end{split}
\ee
a $(n-d-1)\times(n-d-1)$ matrix as before, and
\be\label{AdSdHmat}
\begin{split}
\HH^{\vee}_{\Lambda\,ij}&=\frac{\left\la Z(\sigma_i),\,Z(\sigma_j)\right\ra}{(i\,j)}\,, \qquad i,j\in\bolth\,, i\neq j \\
\HH^{\vee}_{\Lambda\,ii}&=-\frac{\left\la Z(\sigma_i)\,\d Z(\sigma_i)\right\ra}{\D\sigma_i}\,, \qquad i\in\bolth\,,
\end{split}
\ee
a $(d+1)\times (d+1)$ matrix as before. The reduced determinants appearing in \eqref{AdS1} are defined in the same way as before, by \eqref{redHdet}, \eqref{reddHdet}.

As in Minkowski space, the reduced determinant $\mathrm{det}'(\HH^{\vee}_{\Lambda})$ can still be interpreted as a resultant, but now for the full map $Z:\P^1\to\PT_{\Lambda}$, ensuring that the entire expression only has support for holomorphic maps which land inside the twistor space
\be\label{AdSPT}
\PT_{\Lambda}=\left\{Z\in\P^3\,|\,I_{AB}\,Z^{B}\neq 0\right\}\,,
\ee
for $I_{AB}$ given by \eqref{AdSIT}. Observe that all of the ingredients in the formula \eqref{AdS1} have a smooth flat space limit, in the sense that when $\Lambda\to0$, the formula reduces to the original Cachazo-Skinner formula \eqref{CSform}\footnote{Observe that \eqref{AdS1} is polynomial in $\Lambda$, so it can also be viewed as a formula for graviton scattering in dS$_4$ by simply continuing to $\Lambda>0$.}.

\medskip

Now, by following the exact same steps as in Section~\ref{sec:deriv}, it is possible to re-write the formula \eqref{AdS1} as
\be\label{AdS2}
\cM^{\Lambda}_{n,d}=\sum_{\substack{b\rho,b\omega\in\mathcal{S}(\bolh) \\ a\bar{\rho},a\bar{\omega}\in\mathcal{S}(\bolth)}}\int\d\mu_{n,d}\,\cI_{n}^{\bolth}[a\bar{\rho}b\rho] \, S^{\Lambda}_{n,d}[\rho,\bar{\rho}|\omega,\bar{\omega}]\,\cI_{n}^{\bolth}[\bar{\omega}^{\mathrm{T}}ab\omega] \,\prod_{i\in\bolh}h_i\,\prod_{j\in\bolth}\tilde{h}_j\,,
\ee
in terms of an AdS$_4$ integral kernel
\be\label{AdSIK1}
S_{n,d}^{\Lambda}=\mathcal{D}(\hat{\omega})\,\bbS^{\Lambda}_{\bolh}[\rho|\omega]\,\bbS^{\Lambda}_{\bolth}[\bar{\rho}|\bar{\omega}]\,.
\ee
The prefactor $\mathcal{D}(\hat{\omega})$ is the same as before \eqref{genwfactor}, while the positive and negative helicity kernels have the same structure as \eqref{phmker} and \eqref{nhmker}, but defined in terms of
\be\label{AdSphi}
\phi_{ij}=(i\,j)\,\left[\frac{\partial}{\partial Z(\sigma_i)},\,\frac{\partial}{\partial Z(\sigma_j)}\right]\,\prod_{l\in\bolth\setminus\{a,y\}}(i\,l)\,(j\,l)\,, \quad i,j\in\bolh\,,
\ee
and
\be\label{AdStildephi}
\tilde{\phi}_{ij}=\frac{\left\la Z(\sigma_i),\,Z(\sigma_j)\right\ra}{(i\,j)}\,\prod_{k\in(\bolth\cup\{b,t\})\setminus\{i,j\}}\frac{1}{(k\,i)\,(k\,j)}\,, \quad i,j\in\bolth\,,
\ee
respectively.

Now, as Yang-Mills theory is classically conformally invariant, it follows that the functional form of the RSVW formula \eqref{RSVW} should be un-modified on AdS$_4$, with the only imprint of the conformal structure coming from boundary conditions arising from the moduli integrals associated with the new hypersurface at infinity in \eqref{AdSPT}. In particular, $\mathcal{I}_{n}^{\tilde{\bolg}}[\alpha]$ should still be the integrand for tree-level gluon boundary correlators in AdS$_4$. Taken at face value, \eqref{AdS2} seems to provide an explicit twistorial double copy on AdS$_4$ in terms of the integral kernel \eqref{AdSIK1}.

While this conclusion is certainly tantalizing, a few words of caution are in order. Firstly, we reiterate that the precise connection between \eqref{AdS1} or the RSVW formula in AdS$_4$ with graviton/gluon boundary correlators has not been established. While these formulae certainly know something about the tree-level boundary correlators in AdS, they could be missing pure boundary contributions such as those which contribute to correlators with all external particles of the same helicity~\cite{Maldacena:2011nz}. Secondly, we do not currently have any physical interpretation of the inverse of the AdS kernel \eqref{AdSIK1}. If this were a true manifestation of double copy, one would expect this inverse to encode information about BAS `scattering' in some background (possibly AdS), but this does not seem to be immediately obvious. Finally, we observe that other explorations of double copy in AdS using explicitly Witten diagram/boundary correlator computations in spacetime have found structures which seem much more complicated than \eqref{AdS2} (cf., \cite{Armstrong:2020woi,Alday:2021odx,Zhou:2021gnu,Diwakar:2021juk,Alday:2022lkk,Herderschee:2022ntr,Cheung:2022pdk,Drummond:2022dxd,Lee:2022fgr,Lipstein:2023pih,Armstrong:2023phb,Mei:2023jkb,Liang:2023zxo}), so an explicit matching between our formulae and position or momentum space expressions is required before drawing any conclusions.


\subsection{Self-dual radiative spacetimes}

A self-dual (SD) radiative spacetime is a complex, asymptotically flat spacetime which is a self-dual solution of the vacuum Einstein equations uniquely determined by its radiative characteristic data on complexified null infinity~\cite{Newman:1976gc,Hansen:1978jz,Ludvigsen:1981,ko1981theory}. Historically, SD radiative metrics were also called `heavens' or $\cH$-spaces. Being vacuum self-dual solutions, any SD radiative metric corresponds to a twistor space $\CPT$ obtained as a complex deformation of $\PT$ which preserves the holomorphic fibration over $\P^1$ and the (weighted) Poisson structure on its fibres~\cite{Penrose:1976js}. 

The radiative data defining a SD radiative spacetime is encoded by a spin- and conformally-weighted function of three variables, denoted $\tilde{\bigma}^0$, corresponding to the asymptotic shear of constant $u$ hypersurfaces in (complexified) Bondi coordinates~\cite{Sachs:1961zz,Jordan:1961,Newman:1961qr,Bondi:1962px,Sachs:1962wk}. Pulled back to twistor space, this object is~\cite{Eastwood:1982}
\be\label{aShear}
\tilde{\bigma}^0([\mu\,\bar{\lambda}],\,\lambda_{\alpha},\,\bar{\lambda}_{\dot\alpha})\,,
\ee
where $\bar{\lambda}_{\dot\alpha}$ is the complex conjugate of $\lambda_{\alpha}$ and $\tilde{\bigma}^0$ has holomorphic homogeneity $+1$ and anti-holomorphic homogeneity $-3$ under projective rescalings on twistor space. This data can be used to directly parametrize the complex structure on $\CPT$ in terms of a Dolbeault operator
\be\label{defCStr}
\bar{\nabla}=\dbar+\tilde{\bigma}^0\,\D\bar{\lambda}\,\bar{\lambda}^{\dot\alpha}\,\frac{\partial}{\partial\mu^{\dot\alpha}}\,,
\ee
which automatically obeys the integrability condition $\bar{\nabla}^2=0$.

Using a chiral sigma model formulation of the holomorphic curves in $\CPT$~\cite{Adamo:2021bej}, it has been possible to derive the MHV amplitude for tree-level graviton scattering in any SD radiative spacetime, and there is a conjectural formula for the complete tree-level S-matrix which passes several non-trivial tests~\cite{Adamo:2022mev}. This general formula for the N$^{d-1}$MHV graviton tree-amplitude on the SD radiative background is given by:
\begin{multline}\label{SDRamp1}
\cM_{n,d}=\sum_{t=0}^{n-d-3}\sum_{p_1,\ldots,p_t}\int\d\mu_{d}\,|\bolth|^8\,\mathrm{det}'(\HH^{\vee})\,\prod_{i\in\bolh}h_i(Z(\sigma_i))\,\prod_{j\in\bolth}\tilde{h}_j(Z(\sigma_j)) \\
\times\,\left.\left(\prod_{\m=1}^{t}\D\sigma_\m\wedge\D\bar{\lambda}(\sigma_\m)\,\frac{N^{(p_\m-2)}(\sigma_\m)}{p_\m!}\frac{\partial^{p_\m}}{\partial\varepsilon_\m^{p_\m}}\right)\mathrm{det}'(\cH)\right|_{\varepsilon=0}\,.
\end{multline}
Here, the sum over $t=0,\ldots,n-d-3$ indexes the number of tail contributions to the amplitude, where external gravitons scatter off the background curvature leading to graviton-background and background-background interactions. Each index $p_\m$, for $m=1,\ldots,t$ is summed over all values greater than two, and encodes how many times a given background insertion contributes to the amplitude. These background insertions appear through 
\be\label{SDRnews}
N^{(k)}(\sigma_\m):=-\left.\frac{\partial^{k+1}}{\partial u^{k+1}}\tilde{\bigma}^0(u,\lambda(\sigma_\m),\bar{\lambda}(\sigma_\m))\right|_{u=[\mu(\sigma_\m)\,\bar{\lambda}(\sigma_\m)]}\,,
\ee
the $k$-times differentiated news function of the SD radiative background.

While the matrix $\HH^{\vee}$ is unchanged from Minkowski space due to the self-duality of the background, $\HH$ is replaced in each term of \eqref{SDRamp1} by a $(n+t-d-1)\times(n+t-d-1)$ matrix $\cH$, which is then fed into a reduced determinant akin to \eqref{redHdet}. From this reduced determinant, one extracts only polynomials in the $\m^{\mathrm{th}}$ background insertion of order $p_\m$; this is the role played by the formal parameter $\varepsilon_\m$ which enters (at most linearly) in the entries of $\cH$~\cite{Adamo:2022mev}.

For our purposes, all that is required is the decomposition of $\mathrm{det}'(\cH)$ according to the matrix-tree theorem. This follows by the same methods as in flat space, to give
\be\label{cHdecomp1}
\mathrm{det}'(\cH)=\frac{1}{|\bolth|^2}\,\prod_{\substack{k\in\bolh\cup\mathbf{t} \\ l\in\bolth}}\frac{1}{(k\,l)^2}\sum_{\substack{T^{b} \\ \mathrm{spanning}\,\,\bolh\cup\mathbf{t}}}\prod_{(i\to j)}\frac{[K_i\,K_j]}{(i\,j)}\,\prod_{q\in\bolth}(i\,q)\,(j\,q)\,,
\ee
where $\mathbf{t}$ is the set of $\m=1,\ldots,t$ background insertions and 
\be\label{SDmomspin}
K_{i}^{\dot\alpha}=\left\{\begin{array}{l}
                         \im\,t_i\,H^{\dot\beta\dot\alpha}(U,\sigma_i)\,\tilde{\kappa}_{i\,\dot\beta} \quad i\in\bolh \\
                         \im\,\varepsilon_i\,H^{\dot\beta\dot\alpha}(U,\sigma_i)\,\bar{\lambda}_{\dot\beta}(\sigma_i) \quad i\in\mathbf{t}
                         \end{array}\right. \,,
\ee
for $H^{\dot\alpha\dot\beta}(U,\sigma)$ the holomorphic frame of the self-dual spinor bundle on the SD radiative spactime, pulled back to the degree $d$ holomorphic curve in $\CPT$:
\be\label{SDRholframe}
\dbar H^{\dot\alpha\dot\beta}(U,\sigma)=\bar{\lambda}^{\dot\alpha}(\sigma)\,\bar{\lambda}_{\dot\gamma}(\sigma)\,H^{\dot\gamma\dot\beta}(U,\sigma)\,N(\sigma)\,\D\bar{\lambda}(\sigma)\,.
\ee
Using Proposition~\ref{PTprop}, \eqref{cHdecomp1} can be written as
\begin{multline}\label{cHdecomp2}
\mathrm{det}'(\cH)=\frac{(b\,x)\,(b\,y)}{\left|\bolth\cup\{b\}\right|^2} \sum_{\rho^+,\omega^+\in\mathcal{S}(\bolh\cup\mathbf{t}\setminus\{b\})}\brpt_{n+t-d}[b\rho^+x]\,\brpt_{n+t-d}[b\omega^+y] \\
\times\sum_{T\in\mathcal{T}^{b}_{\rho^+,\omega^+}}\prod_{(i\to j)\in E(T)}\left([K_i\,K_j]\,(i\,j)\,\prod_{l\in\bolth\setminus\{x,y\}}\frac{(j\,l)}{(i\,l)}\right)\,,
\end{multline}
in terms of broken Parke-Taylor factors.

As it stands, this decomposition is not yet appropriate to defining a momentum kernel, as it involves orderings on the set $\bolh\cup\mathbf{t}$ which includes background insertions. To compensate for this, we can define a `correction' factor for any $(n+t-d)$-point broken Parke-Taylor factor which breaks it down to the desired $(n-d)$-point object:
\be\label{PTcorrection}
\brpt_{n+t-d}[b\rho^+x]:=\brpt_{n-d}[b\rho x]\,\mathcal{C}(b\rho^+x;b\rho x)\,,
\ee
where $b\rho$ on the right-hand-side is now an ordering on $\bolh$ alone.  By definition, the correction factor $\mathcal{C}(b\rho^+ x;b\rho x)$ will be homogeneous of weight $-2$ on $\P^1$ in each element of $\mathbf{t}$ but homogeneous of weight zero in its other arguments. For example,
\be\label{Correctionex}
\mathcal{C}(b123 t_1 x; b123x)=\frac{(3\,x)}{(3\,t_1)\,(t_1\,x)}\,,
\ee
where $t_1$ denotes a single background insertion. 

\medskip

This allows the graviton amplitude \eqref{SDRamp1} to be rewritten in terms of an integral kernel:
\be\label{SDradamp2}
\cM_{n,d}=\sum_{\substack{b\rho,b\omega\in\mathcal{S}(\bolh) \\ a\bar{\rho},a\bar{\omega}\in\mathcal{S}(\bolth)}}\int\d\mu_{n,d}\,\cI_{n}^{\bolth}[a\bar{\rho}b\rho] \, S^{N}_{n,d}[\rho,\bar{\rho}|\omega,\bar{\omega}]\,\cI_{n}^{\bolth}[\bar{\omega}^{\mathrm{T}}ab\omega] \,\prod_{i\in\bolh}h_i\,\prod_{j\in\bolth}\tilde{h}_j\,,
\ee
where the background-dependent integral kernel still admits a chiral splitting
\be\label{SDradik1}
S^{N}_{n,d}[\rho,\bar{\rho}|\omega,\bar{\omega}]=\mathcal{D}(\hat{\omega})\,\bbS_{\bolth}[\bar{\rho}|\bar{\omega}]\,\bbS^{N}_{\bolh}[\rho|\omega]\,,
\ee
with only the positive helicity factor $\bbS^{N}_{\bolh}$ sensitive to the background. Explicitly, the positive helicity factor of the integral kernel is given by
\begin{multline}\label{SDradik2}
\bbS^{N}_{\bolh}[\rho|\omega]=\sum_{t=0}^{n-d-3}\sum_{p_1,\ldots,p_t}\left(\prod_{\m=1}^{t}\D\sigma_\m\wedge\D\bar{\lambda}(\sigma_\m)\,\frac{N^{(p_\m-2)}(\sigma_\m)}{p_\m!}\frac{\partial^{p_\m}}{\partial\varepsilon_\m^{p_\m}}\right) \\
\sum_{b\rho^+,b\omega^+\in\mathcal{S}_{\rho,\omega}(\bolh\cup\mathbf{t})}\Bigg[\mathcal{C}(b\rho^+x;b\rho x)\,\mathcal{C}(b\omega^+y;b\omega y) \\
\left.\left.\sum_{T\in\mathcal{T}^b_{\rho^+,\omega^+}}\prod_{(i\to j)\in E(T)}[K_i\,K_j]\,(i\,j) \prod_{l\in\bolth\setminus\{x,y\}}(i\,l)\,(j\,l)\right]\right|_{\varepsilon=0}\,.
\end{multline}
The sum over $b\rho^+,b\omega^+\in\mathcal{S}_{\rho,\omega}(\bolh\cup\mathbf{t})$ denotes summing over all compatible extensions of the orderings $b\rho,b\omega$ to include the $t$ background insertions.

\medskip

The double copy interpretation of this formula is not immediately clear. First of all, the extension of the RSVW formula to a SD radiative gauge field background is given by~\cite{Adamo:2020yzi}
\be\label{SDradglu1}
\cA_{n,d}[\rho]=\int\d\mu_{d}\,|\tilde{\bolg}|^4\,\prod_{i=1}^{n}\frac{\mathsf{H}^{-1}(U,\sigma_{\rho(i)})\mathsf{T}^{\msf{a}_{\rho(i)}}\mathsf{H}(U,\sigma_{\rho(i)})\,\D\sigma_i}{(\rho(i)\,\rho(i+1))}\,\prod_{j\in\bolg}a_j\,\prod_{k\in\tilde{\bolg}}b_k\,,
\ee
where $\mathsf{H}(U,\sigma)$ denotes the holomorphic trivialization of a SD radiative gauge field pulled back to a curve in twistor space, defined by
\be\label{RGluHF}
\dbar\mathsf{H}(U,\sigma)=-\tilde{\cA}^{0}(\sigma)\,\mathsf{H}(U,\sigma)\,\D\bar{\lambda}(\sigma)\,,
\ee
where $\tilde{\cA}^{0}$ is the characteristic radiative data of the SD gauge field pulled back to twistor space. Clearly, there is no hint of these holomorphic frames appearing in the formula \eqref{SDradamp2}.

However, in the special case of Cartan-valued SD radiative gauge field backgrounds, then a double copy prescription \emph{is} possible. In this case
\be\label{RGLcart1}
\mathsf{H}^{-1}(U,\sigma_{i})\mathsf{T}^{\msf{a}_{i}}\mathsf{H}(U,\sigma_{i})=\exp\!\left[e_i\,g(U,\sigma_i)\right]\,\mathsf{T}^{\msf{a}_i}\,,
\ee
where $e_i$ is the root-valued charge of $\mathsf{T}^{\msf{a}_i}$ with respect to the Cartan subalgebra (where the background gauge field lives) and $g(U,\sigma)$ is defined by
\be\label{RGLcart2}
\dbar g(U,\sigma)=\tilde{\cA}^{0}(\sigma)\,\D\bar{\lambda}(\sigma)\,.
\ee
With this simplification, \eqref{SDradglu1} becomes
\be\label{SDradglu2}
\cA_{n,d}[\rho]=\int\d\mu_d\,|\tilde{\bolg}|^4\,\mathrm{PT}_n[\rho]\,\prod_{i\in\bolg}a_i\,\e^{e_i\,g(U,\sigma_i)}\,\prod_{j\in\tilde{\bolg}}b_j\,\e^{e_j\,g(U,\sigma_j)}\,,
\ee
and \eqref{SDradamp2} becomes a double copy formula when written as:
\be\label{SDradamp3}
\cM_{n,d}=\sum_{\substack{b\rho,b\omega\in\mathcal{S}(\bolh) \\ a\bar{\rho},a\bar{\omega}\in\mathcal{S}(\bolth)}}\int\d\mu_{n,d}\,\cI_{n}^{\bolth,e}[a\bar{\rho}b\rho] \, S^{N}_{n,d}[\rho,\bar{\rho}|\omega,\bar{\omega}]\,\cI_{n}^{\bolth,-e}[\bar{\omega}^{\mathrm{T}}ab\omega] \,\prod_{i\in\bolh}h_i\,\prod_{j\in\bolth}\tilde{h}_j\,,
\ee
where $\cI_{n}^{\bolth,e}$ denotes the integrand of \eqref{SDradglu2} with Cartan-valued charges $e=\{e_1,\ldots,e_n\}$. 

In particular, by taking the double copy between gluon amplitudes with sign-reversed background charges, the exponentials depending on the gauge field background in \eqref{SDradglu2} cancel, leaving \eqref{SDradamp2} as required. Precisely this charge-reversed prescription was found previously in constructing double copies for the only known examples of 3-point amplitudes on \emph{non-chiral} radiative backgrounds in gauge theory and gravity~\cite{Adamo:2017nia,Adamo:2020qru}. In this sense, \eqref{SDradamp2} does indeed have some double copy interpretation, but it seems strange that such a formula would correctly capture double copy for Cartan-valued SD radiative gauge backgrounds but not fully non-abelian ones. It is not known how to replace the charge-reversed double copy prescription when the gauge field background becomes non-abelian. Furthermore, we do not know how to invert the positive helicity integral kernel \eqref{SDradik2} as a linear map on the space of colour-orderings, akin to~\cite{Mizera:2016jhj}, meaning that we have no plausible link with BAS theory in any form, as would be expected for a truly robust statement of double copy.


\section{Conclusion} \label{sec:conc}

The question posed in the introduction of this paper was: is there a helicity-graded representation of all-multiplicity tree-level gluon and graviton scattering amplitudes which manifests double copy? More specifically, what is the double copy relation between the RSVW formula for gauge theory, and the Cachazo-Skinner formula for gravity? In this paper, we have definitively answers these questions. 

We found a natural double copy structure of the integrands, in the form of a helicity-graded integral kernel in Section \ref{sec:deriv}, derived through quite elementary graph theoretic and combinatorial results. We additionally proved that the inverse of this object should be viewed as the integrand for BAS theory in Section \ref{sec:proof}. A connection like this should be expected of a robust double copy structure. We extended the analysis further to formulae on non-trivial backgrounds -- though the interpretation of these expressions remains open.

\paragraph{Other double copies:} Whilst the derivation of this double copy structure followed a simple combinatorial path, its interpretation is a bit more mysterious. Comparing the structure and properties to other versions of the double copy -- field theory KLT or CHY formulae -- it seems to take a middle ground. At MHV level, our formula \eqref{PT-KLT} is equivalent to the KLT double copy with a certain choice of basis. However, for generic helicity, the formula expresses gravity as a sum over a strictly \emph{smaller} basis of colour orderings. In this way it could be viewed as being computationally `simpler' than KLT. Of course, this relative simplicity is compensated for by the helicity-graded solutions to the scattering equations implied by the integral. 

From the other perspective, we can compare our formula to the CHY double copy. Here, the double copy is naturally manifested as a multiplicative rule on a single term, evaluated on the support of the full set of solutions to the scattering equations. Since our formula only requires a helicity-graded subset of these equations, it is simpler in that sense. However, the integrand is much more complicated due to the presence of the sums over colour-orderings.  

It remains to be seen whether a closer analysis of the relationship between our integral kernel, momentum space KLT and the CHY formulae might provide clues towards new relations between solutions to the scattering equations. Additionally, it seems that the structure of the basis of colour-orderings is essential to the relation -- and encodes the degree of the maps, and thus the helicity configuration of the external legs. This fact -- also at the level of the BAS amplitude -- appears very mysterious. Why should the double copy care about helicity in this way? 

\medskip 

In relation to other versions of the double copy, one may wonder if there is a way to rephrase our results as some kind of colour-kinematic dual statement. Indeed, the basis of this work was the calculation of the colour-kinematics dual numerators arising from the Hodges formula by Frost in~\cite{Frost:2021qju}. There is also a question whether this structure is somehow related to a double copy at the level of twistor actions, as explored in~\cite{Borsten:2022vtg,Borsten:2023paw}, or more generally for Chern-Simons theories~\cite{Ben-Shahar:2021zww}. 

Finally, it would be remiss to not comment on approaches to double copy similar to ours in the literature. The first of these is the Cachazo-Geyer formula~\cite{Cachazo:2012da}, giving a proposal for all-multiplicity graviton amplitudes constructed out of RSVW integrands and the spacetime KLT kernel. It would be interesting to see whether this formula can be related to our new representation of the gravity amplitude via the double copy to give a sharp equivalence. This is not at all obvious, as our kernel is chirally spit and the colour-orderings are helicity graded.

Another double copy for the RSVW and Cachazo-Skinner integrands was proposed in~\cite{Cachazo:2016sdc} in the form of certain `scalar blocks'. It was noted in this paper that the objects constructed have unphysical poles, and must be summed over degree to give the bi-adjoint scalar amplitude. However, they do provide a \emph{minimal} basis for the double copy in the sense that N$^{d-1}$MHV amplitudes in gravity are given as a sum of $E(n -3, d-1)$ distinct objects built from two gauge theory factors. The object we have found is a bit different in nature, and captures the scalar properties in a more physical way. 

\medskip

\paragraph{BAS theory:} The new formula for the doubly colour-ordered BAS amplitude \eqref{PT-BAS} does not provide any calculational advantages for computing these amplitudes (the expansion in trees is precisely that of the usual Feynman diagrams, and there is an additional integral over maps). That being said, it is interesting that twistor space encodes the structure of propagators in this chiral way. It would be interesting to see whether there are any connections to the recent studies of double copy and BAS theory (or $\mathrm{tr}(\phi^3)$ theory) via positive geometry~\cite{Arkani-Hamed:2017mur,delaCruz:2017zqr,Cao:2021dcd,Cachazo:2022vuo,delaCruz:2023sgk,Arkani-Hamed:2024nhp}.

\medskip

\paragraph{Non-trivial backgrounds:} Our exploration of the double copy beyond flat space in Section \ref{sec:bg} suggests that the scalar theory associated to the kernel will be deformed in some way. The examples presented in this paper are definitely not exhaustive, and we expect that the true relations will arise from different configurations. But notably, with twistor space, we have one of the best handles on all-multiplicity data for certain backgrounds, and this may prove the playground to test the existence of double copy beyond flat space. Perhaps the deformation of the kernel in various theories can be related to the KLT bootstrap program~\cite{Chi:2021mio}, a generalisation of the Berends-Giele currents used in~\cite{Frost:2020eoa, Mafra:2020qst}, or the stringy version of~\cite{Mizera:2016jhj}. 

In the same vein, one should compare our construction to existing versions of the double copy on non-trivial backgrounds. The objects we consider might live in a different space, but the final result should be intuitively similar to these spacetime approaches~\cite{Adamo:2017nia,Adamo:2018mpq,Adamo:2020qru,Armstrong:2020woi, Alday:2021odx, Zhou:2021gnu, Diwakar:2021juk, Alday:2022lkk,Herderschee:2022ntr,Cheung:2022pdk,Drummond:2022dxd,Lee:2022fgr, Lipstein:2023pih,Armstrong:2023phb,Mei:2023jkb,Liang:2023zxo,Sivaramakrishnan:2021srm,Ilderton:2024oly}. In connection to CHY, perhaps our AdS formula can be related via a version of the scattering equations in (A)dS~\cite{Eberhardt:2020ewh,Roehrig:2020kck,Gomez:2021qfd,Gomez:2021ujt}.

\acknowledgments

SK would like to thank Jan Boruch, Wei Bu and Paul\'{i}na Smol\'{a}rov\'{a} for useful comments on aspects of this work. TA is supported by a Royal Society University Research Fellowship, the Leverhulme Trust grant RPG-2020-386, the Simons Collaboration on Celestial Holography MPS-CH-00001550-11 and the STFC consolidated grant ST/X000494/1. SK is supported by an EPSRC studentship.

\appendix

\section{`Parity' invariance of the twistor BAS amplitude} \label{parityApp}

Even though massless scalar particles do not have a helicity (i.e., all of them are helicity zero), the BAS amplitude \eqref{PT-BAS} is graded by the degree of the holomorphic map to twistor space, which is traditionally associated with the helicity of the external particles\footnote{In the BAS formula, degree is instead related to the structure of the colour ordering.}. A desirable property of the formula would be that it is `parity' symmetric in the sense that changing the degree of the map to $\tilde{d} = n - d - 2$ and interchanging angle and square brackets reproduces the same quantity. 

One can show this following methods similar to~\cite{Roiban:2004yf,Witten:2004cp,Bullimore:2012cn} by first noting that, when evaluated on the twistor momentum eigenstates \eqref{momeigsc} we have
\begin{equation}
\langle \lambda(\sigma_i)\, \lambda(\sigma_j) \rangle = \frac{\langle i \, j \rangle}{t_i\, t_j}\,,
\end{equation}
where the $t_i$ are the scaling parameters associated to each momentum eigenstate in twistor space. Now, change variables from $t_i$ to $s_i$ for each $i=1,\ldots,n$, where 
\begin{equation}
    s_i\, t_i := \prod_{j \neq i} \frac{1}{(i\,j)}\,,
\end{equation}
where the product is over \emph{all} insertion points $j\neq i$. As shown in~\cite{Roiban:2004yf}, this transforms the moduli integral over degree $d$ curves to a moduli integral over degree $n- d - 2$ curves, up to a factor
\begin{equation}
    \prod_{i = 1}^n  s_i^4  \,\prod_{j\neq i}\, (i\,j)^2 . \label{mapjacobi}
\end{equation}
At the same time, the measure in the wavefunctions transforms to
\begin{equation}
    t_i\, \d t_i = - \frac{\d s_i}{s_i^3} \,\prod_{j \neq i} \,\frac{1}{(i\,j)^2}\,.
\end{equation}
These two Jacobians combine to give the scaling parameter measure $s_i\,\d s_i$ appropriate to scalar momentum eigenstates on twistor space. Therefore, all integral measures are invariant up to a sign, and it remains to be checked that the integrand behaves in the same way. 

The prefactor $\mathcal{D}^{-1}$ will not change, and  is already symmetric in the exchange of $\bolh$ and $\bolth$. The transformation of the quantities $\phi_{ij}$, $\tilde{\phi}_{ij}$ is less trivial. Recall that
\begin{equation}
    \phi_{ij} = -  t_i t_j\, [i\, j] (i\,j) \prod_{l \in \bolth \setminus \{a, y\}} \,(i\,l)(j\,l)\,, 
\end{equation}
so changing variables to $s_i$ gives
\begin{align}
    \phi_{ij} &= - \frac{[i\,j]}{s_i s_j} (i\,j) \prod_{l \neq i} \,\frac{1}{(i\,l)} \prod_{k \neq j}\, \frac{1}{(j\,k)} \prod_{l \in \bolth\setminus \{ay\}} (i\,l)(j\,l)\\
    & = -  \frac{[i\,j]}{s_i s_j} (i\,j) \prod_{l \in \bolh \cup \{a, y\} \setminus \{i\}} \,\frac{1}{(i\,l)} \prod_{k \in \bolh \cup \{a, y\} \setminus \{j\}}\, \frac{1}{(j\,k)} \\
    & = -  \frac{[i\,j]}{s_i s_j (i\,j)} \prod_{l \in \bolh \cup \{a, y\} \setminus \{i, j\}} \,\frac{1}{(i\,l)(j\,l)}\,,
\end{align}
which, upon also transforming $[i\,j] \rightarrow \langle i \,j \rangle$, reproduces the structure of $ - \tilde{\phi}_{ij}$. 

A similar calculation for $\tilde{\phi}_{ij}$ yields
\begin{equation}
    \tilde{\phi}_{ij} = \langle i\, j \rangle \, s_i s_j (i\,j) \prod_{l \in \bolh \setminus \{b, t\}} \,(i\,l)(j\,l)
\end{equation}
under the reparametrisation from $t_i$ to $s_i$. Once again, transforming $\langle i\, j \rangle \rightarrow [i\, j]$ reproduces $- \phi_{ij}$ as desired. The overall minus sign cancels the minus sign in the twistor wavefunction. Thus, we have established that $\bar{m}_{n, d} = m_{n, n-d - 2}$. This fact is used in the proof of Theorem \ref{BASthm} to evaluate the N$^{n - 3}$MHV amplitudes at the end of the iterative argument. 


\section{Soft limits} \label{softApp}

In this section we verify that the proposed formula \eqref{PT-BAS} reproduces the soft limits of biadjoint scalar amplitudes. Multi-particle soft limits also follow easily as a corrollary of this derivation. These two results combined verify that there are no unphysical poles in the proposed formula, as these would present themselves in an anomalous multi-particle soft factor. This section follows closely the soft analysis in~\cite{Roiban:2004yf} and~\cite{Bullimore:2012cn}. 

We again analyse the amplitude via its tree sub-amplitudes \eqref{eqtreedecomp}, each of which will correspond to the equivalent Feynman diagram. We will analyse the limit of sending a leg $n \in \bolh$ soft, which, by the `parity' invariance shown in Appendix~\ref{parityApp}, is the same as sending a particle in $\bolth$ soft. In the spinor-helicity formalism, we can describe this soft limit as $\kappa_n^{\alpha} \rightarrow 0$ whilst $\tilde{\kappa}_n^{\dot{\alpha}}$ stays constant (the holomorphic soft limit), or as $\tilde{\kappa}_n^{\dot{\alpha}} \rightarrow 0$ whilst $\kappa_n^{\alpha}$ stays constant (the anti-holomorphic soft limit). In a scalar theory, the soft limit should not depend on this choice, so we will verify both.

\paragraph{Holomorphic soft limit:} Starting with the holomorphic soft limit, the only dependence on $n$ is in the associated twistor wavefunctions $\varphi_n (Z(\sigma))$ and the quantities $\phi_{n j}$. Therefore, we are interested in analysing the singularity structure of
\begin{equation} \label{softpre}
    \int \d t_n \, t_n \, \D\sigma_n \prod_{l \in \bolth \setminus \{a, y\}} (n\,l)^2  \Bigg(\prod_{I(E) \ni n} \frac{1}{\phi_E}\Bigg) \,\bar{\delta}^2 (\kappa_n - t_n \lambda (\sigma_n)) \, \e^{\im t_n \, [\mu(\sigma_n) \,n]}\,,
\end{equation}
where the second product appearing here is over all `propagators' in the primary tree that contain contributions from particle $n$, given by 
\begin{equation}
    \phi_E = \sum_{\{i, j\} \subset I(E)} \phi_{ij}, \qquad \phi_{ij} = - t_i t_j \,[i \,j]\, (i\,j) \prod_{l \in \bolth \setminus \{a, y\}} \,(i\,l)\,(j\,l)\,.
\end{equation}
In \eqref{softpre}, the integral over the scaling parameter $t_n$ can be performed against one of the holomorphic delta functions, leaving
\begin{equation} \label{softpre2}
    \int \frac{\D\sigma_n \, \langle a \, b\rangle }{\langle \lambda(\sigma_n) \, a \rangle ^2} \prod_{l \in \bolth \setminus \{a, y\}} (n\,l)^2 \, \prod_{E \ni n} \,\frac{1}{\phi_E}\,\,\bar{\delta}\!\left( \frac{\langle n \, b \rangle}{\langle n \, a \rangle} - \frac{\langle \lambda(\sigma_n)\, b \rangle}{\langle \lambda (\sigma_n)\, a \rangle} \right)\,  \e^{\im t_n \, [\mu(\sigma_n) \,n]}, \quad t_n \rightarrow \frac{\langle n \,a  \rangle }{\langle \lambda(\sigma_n)\, a \rangle}\,,
\end{equation}
where $\{a^\alpha,b^\alpha\}$ is an arbitrary. choice of (constant) spinor dyad used to decompose the two holomorphic delta functions.

Since $t_n$ is fixed to be of order $\kappa_n$ by this integration, terms coming from the exponential in \eqref{softpre2} can be neglected, as they will be subleading in the holomorphic soft expansion. On the support of the other holomorphic delta functions in the subamplitude (which we have not written explicitly as they do not contain singular behaviour in the soft limit), the argument of the remaining holomorphic delta function in \eqref{softpre2} can be rewritten as 
\begin{equation}
    \bar{\delta}\left( \frac{\langle n\, j \rangle \langle a \,b \rangle}{\langle n \, a \rangle \langle j \, a \rangle} - \Bigg[ \frac{\langle\lambda(\sigma_j)\, b\rangle}{\langle \lambda(\sigma_j)\, a \rangle} - \frac{\langle \lambda (\sigma_n)\, b \rangle}{\langle \lambda (\sigma_n) \, a \rangle} \Bigg]\right)\,, \label{deltafunc}
\end{equation}
using  $\langle j \,b \rangle /\langle j \, a\rangle - \langle \lambda (\sigma_j) \, b\rangle / \langle \lambda (\sigma_j)\, a \rangle=0$ for some $j\neq n$. The quantity in the brackets here is a rational function of $\sigma_n$ and vanishes when $\sigma_n \rightarrow \sigma_j$. Therefore, we can introduce a new variable $\omega:= (n\,j)$, and express the bracketed quantity as $\omega \,F(\omega, \sigma_j)$ for some function $F$ that is regular at $\omega = 0$. In terms of this $\omega$, the other contributions to \eqref{softpre2} are
\begin{equation}
    \phi_{n i}(\omega) = - \frac{ t_i \,\langle n \,a \rangle\, [n \,i]\, (\omega - (i\,j))  }{\langle \lambda (\sigma_n)\, a \rangle}\prod_{l \in \bolth \setminus \{a, y\}} (\omega - (j\,l))\,(i\,l)\,,
\end{equation}
and the prefactor
\begin{equation}
    \prod_{l \in \bolth \setminus \{a, y\}} (n\,l)^2 = \prod_{l \in \bolth \setminus \{a, y\}} (\omega - (j\,l))^2\,.
\end{equation}
The remaining integral in \eqref{softpre2} thus takes the form
\begin{equation}
  \int_{\mathbb{C}^{\ast}} \d \omega \,f(\omega) \,\bar{\delta}\!\left(\frac{\langle n\, j \rangle\,\langle a \, b \rangle}{\langle n \, a \rangle\, \langle j \, a \rangle}- \omega \,F(\omega, \sigma_j) \right)\,,  \label{softfinal}
\end{equation}
where 
\begin{equation}
    f(\omega) = \frac{\langle a \, b \rangle}{\langle \lambda (\sigma_n) \, a \rangle^2} \prod_{l \in \bolth \setminus \{a, y\}} (\omega - (j\,l))^2 \,  \prod_{E \ni n} \frac{1}{\phi_E(\omega)} \label{smallf}
\end{equation}
is a rational function of $\omega$. 

We now investigate the roots of the argument of the holomorphic delta-function in \eqref{softfinal} as $\langle n\, j \rangle \rightarrow 0$. One of these is when $\omega$ becomes small, of the same order as $\langle n\, j \rangle$, whilst the function $F(\omega, \sigma_j)$ remains of order unity. The other roots are when $F(\omega, \sigma_j)$ becomes small. For the latter case, performing the integral will not give us singular behaviour unless this root is also pole of $f(\omega)$. However, since $f(\omega)$ has poles determined by kinematic square brackets while $F(\omega, \sigma_j)$ only depends on kinematic angle brackets, for generic momentum configurations these will not have coinciding singular points. Additionally, due to the form of $F(\omega, \sigma_j)$ from \eqref{deltafunc}, there is no singularity coming from $1/\langle \lambda(\sigma_n) \, a \rangle$ in \eqref{smallf}.

Thus, we need only consider the case where $\omega$ becomes of order $\langle n \, j \rangle$. This leads to a singularity if $\phi_E = \phi_{nj}$ for one of the edges; that is, if the soft scalar $n$ attaches directly to an external $j$ line. Here
\begin{equation}
    \phi_{nj} = - \omega \, \frac{t_j \, \langle n \, a \rangle [n\,j] }{ \langle \lambda(\sigma_n)\, a \rangle} \prod_{l \in \bolth \setminus \{a, y\}} (\omega - (j\,l))(j\,l). 
\end{equation}
Notably we do not get a soft singularity if $n$ does not attach to an external line. If it does attach to an external line, the particle $j$ is specified \emph{a posteriori} as corresponding to this line\footnote{If one does not choose $j$ wisely in this way, we can not neglect roots where $F(\omega, \sigma_j)$ becomes small. This is because $1/\phi_{ni}$ for $i \neq j$ has a non-zero pole in $\omega$ that is independent of square brackets and can therefore coincide with a root of $F(\omega, \sigma_j)$.}.

Evaluating the integral \eqref{softfinal} in the limit $\langle n \, j \rangle \rightarrow 0$, we see that all dependence on $F(\omega, \sigma_j)$ drops out, and we can set $\omega = 0$ and $\sigma_n = \sigma_j$ in the remaining integrand.  The `soft factor' multiplying the lower-point tree subamplitude with leg $n$ removed is thus
\begin{equation}
    \frac{\langle \lambda (\sigma_j)\, a \rangle}{ t_j \, \langle j \, a \rangle\, \langle n\, j \rangle\, [n\, j]} = \frac{1}{2\,  k_n \cdot k_j}\,,
\end{equation}
on the support of the holomorphic delta-functions for particle $j$. Summing over all the possible trees, we recover the full soft structure of BAS theory at tree-level (cf., \cite{Abhishek:2020xfy}).

\medskip

\paragraph{Anti-holomorphic soft limit:} The analysis for the anti-holomorphic soft limit, where $\tilde{\kappa}^{\dot{\alpha}}_n \rightarrow 0$ whilst keeping $\kappa^{\alpha}_n$ constant, is much easier as the only dependence on $[n\,j]$ is trivial in $\phi_{nj}$, and is subleading in the twistor wavefunction. The only singularity is again when $n$ attaches directly to the external leg $j$, appearing through
\begin{equation}
    \frac{1}{\phi_{nj}} = - \frac{1}{t_n t_j \,[n \, j]\, (n\,j)} \prod_{l \in \bolth \setminus \{a, y\}} \,\frac{1}{(n\,l)\,(j\,l)}.
\end{equation}
The remaining $n$-dependence in $\phi_E^{-1}$ is subleading in the anti-holomorphic soft limit and therefore reduces to the product over edges in the reduced tree with leg $n$ removed. We are thus interested in the singularity structure of
\begin{equation}
    \int \d t_n \, t_n \, \D\sigma_n\, \prod_{l \in \bolth \setminus \{a, y\}} (n\, l)^2 \,\frac{1}{\phi_{nj}} \,\bar{\delta}^2 (\kappa_n - t_n \lambda (\sigma_n)) \, \e^{\im t_n \, [\mu(\sigma_n) \,n]}\,.
\end{equation}
Performing the integral over $t_n$ against one of the holomorphic delta functions and keeping track of the remaining $n$-dependence gives
\begin{gather}
    - \int \frac{\D\sigma_n}{ \langle \lambda(\sigma_n) \, a \rangle} \frac{\langle a \, b \rangle}{ t_j\, [n\,j]\, \langle n \, a \rangle\, (n\,j)} \prod_{l \in \bolth \setminus \{a, y\}} \frac{(n\,l)}{(i\,l)} \,\bar{\delta}\!\left( \frac{\langle n \, b \rangle}{\langle n \, a \rangle} - \frac{\langle \lambda(\sigma_n)\, b \rangle}{\langle \lambda (\sigma_n)\, a \rangle} \right)\,\e^{\im t_n \, [\mu(\sigma_n) \,n]}, \\ t_n \rightarrow \frac{\langle n \, a \rangle}{\langle \lambda(\sigma_n)\, a \rangle}\,, \label{antisoft1}
\end{gather}
as in \eqref{softpre2}, but without the extra factors coming from $\phi_E^{-1}$ which have already simplified in this anti-holomorphic soft limit. 

The exponential terms are subleading because of their dependence on $\tilde{\kappa}_n$. The remaining holomorphic delta function cannot be further massaged in this case as there are no small parameters in which to expand, but its argument can still be localized on the support of the holomorphic delta functions corresponding to particle $j$ (which we have not written explicitly). This can be used to set
\begin{equation}
    t_j = \frac{\langle j \, a \rangle}{\langle \lambda(\sigma_j) \, b \rangle}\,,
\end{equation}
for $\{a^{\alpha},b^{\alpha}\}$ again an arbitrary spinor dyad. Collecting factors and applying the Schouten identity to the argument of the holomorphic delta function gives
\begin{equation}
   -  \frac{1}{[n \, j ]} \int  \D\sigma_n \,\frac{\langle \lambda (\sigma_j) \, a\rangle }{ \langle j \, a \rangle \,(n\,j)} \prod_{l \in \bolth \setminus \{a, y\}} \frac{(n\,l)}{(j\,l)} \,\,\bar{\delta}(\langle \lambda (\sigma_n) \, n \rangle)\,.\label{preIBP}
\end{equation}
Recalling the definition \eqref{holdel}
\begin{equation}
    \bar{\delta}( \langle \lambda(\sigma_n) \, n \rangle ) \coloneqq \frac{1}{2 \pi \im} \, \,\bar{\del} \,\Bigg( \frac{1}{ \langle \lambda(\sigma_n) \, n \rangle} \Bigg)\,,
\end{equation}
we see that an integration-by-parts can be performed in \eqref{preIBP}.

In particular, note that given meromorphic functions $f$ and $g$ in $z$, it follows that
\begin{equation}
    \int_{\mathbb{C}^*} \d z \,  f(z) \, \bar{\delta} (g(z)) = \frac{1}{2 \pi \im} \int_{\mathbb{C}^{\ast}} \d z\wedge  \d \bar{z}\, \frac{\partial}{\partial \bar{z}} \Bigg(\frac{f(z)}{g(z)} \Bigg) - \int_{\mathbb{C}^{\ast}} \d z \, \frac{1}{g(z)} \, \bar{\delta} \Bigg( \frac{1}{f(z)}\Bigg)\,.
\end{equation}
The first term on the left-hand side is a boundary term, equal to the pole at $z = \infty$ of $\d z\,f(z)/g(z)$. If no such pole exists, then
\begin{equation}
    \int_{\mathbb{C}^{\ast}} \d z \,  f(z) \, \bar{\delta} (g(z)) = - \int_{\mathbb{C}^{\ast}} \d z \, \frac{1}{g(z)} \, \bar{\delta} \Bigg( \frac{1}{f(z)}\Bigg).
\end{equation}
A rational differential form does not have a pole at infinity if it behaves at least as $O(z^{-2})$ as $z \rightarrow \infty$. This is true for our integrand
\begin{equation}
    -\frac{  \langle \lambda (\sigma_j) \, a \rangle }{ \langle j \, a \rangle (n\,j)} \prod_{l \in \bolth \setminus \{a, y\}} \frac{(n\,l)}{(j\,l)} \,\,\ \frac{1}{\langle \lambda (\sigma_n) \, n \rangle}\,,
\end{equation}
with $\sigma_n$ playing the role of $z$, due to $\lambda(\sigma_n)$ being a polynomial of degree $d\geq1$ in $\sigma_n$, and $\bolth$ being a set of $d+1$ elements. 

So finally, \eqref{preIBP} is turned into
\begin{equation}
    \frac{1}{[n \, j ]} \int  \D\sigma_n \,\,\frac{\langle \lambda (\sigma_j) \, a\rangle }{ \langle j \, a \rangle \,\langle \lambda (\sigma_n) \, n \rangle} \prod_{l \in \bolth \setminus \{a, y\}} \frac{(n\,l)}{(j\,l)} \,\,\bar{\delta} ((n\,j))\,.
\end{equation}
The new holomorphic delta function, which emerges after the integration-by-parts, enables us to set $\sigma_n=\sigma_j$, and then on the support of the holomorphic delta functions for particle $j$ we can replace $\lambda(\sigma_j)$ with $\kappa_j$ since the whole expression has balanced homogeneity. This leaves the final soft factor
\begin{equation}
    \frac{1}{ [n\,j]\, \langle n\, j \rangle} = \frac{1}{2\, p_n \cdot p_j}\,,
\end{equation}
where $j$ is the external leg that $n$ attaches to directly in the selected tree subamplitude. Repeating this for all trees contributing to the formula, we recover the full soft structure of BAS from the anti-holomorphic soft limit.

The same analysis can be adapted to the case of multi-particle soft limits, verifying that the amplitude has no unphysical poles. This is in contrast to the scalar blocks constructed in~\cite{Cachazo:2016sdc}, for which it is easy to see that there are spurious soft singularities that are additionally dependent on the `parity' of the soft limit taken.

\bibliography{KLT}
\bibliographystyle{JHEP}
\end{document}